\documentclass[reqno, eucal]{amsart}
\usepackage{amsmath,amssymb,amsthm,amscd}
\usepackage[mathscr]{eucal}
\usepackage[dvips]{graphicx,color}
\usepackage{epic,eepic}

\usepackage{here}

\numberwithin{equation}{section}

\newtheorem{theorem}{Theorem}
\newtheorem{lemma}[theorem]{Lemma}

\newtheorem{proposition}[theorem]{Proposition}

\theoremstyle{definition}

\newtheorem{example}[theorem]{Example}
\newtheorem{remark}[theorem]{Remark}

\newcommand{\bk}{{\bf k}}
\newcommand{\bb}{{\bf b}}
\newcommand{\bc}{{\bf c}}

\newcommand{\Z}{{\mathbb Z}}
\newcommand{\R}{{\mathbb R}}
\newcommand{\C}{{\mathbb C}}

\begin{document}
\title[Zero range process]
{Density and current profiles  
in $\boldsymbol{U_q(A^{(1)}_2)}$ zero range process}
\author{A.~Kuniba}
\email{atsuo.s.kuniba@gmail.com}
\address{Institute of Physics, University of Tokyo, Komaba, Tokyo 153-8902, Japan}

\author{V.~V.~Mangazeev}
\email{vladimir.mangazeev@anu.edu.au}
\address{Department of Theoretical Physics, 
Research School of Physics and Engineering,
Australian National University, Canberra, ACT 0200, Australia.}

\maketitle

\begin{center}{\bf Abstract}\end{center}

The stochastic $R$ matrix for $U_q(A^{(1)}_n)$ introduced recently  
gives rise to an integrable zero range process of $n$ classes of particles 
in one dimension.
For $n=2$ we investigate how finitely many first class 
particles fixed as defects influence the grand canonical ensemble of 
the second class particles.
By using the matrix product stationary probabilities involving infinite products of 
$q$-bosons, exact formulas are derived for the 
local density and current of the second class particles in the large volume limit.
\vspace{0.2cm}

\section{Introduction}\label{sec1}

Zero range processes (ZRPs) \cite{S} are stochastic particle systems
on lattice modeling various flows in granules, 
queuing networks, traffic and so forth.
Their characteristic feature is that 
particles are allowed to share a site and hop over the lattice
with the rate that only depends on the occupancy 
and the list of leaving particles at the departure site\footnote{
This slightly generalizes the original terminology in that 
an arbitrary number of particles 
are allowed to jump out simultaneously.
Such multiple jumps can be suppressed by setting $\mu\rightarrow 0$
in our model. See 
the explanation after (\ref{wm}).}.
To describe their hydrodynamic limit and the rich behavior 
like condensation has been an important issue 
in non-equilibrium statistical mechanics.
See for example \cite{EH, GSS, KL} and references therein.

In the recent work \cite{KMMO}, 
a stochastic $R$ matrix for the 
quantum affine algebra $U_q(A^{(1)}_n)$ was constructed.
It gives rise to discrete and continuous time Markov processes
associated with a commuting family of Markov transfer matrices.
They are formulated as stochastic dynamics of 
$n$ classes of particles in one dimension with zero range type interaction.
Many integrable Markov processes studied earlier, 
e.g. \cite{BC, BCG, KMO2, P, SW2, T0, T1}  
can be identified with their special cases 
as summarized in \cite[Fig.1,2]{Kuan}.
In this paper we will be concerned with the version of the model
introduced in \cite[Sec.3.3, 3.4]{KMMO}, which will be called 
the $U_q(A^{(1)}_n)$ ZRP.
When $n=1$ it reduces to the zero range chipping model 
introduced in \cite{P}.

Stationary states of the $U_q(A^{(1)}_n)$ ZRP 
were obtained in \cite{KO1,KO2}.
Let $\sigma_i =(\sigma_{i,1},\ldots, \sigma_{i,n}) \in \Z_{\ge 0}^n$
be a local state, which means that 
there are $\sigma_{i,k}$ particles of class $k$ at the lattice site $i$.
For the length $L$ periodic chain,
the probability of finding the system in a given configuration
$(\sigma_1,\ldots, \sigma_L) \in (\Z_{\ge 0}^n)^L$
is expressed, up to normalization, by the matrix product formula
\begin{align}\label{nzm}
{\mathbb P}(\sigma_1,\ldots, \sigma_L)
= \mathrm{Tr}(X_{\sigma_1}\cdots X_{\sigma_L}),
\end{align}
where $X_{\sigma_i}$ is an operator acting on the tensor product 
$F^{\otimes \frac{n}{2}(n-1)}$ of the $q$-boson Fock space 
$F = \oplus_{m \ge 0}\C(q)|m\rangle$.
For $n=1$, the $X_{\sigma_i}$ 
is just a scalar meaning that the stationary measure is factorized.
However it is an exceptional feature limited to $n=1$.
For instance when $n=2$, the operator 
$X_{\alpha_1,\alpha_2}$ for the local state 
$(\alpha_1, \alpha_2) \in \Z_{\ge 0}^2$ reads  
\begin{align}\label{aoi}
X_{\alpha_1,\alpha_2}
=\frac{(\mu;q)_{\alpha_1+\alpha_2}}
{(q;q)_{\alpha_1}(q;q)_{\alpha_2}}
\frac{(\mu \bb;q)_\infty}{(\bb;q)_\infty}\bk^{\alpha_2}\bc^{\alpha_1},
\qquad
\frac{(\mu \bb;q)_\infty}{(\bb;q)_\infty} 
= \sum_{j \ge 0}
\frac{(\mu;q)_{j}}
{(q;q)_{j}}\bb^j,
\end{align}
where $\mu$ is another model parameter and 
the symbol $(z;q)_m$ is the $q$-shifted factorial 
defined in the end of this section.
The operators $\bb, \bc$ and $\bk$ are the $q$-boson creation, 
annihilation and the number operators acting on $F$ as
\begin{align*}
\bb | m \rangle &= |m+1\rangle,\qquad \bc | m \rangle = (1-q^m)|m-1\rangle,
\qquad \bk |m\rangle = q^m |m \rangle.
\end{align*}
For $n$ general the operator $X_{\alpha_1,\ldots, \alpha_n}$ 
for the local state $(\alpha_1,\ldots, \alpha_n) \in \Z_{\ge 0}^n$
possesses a nested structure with respect to the rank $n$ \cite{KO2}.
Thus the $U_q(A^{(1)}_n)$ ZRPs form the first systematic examples 
of multispecies (or multi-class) ZRPs whose stationary measure on the ring
is not factorized.
The relevant matrix product operators 
are also quite distinct from those in the exclusion type processes
(cf. \cite{DJLS,GDW,PEM}) 
in that they involve quantum dilogarithm type {\em infinite products} of $q$-bosons,
offering a challenge to extract physics of the model.

With this background in mind we present in this paper a modest analysis of   
stationary properties of the $U_q(A^{(1)}_2)$ ZRP
based on the matrix product formula (\ref{nzm})--(\ref{aoi}).
We introduce finitely many first class particles as {\em defects}
and investigate their influence on the second class particles
whose density is kept finite in the infinite volume limit.
To motivate this setting, although somewhat technically, note that one must pick  
an equal number of $\bb$'s and $\bc$'s to get 
a nonzero contribution to the trace (\ref{nzm}).
Therefore the sum of the expansion index $j$ in (\ref{aoi}) coming from 
$X_{\sigma_1}, \ldots, X_{\sigma_L}$ in (\ref{nzm}) should coincide  
with the total number of the first class particles.
Gathering all such contributions 
in the limit $L \rightarrow \infty$
is a feasible task 
at least if the first class particles are kept finite.

The second class particles will be treated in the {\em grand canonical ensemble}, 
which means that $X_{\alpha_1,\alpha_2}$ is effectively replaced by  
the generating series with respect to $\alpha_2$ in the {\em  fugacity} $y$:
\begin{align*}
A_{\alpha_1} = \sum_{\alpha_2 \ge 0}y^{\alpha_2}X_{\alpha_1,\alpha_2}
=\frac{(\mu;q)_{\alpha_1}}{(q;q)_{\alpha_1}}
\frac{(\mu \bb;q)_\infty}{(\bb;q)_\infty}
(y \bk;q)_\infty^{-1} \bc^{\alpha_1}
(\mu y \bk;q)_\infty.
\end{align*}
See (\ref{Adef}) and (\ref{akn}).
We stay in the regime $0 < q, \mu < 1$ 
where there is no symptom of condensation.
See the remarks around (\ref{ry}) concerning this point.
The equivalence with the canonical ensemble treatment will be 
argued in Section \ref{ss:ce}.
The basic quantity is the probability $P(r,m)$ 
that exactly $m$ second class particles are found at the site $r$ 
under the condition that
the sites $1,2,\ldots, s$ of the periodic lattice $\Z_L$ of size $L$
contain $d_1,\ldots, d_s$ first class particles.
To avoid the ambiguity we assume $d_1, d_s \ge 1$ but 
the choice $d_i=0$ is still allowed for $0 < i < s$.
So they form a cluster of defects of size $s$ in general.
Up to normalization the conditional probability $P(r,m)$ 
is given by replacing the $r$ th operator from the left in 
\begin{align*}
\mathrm{Tr}\bigl(A_{d_1}\cdots A_{d_s} A_0^{L-s}\bigr)
\end{align*}
by $y^mX_{d_r,m}$ if $1 \le r \le s$ and by $y^mX_{0,m}$ elsewhere\footnote{
Precise treatment involves a regularization as mentioned after (\ref{Pleft}).}.  
Our task is to evaluate it in the infinite volume limit $L \rightarrow \infty$
with $s$ and $d_1, \ldots, d_s$ kept fixed.
The limit separates the periodic lattice $\Z_L$ into the three distinct regions 
I, II and III, which are the inside ($1 \le r \le s$), the right ($r>s$), 
and the left ($r \le 0$) of the defect cluster, respectively.
The $\Z_L$ periodicity of the lattice implies that the probability
$P(r,m)$ for the right region $r>s$ and the left region $r\leq0$ 
should match when $r \to \infty$. This is indeed the case as seen 
from (\ref{rinf}) and (\ref{N3}).

Once the conditional probability is determined, the
local density and current of the second class particles 
are derived at any site $r$.
They are physical quantities {\em seen} from the defects.
The necessary calculations are elementary.
The final results are summarized 
in Theorem \ref{th:naka}, \ref{th:migi} and \ref{th:hidari} 
for the regions I, II and III, respectively.
They are expressed in terms of the $q$-digamma function and its derivative
together with the functions
$G_{m,l}(d_1,\ldots, d_s)$ (\ref{Gm}) which incorporate 
the effect of defects.
Curiously the latters are related to
the monodromy matrices of the $U_q(A^{(1)}_1)$ ZRP
containing the fugacity $y$ as a spectral parameter.
In Proposition \ref{pr:excess} we will also show that 
the defects decrease the 
second class particles in the entire system 
exactly by their number 
$d_1+\cdots + d_s$ compared from the defect-free situation.

We present the profiles of the local density and currents in a number of 
figures for various values of $q, \mu, \rho$ 
and the defect pattern $(d_1,\ldots, d_s)$,
where $\rho$ denotes the average density of the second class particles.
The density profiles exhibit a peak and a valley at the left and the right
boundaries of the defect cluster, respectively.
In the current profiles the peak is not observed but 
other behavior is more or less similar to the density. 
The detail is dependent on the pattern $(d_1,\ldots, d_s)$ and 
it is not easy to provide an intuitive explanation in general.
However our result captures the mode $\eta_j$ (\ref{etadef}) 
controlling the correlation length, 
which shows that the influence of the defects 
reaches longer distance when the average density $\rho$ is lower.
Moreover we provide especially simple profiles
in the limits $\rho \rightarrow 0$ and $\rho \rightarrow \infty$ 
in the case of homogeneous defects
in Figure \ref{sch}, which makes the general case easy to infer.

We remark that analyses of density and current 
based on the grand canonical ensemble similar to the present paper
have been done for a class of two species exclusion type processes.
See for example \cite{DJLS, RSS} and references therein.

The layout of the paper is as follows.
We recall the $U_q(A^{(1)}_n)$ ZRP \cite{KMMO} in Section \ref{sec2} and 
the matrix product formula for the stationary
probabilities for $n=2$ \cite{KO1} in Section \ref{sec3}.
The local density and current of the 
second class particles in the grand canonical ensemble are 
formulated in Section \ref{sec4}.
We first deal with the defect-free single species case  
in Section \ref{sec5} as a preparatory warm-up
partly reproducing known facts in earlier works, e.g. \cite{BC,EH,P}.
Our main results in the presence of defects are 
given in Section \ref{sec:results} with a number of 
figures showing the density and current profiles.
Their derivation are detailed in Section \ref{sec6}.
Section \ref{sec7} contains a brief summary and discussion.
Technical lemmas are collected in Appendix \ref{app:ls}.

Throughout the paper we use the notation
$\theta(\mathrm{true})=1, 
\theta(\mathrm{false}) =0$,
$(z)_m = (z; q)_m = \prod_{j=0}^{m-1}(1-zq^j)$
and the $q$-binomial 
$\binom{m}{k}_{\!q} = \theta(k \in [0,m])
\frac{(q)_m}{(q)_k(q)_{m-k}}$.
The symbols $(z)_m$ appearing in this paper always mean $(z; q)_m$.
For integer arrays 
$\alpha=(\alpha_1,\ldots, \alpha_m), \beta=(\beta_1,\ldots, \beta_m)$ 
of any length $m$, we write 
$|\alpha | = \alpha_1+\cdots  + \alpha_m$ and 
the Kronecker delta 
$\delta_{\alpha, \beta} = \delta^{\alpha}_{\beta} =
 \prod_{i=1}^m\theta(\alpha_i=\beta_i)$. 
The relation $\alpha \le \beta$ or equivalently  
$\beta \ge \alpha$ is defined by $\beta-\alpha \in \Z^m_{\ge 0}$. 

\section{$U_q(A^{(1)}_n)$ zero range process}\label{sec2}

\subsection{Stochastic $R$ matrix}
Set 
$W = \bigoplus_{\alpha=(\alpha_1,\ldots, \alpha_n) \in \Z_{\ge 0}^n}
\C |\alpha\rangle$.
Define the operator $\mathscr{S}(\lambda,\mu) 
\in \mathrm{End}(W \otimes W)$ 
depending on the parameters $\lambda$ and $\mu$ by
\begin{align}
&\mathscr{S}(\lambda,\mu)(|\alpha\rangle \otimes |\beta\rangle ) = 
\sum_{\gamma,\delta \in 
\Z_{\ge 0}^n}\mathscr{S}(\lambda,\mu)_{\alpha,\beta}^{\gamma,\delta}
\,|\gamma\rangle \otimes |\delta\rangle,
\label{ask1}\\
&\mathscr{S}(\lambda,\mu)^{\gamma,\delta}_{\alpha, \beta} 
= \theta(\gamma+\delta=\alpha+\beta)
\Phi_q(\gamma | \beta; \lambda,\mu), \label{ask2}
\end{align}
where $\Phi_q(\gamma | \beta; \lambda,\mu)$ is given by

\begin{align}
\Phi_q(\gamma|\beta; \lambda,\mu)  = 
q^{\varphi(\beta-\gamma, \gamma)}
\left(\frac{\mu}{\lambda}\right)^{|\gamma|}
\frac{(\lambda)_{|\gamma|}(\frac{\mu}{\lambda})_{|\beta|-|\gamma|}}
{(\mu)_{|\beta|}}
\prod_{i=1}^{n}\binom{\beta_i}{\gamma_i}_{\!q},
\quad
\varphi(\alpha, \beta) &= \sum_{1 \le i < j \le n}\alpha_i\beta_j. \label{mho}
\end{align}
The sum (\ref{ask1}) is finite due to the $\theta$ factor in (\ref{ask2}).
According to the direct sum decomposition
$W \otimes W = \bigoplus_{\gamma \in \Z_{\ge 0}^n}
\left(\bigoplus_{\alpha+\beta=\gamma}\C
|\alpha\rangle \otimes |\beta\rangle\right)$,
$\mathscr{S}(\lambda,\mu)$ splits into the corresponding submatrices. 
Note that 
$\mathscr{S}(\lambda,\mu)^{\gamma,\delta}_{\alpha, \beta}=0$ 
unless $\gamma \le \beta$ hence $\alpha \le \delta$ as well.
The difference property 
$\mathscr{S}(\lambda,\mu)=\mathscr{S}(c\lambda,c\mu)$ is absent.
We call $\mathscr{S}(\lambda, \mu)$ 
the {\em stochastic $R$ matrix} and depict its elements as
\begin{equation}\label{vertex}
\begin{picture}(200,45)(-90,-21)
\put(-70,-2){$\mathscr{S}(\lambda,\mu)^{\gamma,\delta}_{\alpha, \beta}\,=$}
\put(0,0){\vector(1,0){24}}
\put(12,-12){\vector(0,1){24}}
\put(-10,-2){$\alpha$}\put(27,-2){$\gamma$}
\put(9,-22){$\beta$}\put(9,16){$\delta$}
\end{picture}
\end{equation}
The stochastic $R$ matrix was constructed \cite{KMMO} based on the 
quantum $R$ matrix of the symmetric tensor representation of 
the quantum affine algebra $U_q(A^{(1)}_n)$.
It satisfies the Yang-Baxter equation 
and the sum-to-unity condition \cite{KMMO}:
\begin{align}
&\mathscr{S}_{1,2}(\nu_1,\nu_2)
\mathscr{S}_{1,3}(\nu_1, \nu_3)
\mathscr{S}_{2,3}(\nu_2, \nu_3)
=
\mathscr{S}_{2,3}(\nu_2, \nu_3)
\mathscr{S}_{1,3}(\nu_1, \nu_3)
\mathscr{S}_{1,2}(\nu_1,\nu_2),
\label{sybe}\\
&\sum_{\gamma, \delta \in \Z_{\ge 0}^n}
\mathscr{S}(\lambda,\mu)^{\gamma,\delta}_{\alpha, \beta} =1\qquad
(\forall \alpha, \beta \in \Z_{\ge 0}^n).
\nonumber
\end{align}
The latter is a consequence of the sum rule \cite{KMMO}:
\begin{align}\label{syk}
\sum_{\gamma \in \Z_{\ge 0}^n}
\Phi_q(\gamma | \beta; \lambda,\mu)
= 1 \quad(\forall \beta \in \Z_{\ge 0}^n),
\end{align}
where the summand is zero unless $\gamma \le \beta$.

\subsection{\mathversion{bold}Discrete time 
$U_q(A^{(1)}_n)$ ZRP}\label{ss:sae}
For a positive integer $L$ we introduce the operator
\begin{align}\label{ngm}
T(\lambda|\mu_1,\ldots, \mu_L) = 
\mathrm{Tr}_{W}\left(
\mathscr{S}_{0,L}(\lambda,\mu_L)\cdots \mathscr{S}_{0,1}(\lambda,\mu_1)
\right)
\in \mathrm{End}(W^{\otimes L})
\end{align}
depending on the parameters 
$\mu_1, \ldots, \mu_L$ and $\lambda$.
To explain the notation, 
consider the space  $W\otimes W^{\otimes L}$ labeled as
$W_0\otimes W_1 \otimes \cdots \otimes W_L$ for distinction ($W_i = W$).
Then $\mathscr{S}_{0,i}(\lambda, \mu_i)$ acts as 
the stochastic $R$ matrix $\mathscr{S}(\lambda, \mu_i)$ (\ref{ask1}) 
on $W_0 \otimes W_i$ and as the identity elsewhere.
The trace in (\ref{ngm}) is taken over $W_0$ leaving an 
operator acting on $W_1 \otimes \cdots \otimes W_L = W^{\otimes L}$.

Write the action of 
$T=T(\lambda|\mu_1,\ldots, \mu_L) $ as
\begin{align*}
T(|\beta_1\rangle \otimes \cdots \otimes |\beta_L\rangle) 
= \sum_{\alpha_1,\ldots, \alpha_L \in \Z_{\ge 0}^n}
 T_{\beta_1,\ldots, \beta_L}^{\alpha_1,\ldots, \alpha_L}
|\alpha_1\rangle \otimes \cdots \otimes |\alpha_L\rangle 
\in W^{\otimes L}.
\end{align*}
Then the matrix element is depicted by the concatenation of (\ref{vertex}) as 
\begin{equation}\label{tdiag}
\begin{picture}(250,50)(10,-25)
\put(-20,0){$T_{\beta_1,\ldots, \beta_L}^{\alpha_1,\ldots, \alpha_L}=
{\displaystyle \sum_{\gamma_1,\ldots, \gamma_L \in \Z_{\ge 0}^n}}$}

\put(100,0){
\put(0,0){\vector(1,0){24}}
\put(12,-12){\vector(0,1){24}}
\put(-12,-2){$\gamma_L$}\put(28,-2){$\gamma_1$}
\put(9,-22){$\beta_1$}\put(8,16){$\alpha_1$}}

\put(140,0){
\put(0,0){\vector(1,0){24}}
\put(12,-12){\vector(0,1){24}}
\put(27,-2){$\gamma_2$}
\put(9,-22){$\beta_2$}\put(8,16){$\alpha_2$}}

\put(182,-3){$\cdots$}

\put(220,0){
\put(0,0){\vector(1,0){24}}
\put(12,-12){\vector(0,1){24}}
\put(-24,-2){$\gamma_{L-1}$}\put(27,-2){$\gamma_L$,}
\put(9,-22){$\beta_L$}\put(8,16){$\alpha_L$}}
\end{picture}
\end{equation}
which is a customary diagram for row transfer matrices of vertex models 
on the length $L$ periodic lattice \cite{Bax}.
By the construction it satisfies the weight conservation:
\begin{align}\label{wc}
T_{\beta_1,\ldots, \beta_L}^{\alpha_1,\ldots, \alpha_L} = 0
\;\;\text{unless}\;\;
\alpha_1+\cdots +\alpha_L = 
\beta_1+\cdots + \beta_L \in \Z_{\ge 0}^{n}.
\end{align}
Thanks to the Yang-Baxter equation (\ref{sybe}), 
the matrix (\ref{ngm}) forms a commuting family (cf. \cite{Bax}):
\begin{align}\label{mri}
[T(\lambda|\mu_1,\ldots, \mu_L), 
T(\lambda'|\mu_1,\ldots, \mu_L)]=0.
\end{align}

Let $t$ be a time variable and consider the evolution equation 
\begin{align}\label{dmt}
|P(t+1)\rangle = T(\lambda|\mu_1,\ldots, \mu_L)
|P(t)\rangle \in W^{\otimes L}.
\end{align}
Due to the weight conservation (\ref{wc}) it splits into 
finite-dimensional subspaces which we call {\em sectors}.
In terms of the array $m=(m_1,\ldots, m_n) \in \Z^n_{\ge 0}$ 
and the set 
\begin{align*}
S(m) = \{(\sigma_1,\ldots, \sigma_L) \in (\Z^n_{\ge 0})^L\mid 
\sigma_1+\cdots + \sigma_L = m\},
\end{align*}
the corresponding sector, which will also be referred to as $m$, is given by 
$\oplus_{(\sigma_1,\ldots, \sigma_L)\in S(m)}
\C |\sigma_1\rangle \otimes \cdots \otimes  |\sigma_L\rangle$.
The vector 
$|\sigma_1\rangle \otimes \cdots \otimes  |\sigma_L\rangle \in W^{\otimes L}$ 
with $\sigma_i=(\sigma_{i,1},\ldots, \sigma_{i,n}) \in \Z_{\ge 0}^n$ 
is interpreted as a state of the system where the $i$ th site from the left 
is populated with $\sigma_{i,a}$ particles of the $a$ th class or species.  
The array $m=(m_1,\ldots, m_n)$ indicates that 
there are $m_a$ particles of class $a$ in total in the corresponding sector.

In order to interpret (\ref{dmt}) as the master 
equation of a discrete time Markov process, the matrix 
$T=T(\lambda|\mu_1,\ldots, \mu_L) $ should fulfill the following conditions:

\vspace{0.2cm}
\begin{enumerate}
\item  Non-negativity; all the elements (\ref{tdiag}) belong to $\R_{\ge 0}$,

\item Sum-to-unity; $\sum_{\alpha_1,\ldots, \alpha_L\in \Z_{\ge 0}^{n}}
T_{\beta_1,\ldots, \beta_L}^{\alpha_1,\ldots, \alpha_L} = 1$ for any
$(\beta_1,\ldots, \beta_L) \in (\Z_{\ge 0}^n)^L$.
\end{enumerate}

\vspace{0.2cm}\noindent
The property  (i)  holds if 
$\Phi_q(\gamma|\beta; \lambda,\mu_i)\ge 0$
for all $i \in \Z_L$.
This is achieved 
by taking $0 < \mu^{\epsilon}_i < \lambda ^{\epsilon} < 1, q^{\epsilon}<1$
for $\epsilon=\pm 1$.
The property (ii) assures the total probability conservation and can be 
shown by using (\ref{syk}) in a similar manner to \cite[Sec.3.2]{KMMO}.

Henceforth we call the $T(\lambda|\mu_1,\ldots, \mu_L) $
{\em Markov transfer matrix} assuming 
$0 < \mu^{\epsilon}_i < \lambda ^{\epsilon} < 1, q^{\epsilon}<1$ always. 
The choice of $\epsilon=\pm 1$ specifies one 
of the two physical regimes of the system.
The evolution equation (\ref{dmt}) 
describes a stochastic dynamics of $n$ classes of particles
hopping to the right periodically 
via an extra lane (horizontal arrows in (\ref{tdiag}))
which particles get on or get off when they leave or arrive at a site.
The rate of such local processes is specified by 
(\ref{ask2}), (\ref{mho}) and (\ref{vertex}).
For $n=1$ and the homogeneous choice $\mu_1=\cdots= \mu_L$, it reduces to 
the model introduced in \cite{P}.

\subsection{\mathversion{bold}Continuous time 
$U_q(A^{(1)}_n)$ ZRP\mathversion{bold}}

One can derive continuous time versions of (\ref{dmt}) 
from the homogeneous case $\mu_1=\cdots = \mu_L=\mu$
by taking the logarithmic derivative either at $\lambda=1$ or $\lambda=\mu$ 
\cite[Sec.3.4]{KMMO}.
The result is given by
\begin{align}\label{skb}
\frac{d}{dt}|P(t)\rangle &= H |P(t)\rangle \in W^{\otimes L},\quad
H = aH_+ + bH_- \; \;(a, b \in \R_{\ge 0}),\\
H_+ &= \left.-\epsilon\mu^{-1}
\frac{\partial \log 
T(\lambda| \mu, \ldots, \mu)}{\partial \lambda}\right|_{\lambda=1},
\qquad
H_- = \left. \epsilon\mu\,
\frac{\partial \log T(\lambda| \mu, \ldots, \mu)}{\partial \lambda} 
\right|_{\lambda = \mu}
\label{Hs}
\end{align}
with
$H_\pm = \sum_{i\in \Z_L} h_{\pm ,i,i+1}$. 
The summands 
$h_{\pm, i,i+1}$ act on the adjacent  
$(i, i\!+\!1)$ th components of $W^{\otimes L}$ as $h_\pm$ and 
as the identity elsewhere.
The pairwise interactions $h_\pm$ are specified by
\begin{align}
h_+(|\alpha\rangle \otimes |\beta\rangle)
&=\epsilon\!\!\!\sum_{\gamma \in \Z^n_{\ge 0}\setminus \{0^n\}}  
\!\!\!\frac{q^{\varphi(\alpha-\gamma,\gamma)}
\mu^{|\gamma|-1}(q)_{|\gamma|-1}}
{(\mu q^{|\alpha|-|\gamma|})_{|\gamma|}}
\prod_{i=1}^n
\binom{\alpha_i}{\gamma_i}_{\!q}|\alpha-\gamma\rangle 
\otimes |\beta+\gamma\rangle
\nonumber\\ 
&-\epsilon\sum_{i=0}^{|\alpha|-1}\frac{q^i}{1-\mu q^i}
|\alpha\rangle \otimes |\beta\rangle,
\label{ri1}\\
h_-(|\alpha\rangle \otimes |\beta\rangle)
&= \epsilon\!\!\!\sum_{\gamma \in \Z^n_{\ge 0}\setminus \{0^n\}}
\!\!\!\frac{q^{\varphi(\gamma, \beta-\gamma)}
(q)_{|\gamma|-1}}
{(\mu q^{|\beta|-|\gamma|})_{|\gamma|}}
\prod_{i=1}^n
\binom{\beta_i}{\gamma_i}_{\!q}
|\alpha+\gamma\rangle \otimes |\beta-\gamma\rangle
\nonumber\\
&-\epsilon\sum_{i=0}^{|\beta|-1}\frac{1}{1-\mu q^i}
|\alpha\rangle \otimes |\beta\rangle,
\label{ri2}
\end{align}
where $\varphi(\alpha, \beta)$ is defined in (\ref{mho}) and 
$0^n$ stands for $(0,\ldots, 0)\in \Z_{\ge 0}^n$.
We have included the sign factor $\epsilon=\pm 1$ to cover the two regimes.
Denote the action of the matrices ${\mathcal H}= H_+, H_-$
on the base vectors as ${\mathcal H}
(|\beta_1\rangle \otimes \cdots \otimes |\beta_L \rangle)
= \sum_{\alpha_1,\ldots, \alpha_L}
{\mathcal H}^{\alpha_1, \ldots, \alpha_L}_{\beta_1,\ldots, \beta_L}
|\alpha_1\rangle \otimes \cdots \otimes |\alpha_L\rangle$.
The equation (\ref{skb}) can be viewed as the master equation 
of a continuous time Markov process if the following conditions are satisfied: 

\vspace{0.1cm}
(i)' Non-negativity; 
${\mathcal H}^{\alpha_1, \ldots, \alpha_L}_{\beta_1,\ldots, \beta_L} 
\in \R_{\ge 0}$ for any pair such that 
$(\alpha_1, \ldots, \alpha_L) \neq (\beta_1,\ldots, \beta_L)$,

\vspace{0.1cm}
(ii)' Sum-to-zero;
$\sum_{\alpha_1, \ldots, \alpha_L}
{\mathcal H}^{\alpha_1, \ldots, \alpha_L}_{\beta_1,\ldots, \beta_L} = 0$ 
for any $(\beta_1,\ldots, \beta_L)$.

\vspace{0.1cm}\noindent
The latter represents the total probability conservation.
It was shown in \cite{KMMO} that 
(i)' and (ii)' are satisfied if 
$0 < q^\epsilon, \mu^\epsilon <1$
for $\epsilon = \pm 1$.

The commutativity (\ref{mri}) leads to $[H_+, H_-]=0$.
Therefore they share the same eigenvectors with the superposition
$H= H(a,b,\epsilon,q,\mu) =a H_+(\epsilon,q,\mu)+ b H_-(\epsilon,q,\mu)$
in (\ref{skb}).
A curious symmetry  
$H(a,b,-\epsilon, q^{-1}, \mu^{-1})
= \mathscr{P}H(\mu b, \mu a,\epsilon,q,\mu)\mathscr{P}^{-1}$
is known to hold \cite[Rem.9]{KMMO}, where 
$\mathscr{P} = \mathscr{P}^{-1} 
\in \mathrm{End}(W^{\otimes L})$ is the `parity' operator 
reversing the sites as
$\mathscr{P}
(|\alpha_1\rangle \otimes \cdots \otimes |\alpha_L\rangle)
= |\alpha_L\rangle \otimes \cdots \otimes |\alpha_1\rangle$.
 
The Markov processes (\ref{skb}) 
is naturally interpreted as the stochastic dynamics of 
$n$ classes of particles on the ring of length $L$.
The base vector 
$|\alpha_1\rangle \otimes \cdots \otimes |\alpha_L\rangle$ 
with $\alpha_i = (\alpha_{i,1}\ldots, \alpha_{i,n}) \in \Z^n_{\ge 0}$ 
represents a state in which 
there are $\alpha_{i,a}$ class $a$ particles at the $i$ th site.
There is no constraint on the number of particles that occupy a site.
The matrices $H_+$ and $H_-$ describe their stochastic hopping
to the right and the left nearest neighbor sites, respectively.
The transition rate can be read off the first terms on the RHS of 
(\ref{ri1}) and (\ref{ri2}), where the array 
$\gamma=(\gamma_1,\ldots, \gamma_n)$ specifies the numbers of 
particles that are jumping out.
The superposition $H(a,b,\epsilon,q,\mu)$ corresponds to 
a {\em mixture} of such right and left moving dynamics.
The rate is determined from the original occupancy 
($\alpha$ for $h_+$ and $\beta$ for $h_-$) and  
the list of leaving particles ($\gamma$ for $h_\pm$) 
at the departure site and it is independent of the status of the destination site.
Thus it defines a ZRP of $n$ classes of particles 
in a slightly generalized sense in that the rate is allowed to 
depend on $\gamma$.
Here is a snapshot of the system for the $n=2$ case.

\setlength{\unitlength}{1mm}
\begin{picture}(100,37)(-5,31)
\thicklines
\put(50,50){\ellipse{80}{20}}

\put(12,51){\line(0,1){3.8}} 
\put(20,54.4){\line(0,1){3.8}}
  \put(22,58){$\bullet$}\put(24,58.5){$\bullet$}\put(26,58.8){$\circ$}

\put(30,56.4){\line(0,1){3.8}}
\put(40,57.4){\line(0,1){3.8}}
   \put(42.5,60.7){$\circ\,\circ$}

\put(50,58){\line(0,1){3.8}}
  \put(52.3,60.5){$\bullet$} \put(54.5,60.5){$\circ$}\put(56.5,60.5){$\circ$}
  
\put(60,57.6){\line(0,1){3.8}}

\put(70,56.5){\line(0,1){3.8}}
  \put(71.6,58.7){$\circ$}\put(73.6,58.5){$\circ$}\put(75.6,58.1){$\bullet$}

\put(80,54.7){\line(0,1){3.8}}
  \put(82.2,56.5){$\bullet$}\put(84.2,55.7){$\circ$}

\put(87,51.5){\line(0,1){3.8}}

\put(12,45){\line(0,1){4}}
  \put(14.2,46){$\bullet$}\put(16,45.3){$\bullet$}

\put(20,41){\line(0,1){4}}
\put(30,39){\line(0,1){4}}
  \put(34,41.5){$\circ$}

\put(35,36.1){$\cdot$}
\put(36.5,35.9){$\cdot$}
\put(38,35.8){$\cdot$}
\put(39.5,35.7){$\cdot$}
\put(41,35.6){$\cdot$}
\put(42.5,35.6){$\cdot$}

\put(40,38){\line(0,1){4}}
\put(-0.5,0){
 \put(42,41){$\circ \circ$}\put(42,43){$\circ \circ$}\put(43,45){$\circ$}
 \put(45.8,41){$\bullet \bullet$}\put(46.5,43){$\bullet$}}

\put(50,38){\line(0,1){4}}
  \put(54,41){$\bullet$}

\put(60,38){\line(0,1){4}}
\put(-0.8,0){
  \put(62,41.4){$\circ$}\put(64,41.6){$\circ$}
  \put(62.6,43.5){$\circ$}
  \put(66,41.8){$\bullet$}\put(68,42){$\bullet$}}

\put(70,39){\line(0,1){4}}
  \put(72.5,42.6){$\circ$}\put(74.5,43){$\bullet$}\put(76.7,43.4){$\bullet$}

\put(80,41){\line(0,1){4}}
\put(87,44){\line(0,1){4}}

\put(14,41){$\sigma_L$}
\put(23,38){$\sigma_1$}
\put(52.5,35){$\sigma_{i-1}$}
\put(63.7,35.5){$\sigma_i$}
\put(73,37){$\sigma_{i+1}$}

\put(83.5,38.5){$\cdot$}
\put(85,39.1){$\cdot$}
\put(86.3,39.7){$\cdot$}
\put(87.6,40.7){$\cdot$}
\put(88.6,41.7){$\cdot$}

\put(59.5,46.5){\oval(7,7)[t]}\put(56.1,45.5){\vector(0,-1){1}}
\put(52.3,51.9){$h_{-,i-1,i}$}
\put(68.5,47){\oval(7,7)[t]}\put(72.2,46.5){\vector(0,-1){1}}
\put(66,52.4){$h_{+,i,i+1}$}

\put(0,-6){
\put(100,66){$\circ\,$ first class particle}
\put(100,61){$\bullet\,$ second class particle}

\put(104,52){$\sigma_{i-1}=(0,1)$,}
\put(107.5,48){$\sigma_i=(3,2)$,}
\put(104,44){$\sigma_{i+1}=(1,2)$.}
}

\end{picture}
\setlength{\unitlength}{1pt}

\noindent
The local hopping in (\ref{ri1}) and (\ref{ri2})
with $n=2$ is depicted as  

\setlength{\unitlength}{0.7mm}
\begin{picture}(200,40)(-3,32)
\put(3,55){rate}
\put(0,48){$\displaystyle{a w_+(\gamma|\alpha)}:$
\put(10,-10){
\drawline(0,5)(0,0)
\drawline(0,0)(50,0)\drawline(50,0)(50,5)
\drawline(25,0)(25,5)
\put(-1,0){
\drawline(12,12)(12,17)\drawline(12,17)(37,17)
\put(37,17){\vector(0,-1){5}}
}
\put(3,2){$\overbrace{\circ ... \circ}^{\alpha_1}
\overbrace{\bullet...\bullet}^{\alpha_2}$}
\put(28,2){$\circ ... \circ \bullet...\bullet$}
\put(14, 19){$\overbrace{\circ ... \circ}^{\gamma_1}
\overbrace{\bullet...\bullet}^{\gamma_2}$}}

}
\put(75,38){
\put(33,17){rate}
\put(30,10){$\displaystyle{b w_-(\gamma|\beta)}:$

\put(7,-10){
\drawline(0,5)(0,0)
\drawline(0,0)(50,0)\drawline(50,0)(50,5)
\drawline(25,0)(25,5)
%
\put(-1,0){
\drawline(37,12)(37,17)\drawline(37,17)(12,17)
\put(12,17){\vector(0,-1){5}}
}
\put(28,2){$\overbrace{\circ ... \circ}^{\beta_1}
\overbrace{\bullet...\bullet}^{\beta_2}$}
\put(3,2){$\circ ... \circ \bullet...\bullet$}
\put(14, 19){$\overbrace{\circ ... \circ}^{\gamma_1}
\overbrace{\bullet...\bullet}^{\gamma_2}$}}
}

}
\end{picture}
\setlength{\unitlength}{1pt}

\noindent
where the rate $w_+(\gamma|\alpha)$ and $w_-(\gamma|\beta)$ 
with $\epsilon=+1$ are given for $|\gamma | >0$ by 
\begin{align}
w_+(\gamma|\alpha) 
&=\frac{q^{(\alpha_1-\gamma_1)\gamma_2}\mu^{\gamma_1+\gamma_2-1}
(q)_{\gamma_1+\gamma_2-1}}
{(\mu q^{\alpha_1+\alpha_2-\gamma_1-\gamma_2})_{\gamma_1+\gamma_2}}
\frac{(q)_{\alpha_1}}{(q)_{\gamma_1}(q)_{\alpha_1-\gamma_1}}
\frac{(q)_{\alpha_2}}{(q)_{\gamma_2}(q)_{\alpha_2-\gamma_2}},
\label{wp}\\
w_-(\gamma|\beta) 
&=\frac{q^{\gamma_1(\beta_2-\gamma_2)}
(q)_{\gamma_1+\gamma_2-1}}
{(\mu q^{\beta_1+\beta_2-\gamma_1-\gamma_2})_{\gamma_1+\gamma_2}}
\frac{(q)_{\beta_1}}{(q)_{\gamma_1}(q)_{\beta_1-\gamma_1}}
\frac{(q)_{\beta_2}}{(q)_{\gamma_2}(q)_{\beta_2-\gamma_2}}.
\label{wm}
\end{align}

The integrable Markov processes recalled here 
cover several models studied earlier.
When $\epsilon=1, \mu\rightarrow 0$ in $H_+$, 
the multiple jumps $|\gamma |>1$ are suppressed in (\ref{ri1}).
So if $\gamma_a=1$ and the other components of $\gamma$ are 0,
the rate reduces to   
$q^{\alpha_1+\cdots + \alpha_{a-1}}\frac{1-q^{\alpha_a}}{1-q}$.
This reproduces the $n$-species $q$-boson process 
in \cite{T1} whose $n=1$ case further goes back to \cite{SW2}.
For $n=1$, 
there are numerous works 
including \cite{BCG,BC,P,T0} for example.
One can overview their interrelation in \cite[Fig.1,2]{Kuan}.
When $\epsilon=1, (\mu,q)\rightarrow (0,0)$ in $H_-$,
a kinematic constraint 
$\varphi(\gamma, \beta-\gamma) = 
\sum_{1 \le i<j \le n}\gamma_i(\beta_j-\gamma_j)=0$ 
occurs in (\ref{ri2}).
In order that $\gamma_a>0$ happens,
the equalities
$\gamma_{a+1}=\beta_{a+1}, 
\gamma_{a+2}=\beta_{a+2}, \ldots, \gamma_n = \beta_n$ must hold.
It means that larger class particles have the {\em priority} to jump out,
which precisely reproduces the $n$ class totally asymmetric zero range process  
explored in \cite{KMO2} after reversing the labeling of the 
classes $1,2, \ldots, n$ of the particles.

\section{Stationary states}\label{sec3}

\subsection{Definition and example}
By definition a stationary state of the discrete time 
$U_q(A^{(1)}_n)$ ZRP (\ref{dmt})
is a vector $|\overline{P}\rangle \in W^{\otimes L}$ 
such that
\begin{align*}
|\overline{P}\rangle= T(\lambda|\mu_1,\ldots, \mu_L)|\overline{P}\rangle.
\end{align*}
The stationary state is unique within each sector $m$ up to normalization.
Apart from $m$, 
it depends on $q$ and the inhomogeneity parameters $\mu_1, \ldots, \mu_L$ but
{\em not} on $\lambda$ thanks to the commutativity (\ref{mri}).
Sectors $m=(m_1,\ldots, m_n)$ such that $\forall m_a \ge 1$ are called {\em basic}.
Non-basic sectors are equivalent to a basic sector of some $n'<n$ models
with a suitable relabeling of the classes.
Henceforth we concentrate on the basic sectors. 
The coefficient appearing in the expansion
\begin{align}\label{hrk}
|\overline{P}\rangle = \sum_{(\sigma_1,\ldots, \sigma_L) \in S(m)}
{\mathbb P}(\sigma_1,\ldots, \sigma_L)
|\sigma_1,\ldots, \sigma_L\rangle
\end{align}
is the stationary probability if it is normalized as
$\sum_{(\sigma_1,\ldots, \sigma_L) \in S(m)}
{\mathbb P}(\sigma_1,\ldots, \sigma_L) = 1$.
In this paper we will abuse the terminology also for 
the unnormalized states and probabilities.

The stationary states for the continuous Markov process (\ref{skb}) are those 
$|\overline{P}\rangle$ that satisfy 
$H |\overline{P}\rangle = 0$.
They are obtained from the discrete time ones just by the
specialization to the homogeneous case
$\mu_1=\cdots = \mu_L = \mu$.
This is because (\ref{skb}) is an infinitesimal version of the 
commuting time evolutions (\ref{dmt}) by the construction. 
In particular the stationary states are independent of $a$ and $b$ in (\ref{skb}).

\begin{example}\label{ex:tgm}
Set $(n,L) = (2,3)$ and consider  
the homogeneous case $\mu_1=\mu_2=\mu_3 = -\kappa$.
The stationary state in the sector $m=(1,2)$ is given by
\begin{align*}
|\overline{P}\rangle = &
3 (1 + q \kappa) (1 + q^2 \kappa) |\emptyset, \emptyset, 122\rangle 
+ (1 + q) (2 + 
    q) (1 + \kappa) (1 + q \kappa) |\emptyset, 2, 12\rangle \\
&+ (2 + q^2) (1 + 
    \kappa) (1 + q \kappa) |\emptyset, 22, 1\rangle + (1 + 2 q^2) (1 + 
    \kappa) (1 + q \kappa) |\emptyset, 1, 22\rangle \\
&+ (1 + q) (1 + 2 q) (1 + 
    \kappa) (1 + q \kappa) |\emptyset, 12, 2\rangle + (1 + q) (1 + q + 
    q^2) (1 + \kappa)^2 |2, 2, 1\rangle+ \text{cyclic},
\end{align*}
where $ |\emptyset, 2, 12\rangle$ for example is 
the multiset representation of the 
base vector
$ |(0,0)\rangle \otimes |(0,1)\rangle \otimes |(1,1)\rangle$ 
 in the multiplicity representation.
The terms 
``cyclic" are those obtained by 
the shift $|\sigma_1, \sigma_2, \sigma_3\rangle 
\rightarrow |\sigma_{i+1}, \sigma_{i+2}, \sigma_{i+3}\rangle$
for $i \in \Z_3\setminus \{0\}$.
Similarly the stationary state in the sector $m=(2,1)$ is given by
\begin{equation}\label{knh}
\begin{split}
|\overline{P}\rangle = &3 (1 + q \kappa) (2 + q + \kappa + 2 q \kappa) (1 + 
    q^2 \kappa) |\emptyset, \emptyset, 112\rangle \\
&+ (1 + q) (1 + \kappa) (1 + 
    q \kappa) (3 + 3 q + 3 q^2 + 2 \kappa + 2 q \kappa + 
    5 q^2 \kappa) |\emptyset, 1, 12\rangle \\
&+ (1 + \kappa) (1 + q \kappa) (3 + 3 q + 
    3 q^2 + \kappa + 5 q \kappa + 2 q^2 \kappa + 
    q^3 \kappa) |\emptyset, 2, 11\rangle \\
&+ (1 + q) (1 + \kappa) (1 + 
    q \kappa) (5 + 2 q + 2 q^2 + 3 \kappa + 3 q \kappa + 
    3 q^2 \kappa) |\emptyset, 12, 1\rangle \\
&+ (1 + \kappa) (1 + q \kappa) (1 + 
    2 q + 5 q^2 + q^3 + 3 q \kappa + 3 q^2 \kappa + 
    3 q^3 \kappa) |\emptyset, 11, 2\rangle \\
&+ (1 + q) (1 + q + q^2) (1 + 
    \kappa)^2 (2 + q + \kappa + 2 q \kappa) |2, 1, 1\rangle + \text{cyclic}.
 \end{split}
 \end{equation}
It has been conjectured \cite[Ex.4]{KO2} that for any sector $m$
there is a normalization such that 
$\mathbb{P}(\sigma_1,\ldots, \sigma_L) \in 
\Z_{\ge 0}[q,-\mu_1,\ldots, -\mu_L]$ for all 
$(\sigma_1,\ldots, \sigma_L) \in S(m)$.
\end{example}

\subsection{Matrix product formula}
In \cite{KO2} a matrix product formula for the stationary probability 
was obtained for general $n$ and inhomogeneity
$\mu_1, \ldots, \mu_L$.
In the rest of the paper we shall 
exclusively deal with 
the $n=2$ case \cite{KO1} with the homogeneous choice 
$\mu_1=\cdots = \mu_L = \mu$ 
in the regime $0 < q, \mu < 1$.
In the continuous time setting, 
it corresponds to $\epsilon=+1$ in (\ref{Hs})--(\ref{ri2}).

To describe the result we need the $q$-boson algebra 
$\mathcal{B}$
generated by $\bb,\bc, \bk$ obeying the relations 
\begin{align}\label{akn}
\bk \bb = q\bb \bk,\qquad \bk\bc = q^{-1} \bc \bk,\qquad
\bb \bc = 1 - \bk,\qquad \bc \bb = 1-q\bk.
\end{align}

Let 
$F = \bigoplus_{m \ge 0}\C(q) |m\rangle$ be the Fock space and 
$F^\ast = \bigoplus_{m \ge 0}\C(q) \langle m |$ be its dual on which 
the $q$-boson operators $\bb, \bc, \bk$ act as
\begin{equation}\label{qb}
\begin{split}
\bb | m \rangle &= |m+1\rangle,\qquad \bc | m \rangle = (1-q^m)|m-1\rangle,
\qquad \bk |m\rangle = q^m |m \rangle,\\
\langle m | \bc &= \langle m+1 |,\qquad
\langle m | \bb = \langle m-1|(1-q^m),\qquad
\langle m | \bk = \langle m | q^m,
\end{split}
\end{equation}
where $|\!-\!1\rangle = \langle -1 |=0$.
They satisfy the defining relations (\ref{akn}).
The bilinear pairing of $F^\ast$ and $F$ is specified as 
$\langle m | m'\rangle = \delta_{m,m'}(q)_m$.
Then $\langle m| (Q|m'\rangle) = (\langle m|Q)|m'\rangle$ is valid and the 
trace is given by $\mathrm{Tr}\,Q= \sum_{m \ge 0}
\frac{\langle m|Q|m\rangle}{(q)_m}$.
As a vector space, the $q$-boson algebra $\mathcal{B}$ 
has the direct sum decomposition
\begin{align}\label{Bd}
\mathcal{B} &=  
\bigoplus_{r \in \Z_{\ge 0}, s \in \Z}\mathcal{B}^r_s,
\qquad\mathcal{B}^r_s =
\C(q)\bb^{\max(s,0)}\bk^r\bc^{\max(-s,0)}.
\end{align}
The trace $\mathrm{Tr}\,Q$ is finite and nonzero only if 
$Q \in \bigoplus_{r \ge 1}\mathcal{B}^r_0$ when it is 
evaluated by $\mathrm{Tr}(\bk^r) = (1-q^r)^{-1}$.
For $Q \in \mathcal{B}$, we let $Q'$ denote the projection of $Q$ 
onto $\bigoplus_{r \ge 1}\mathcal{B}^r_0$.

We introduce the operator 
depending on $\mu, q$ which acts 
on (a completion of) $F^\ast$ and $F$: 
\begin{align}\label{cie}
X_{\alpha}
=\frac{(\mu)_{\alpha_1+\alpha_2}}
{(q)_{\alpha_1}(q)_{\alpha_2}}
\frac{(\mu \bb)_\infty}{(\bb)_\infty}\bk^{\alpha_2} 
\bc^{\alpha_1}\quad \text{for }\;
\alpha = (\alpha_1, \alpha_2) \in \Z_{\ge 0}^2.
\end{align}
We have written $X_{(\alpha_1, \alpha_2)}$ as 
$X_{\alpha_1, \alpha_2}$ for simplicity.
The ratio of the infinite product of operators here is to be understood via the 
series expansion 
\begin{align}\label{air}
\frac{(zw)_\infty}{(z)_\infty}
= \sum_{j \ge 0}\frac{(w)_j}{(q)_j}z^j.
\end{align}
\begin{theorem}\label{th:mp}
For any sector $m=(m_1,m_2)$ with $m_2\ge 1$,  
the stationary probability in (\ref{hrk}) 
of the continuous time $U_q(A^{(1)}_2)$ ZRP (\ref{skb}) or 
the discrete time $U_q(A^{(1)}_2)$ ZRP (\ref{dmt}) with 
homogeneous parameters $\mu_1=\cdots = \mu_L = \mu$
is expressed in the matrix product form
\begin{align}\label{sin}
{\mathbb P}(\sigma_1,\ldots, \sigma_L)
= \mathrm{Tr}(X_{\sigma_1}\cdots X_{\sigma_L}).
\end{align}
\end{theorem}
This result was obtained in \cite[eq.(42)]{KO1}.
The formula (\ref{cie}) depends on $\alpha_1$ and $\alpha_2$ 
quite differently, indicating distinct features 
between the two classes of particles.
Interestingly this is a reflection of a tiny asymmetry 
of the hopping rate (\ref{wp}) and (\ref{wm})
under the interchange 
$(\alpha_1, \beta_1, \gamma_1) \leftrightarrow 
(\alpha_2, \beta_2, \gamma_2)$. 

\begin{example}
Consider the first two terms in (\ref{knh}).
\begin{align*}
{\mathbb P}(\emptyset, \emptyset, 112) &= 
\frac{(\mu)_3}{(q)_1(q)_2}\mathrm{Tr}\Bigl(
\frac{(\mu \bb)^3_\infty}{(\bb)^3_\infty}\bk\bc^2\Bigr)
= 3\frac{(\mu)_3}{(q)_1(q)_2}
\Bigl(\frac{(\mu)_1^2}{(q)_1^2} +
\frac{(\mu)_2}{(q)_2}\Bigr) 
\mathrm{Tr}(\bb^2 \bk\bc^2),\\
{\mathbb P}(\emptyset, 1,12) &= 
\frac{(\mu)_1(\mu)_2}{(q)_1^3}\mathrm{Tr}\Bigl(
\frac{(\mu \bb)_\infty^2}{(\bb)_\infty^2} \bc 
\frac{(\mu \bb)_\infty}{(\bb)_\infty} \bk \bc \Bigr)\\
&= \frac{(\mu)_1(\mu)_2}{(q)_1^3}\left(
\Bigl(2\frac{(\mu)_2}{(q)_2} +\frac{(\mu)_1^2}{(q)_1^2}\Bigr)
\mathrm{Tr}(\bb^2 \bc \bk \bc)+
2\frac{(\mu)_1^2}{(q)_1^2}\mathrm{Tr}(\bb \bc \bb \bk \bc)
+\frac{(\mu)_2}{(q)_2}\mathrm{Tr}(\bc \bb^2 \bk \bc)\right). 
\end{align*}
Setting $\mu = -\kappa$ and substituting 
$\mathrm{Tr}(\bb^2 \bk \bc^2) = q^{-1}
\mathrm{Tr}(\bb^2 \bc \bk \bc) = q^{-2}
\mathrm{Tr}(\bk \bb^2 \bc^2) = (1-q^3)^{-1}$ and 
$\mathrm{Tr}(\bb \bc \bb \bk \bc) = (1+q^2)^2\frac{(q)_1(q)_2}{(q)_4}$,
we reproduce the first two terms in (\ref{knh}) after removing the 
common factor $(q)_1(q)_2(q)_3/(1+\kappa)^2$.
\end{example}

\section{Density and currents of second class particles:\\
Formulation of the problem}
\label{sec4}

Now we come to the main theme of the paper, 
the stationary quantities in 
a grand canonical ensemble with respect to the second class particles 
in the infinite volume limit.
The first class particles are kept finite and regarded as {\em defects}.
In this section we give a general formulation of the 
problem only deferring the results and their derivation to subsequent sections.
We fix the parameters $0 < q,\mu <1$ and 
consider the $\epsilon=+1$ case of (\ref{Hs})--(\ref{ri2}) 
for the continuous time ZRP (\ref{skb}).

Let us introduce the generating series of $X_{m,n}$ (\ref{cie}) 
with respect to the 
second class particles with fugacity $y$:
\begin{align}\label{Adef}
A_m = \sum_{n \ge 0}X_{m,n}y^n
= g_m \frac{(\mu \bb)_\infty}{(\bb)_\infty}
\frac{(q^m \mu y \bk)_\infty}{(y \bk)_\infty}\bc^m,
\qquad
g_m = \frac{(\mu)_m}{(q)_m}.
\end{align}
The quantity $g_m$ introduced here will be used very frequently in the sequel.

We consider the conditional probability in the stationary states supposing that
there are $d_i$ first class particles at site $i$ for 
$i=1,\ldots, s$, and no first class particle is present elsewhere.
So they form a fixed cluster of defects whose spatial extension is $s$ and
the total number of defect particles is $d_1+\cdots + d_s$.
In terms of the site variable 
$\sigma_i = (\sigma_{i,1}, \sigma_{i,2})$, 
the condition is expressed as $\sigma_{i,1}=\theta(1 \le i \le s)d_i$,
allowing $\sigma_{i,2}$ still to fluctuate everywhere. 
We assume that $d_1 \ge 1$ and $d_s \ge 1$ to fix the location of the
defect cluster but allow the choice $d_i=0$ for $1<i < s$.
The basic quantity is the conditional probability that a given site 
$r$ contains exactly $n$ second class particles.
In the customary notation for the probability 
$P({\rm A}|{\rm B})$ of the event A under the condition B, we set
\begin{align}\label{prn}
P(r,n) := 
\begin{cases}
P(\sigma_r=(d_r,n)|\sigma_{i,1}=\theta(1 \le i \le s)d_i) & \text{if } 
1 \le r \le s\\
P(\sigma_r=(0,n)|\sigma_{i,1}=\theta(1 \le i \le s)d_i)
& \text{otherwise}.
\end{cases}
\end{align}
There are three distinguish 
regions in the infinite volume limit $L \rightarrow \infty$ as
\begin{align}\label{reg}
{\rm I}:\; 1 \le r \le s,\qquad
{\rm II}: \; r > s, \qquad
{\rm III}: \; r \le 0,
\end{align}
where $r \le 0$ should be understood as the site $r+L \in \Z_L$
in the limit $L \rightarrow \infty$.
They correspond to the inside, the right and the left side 
of the defect cluster, respectively.
The quantity (\ref{prn}) will be denoted by 
$P_{\rm I}(r,n), P_{\rm II}(r,n)$ and $P_{\rm III}(r,n)$ accordingly.

From the matrix product formula (\ref{sin}) and 
the definition (\ref{Adef}), we have
\begin{align}
P_{\rm I}(r,n)&= \lim_{L \rightarrow \infty}
\frac{y^n\mathrm{Tr}\bigl(A_{d_1}\cdots A_{d_{r-1}}X_{d_r,n}
A_{d_{r+1}}\cdots A_{d_s}
A_0^{L-s}\bigr)'}
{\mathrm{Tr}\bigl(A_{d_1}\cdots A_{d_s} 
A_0^{L-s} \bigr)'},
\label{Pmiddle}\\
P_{\rm II}(r,n) &= \lim_{L \rightarrow \infty}
\frac{y^n\mathrm{Tr}\bigl(A_{d_1}\cdots A_{d_s}
A_0^{r-s-1}X_{0,n} A_0^{L-r}\bigr)'}
{\mathrm{Tr}\bigl(A_{d_1}\cdots A_{d_s} 
A_0^{L-s} \bigr)'},
\label{Pright}\\
P_{\rm III}(r,n)&= \lim_{L \rightarrow \infty}
\frac{y^n\mathrm{Tr}\bigl(X_{0,n}A_0^{|r|} 
A_{d_1}\cdots A_{d_s} 
A_0^{L-|r|-s-1}\bigr)'}
{\mathrm{Tr}\bigl(A_{d_1}\cdots A_{d_s} 
A_0^{L-s} \bigr)'},
\label{Pleft}
\end{align}
where $(\cdots)'$ is defined after (\ref{Bd}).
In view of the matrix product operators (\ref{cie}) and (\ref{Adef}),
the prime restricts the ensemble to those sectors containing at least one 
second class particle for which the traces become finite.
The formulas (\ref{Pright})--(\ref{Pleft}) fix the relative weight of such sectors,
thereby specify what is meant by the ``grand canonical ensemble" 
with respect to the second class particles with fugacity $y$.
Obviously for any $r$ the normalization 
$\sum_{n \ge 0}P(r,n)=1$ should be fulfilled in each region.

Once the probability $P(r,n)$ is obtained, one can evaluate various
physical quantities.
In this paper we investigate 
the expectation number and currents of the second class particles at site $r$.
The former is defined by
\begin{align}\label{occ}
\rho(r) = \sum_{n \ge 0}n P(r,n).
\end{align}
According to the regions, it will be denoted by
$\rho_{\rm I}(r), \rho_{\rm II}(r)$ and $\rho_{\rm III}(r)$.
We call them {\em density} for simplicity.
As for the current, there are two components $J(r)_+$ and $J(r)_-$ 
associated with the local Markov matrices 
$h_+$ (\ref{ri1}) and $h_-$ (\ref{ri2}) with $\epsilon=+1$ respectively:
\begin{align}\label{jr}
J(r)_\pm = 
\begin{cases}
\sum_{n \ge l \ge 1}l w_{\pm}((0,l)|(d_r ,n))P(r,n) & \text{if } 
1 \le r \le s,\\
\sum_{n \ge l \ge 1}l  w_{\pm}((0,l)|(0,n))P(r,n)
& \text{otherwise}.
\end{cases}
\end{align} 
They sum up the contributions from the $l$ hopping second class particles
out of total $n+\theta(1 \le r \le s)d_r$ 
occupants weighted by the probability $P(r,n)$ and 
the hopping rate  $w_\pm$ in (\ref{wp}) and (\ref{wm}).
From (\ref{skb}) and the picture before (\ref{wp}), 
the total current from the site $r$ to $r+1$ 
is obtained by the superposition 
\begin{align}\label{fk}
J(r) = a J(r)_+ - b J(r+1)_-.
\end{align}
Thus it suffices to investigate $J(r)_+$ and $J(r)_-$ separately.
We will also write them as 
$J_{\rm I}(r)_\pm, J_{\rm II}(r)_\pm$ and 
$J_{\rm III}(r)_\pm$ according to 
the regions.

\section{Defect-free case} \label{sec5}

First we illustrate the analysis on the defect-free case $s=0$ as a warm-up.
It corresponds to the single species model, 
and some of the contents are well known by 
earlier works, e.g. \cite{BC,EH,P}.

In the absence of defects, 
the system acquires the $\Z_L$ translational symmetry.
As the result, the probabilities (\ref{Pright}) and (\ref{Pleft}) 
are equal and independent of $r$.
So we simply denote it by $P(n)$.
It is calculated as
\begin{align}
P(n) &= \lim_{L \rightarrow \infty}
\frac{y^n\mathrm{Tr}\bigl(X_{0,n}A_0^{L-1}\bigr)'}
{\mathrm{Tr}\bigl(A_0^L \bigr)'}
=  \lim_{L \rightarrow \infty}
\frac{y^ng_n\mathrm{Tr}\Bigl(
\frac{(\mu \bb)_\infty}{(\bb)_\infty}\bk^n 
\bigl( \frac{(\mu \bb)_\infty}{(\bb)_\infty}
\frac{(\mu y \bk)_\infty}{(y \bk)_\infty}\bigr)^{L-1}\Bigr)'}
{\mathrm{Tr}\Bigl(
\bigl( \frac{(\mu \bb)_\infty}{(\bb)_\infty}
\frac{(\mu y \bk)_\infty}{(y \bk)_\infty}\bigr)^{L}\Bigr)'}
\nonumber\\
&= \lim_{L \rightarrow \infty}
\frac{y^ng_n\mathrm{Tr}\Bigl(
\bk^n \bigl(\frac{(\mu y \bk)_\infty}{(y \bk)_\infty}\bigr)^{L-1}\Bigr)'}
{\mathrm{Tr}\Bigl(
\bigl(\frac{(\mu y \bk)_\infty}{(y \bk)_\infty}\bigr)^{L}\Bigr)'}
= \lim_{L \rightarrow \infty}
\frac{y^ng_n\sum_{m \ge 0}q^{mn}(\Lambda(q^m y)^{L-1}-\delta_{n,0})}
{\sum_{m \ge 0}(\Lambda(q^m y)^{L}-1)}
\nonumber\\
&= y^n g_n\Lambda(y)^{-1} 
\lim_{L \rightarrow \infty}
\frac{\sum_{m \ge 0}q^{mn}(\eta_m^{L-1}-\delta_{n,0}\eta_\infty^{L-1})}
{\sum_{m \ge 0}(\eta_m^{L}-\eta_\infty^{L})},
\label{P01}
\end{align}
where the prime is defined under (\ref{Bd}) and we have set
\begin{align}
\Lambda(y) &= \langle 0 |  \frac{(\mu y \bk)_\infty}{(y \bk)_\infty}
|0 \rangle = \frac{(\mu y)_\infty}{(y)_\infty},
\label{Ladef}\\
\eta_m & = \eta_m(y) = \Lambda(q^m y)\Lambda(y)^{-1} 
= \frac{(y)_m}{(\mu y)_m}.
\label{etadef}
\end{align}
These quantities will be utilized frequently 
in the subsequent calculations.
We will abbreviate $\eta_m(q^ky)$ to $\eta_m$ 
{\em if and only if} $k=0$.

We suppose $0 < y < 1$ whose consistency will be confirmed shortly.
Then 
\begin{align}\label{lin}
1=\eta_0 > \eta_1 > \eta_2 > \cdots \ge 0
\end{align}
holds due to $0 < q, \mu < 1$.
In Appendix \ref{app:ls} we show that 
the limit $\lim_{L\rightarrow \infty}$
and the infinite sum $\sum_{m \ge 0}$ in (\ref{P01})
may be interchanged, i.e.,
\begin{align}\label{ex1}
\lim_{L\rightarrow \infty}\sum_{m \ge 0}q^{mn}
(\eta_m^{L}-\delta_{n,0}\eta^L_\infty)
=\sum_{m \ge 0}q^{mn}\lim_{L\rightarrow \infty}
(\eta_m^{L}-\delta_{n,0}\eta^L_\infty).
\end{align}
Since this is just 
$\sum_{m \ge 0}q^{mn}\delta_{m,0}=1$, 
we obtain the probability
\begin{align}\label{P0}
P(n) = y^n \frac{(\mu)_n}{(q)_n}\frac{(y)_\infty}{(\mu y)_\infty}.
\end{align}
The correct normalization 
$\sum_{n \ge 0}P(n) = 1$ has been achieved 
by virtue of (\ref{air}).
This formally agrees with the probability \cite[eq.(53)]{BC} 
upon identification of the parameter $\alpha$ 
there with the fugacity $y$ here.

Next we relate the density $\rho$ of the second class particles
to the fugacity $y$.
In the grand canonical ensemble under consideration, it is evaluated as
\begin{align*}
\rho &= \lim_{L \rightarrow \infty}\frac{1}{L} 
\frac{\sum_{n_1,\ldots, n_L \ge 0}
(n_1+\cdots + n_L) y^{n_1+\cdots + n_L}
\mathrm{Tr}\bigl(X_{0,n_1}\cdots X_{0,n_L}\bigr)'}
{\mathrm{Tr}\bigl(A_0^L \bigr)'}
\\
&= \lim_{L \rightarrow \infty}\frac{1}{L}
y \frac{\partial}{\partial y}\log 
\mathrm{Tr}(A_0^L)'.
\end{align*}
As mentioned before (\ref{ex1}),
the ``partition function" $\mathrm{Tr}(A_0^L)'$ here
may be replaced with $\Lambda(y)^L$ as $L$ goes to infinity,
leading to 
\begin{align}
\rho &= y \frac{\partial}{\partial y}\log \Lambda(y) = f(y)- f(\mu y)
=\sum_{i\ge 0}
\frac{(1-\mu)yq^i}{(1-y q^i)(1-\mu y q^i)},
\label{rhof}\\
f(\zeta) &= -\zeta \frac{\partial}{\partial \zeta}\log (\zeta)_\infty
= \sum_{i\ge 1}\frac{\zeta^i}{1-q^i}
=\sum_{i \ge 0} \frac{\zeta q^i}{1-\zeta q^i}.
\label{fdef}
\end{align}
The $f(\zeta)$ is a version of the $q$-digamma function.
It monotonously grows from $0$ to $\infty$ 
as $\zeta$ changes from $0$ to $1$ behaving as  
$f(\zeta) \simeq \frac{\zeta}{1-q}\, (\zeta \searrow 0)$ and 
$f(\zeta) \simeq \frac{1}{1-\zeta}\, (\zeta \nearrow 1)$.
Thus the fugacity and density are related asymptotically as
\begin{align}\label{ry}
\rho \simeq \frac{1-\mu}{1-q}y\;\;\;
(\rho \rightarrow 0,\, y \rightarrow 0), \qquad
\rho \simeq \frac{1}{1-y}\;\;\; (\rho \rightarrow \infty,\,
y \rightarrow 1).
\end{align}
The difference $f(y)- f(\mu y)$ in (\ref{rhof}) is similarly 
increasing monotonously from $0$ to $\infty$ for 
$y \in (0,1)$.
Thus $\rho \in (0,\infty)$ is in one-to-one correspondence 
with $y \in (0,1)$, which confirms the consistency of the assumption 
made before (\ref{lin}).
It implies that there is no phase transition typically recognized as condensation 
(cf. \cite{EH}) in ZRPs.
To summarize, 
(\ref{rhof}) determines the relation $y=y(\rho)$ and $\rho=\rho(y)$.

The density $\rho(r)$ (\ref{occ}) is also independent of $r$.
From (\ref{P0}) it is evaluated as
\begin{align*}
\rho(r) = \sum_{n \ge 0} nP(n) = \frac{(y)_\infty}{(\mu y)_\infty}
\sum_{n\ge 0}n y^n\frac{(\mu)_n}{(q)_n} = \rho
\end{align*}
in terms of the average density $\rho$ in (\ref{rhof}) confirming 
the consistency.

Now we are ready to evaluate the currents of the second class particles
(\ref{jr}).
Since they are independent of the site $r$, we simply write it as 
\begin{align*}
J_\pm = \sum_{n \ge l \ge 1}l w_{\pm}((0,l)|(0,n))P(n).
\end{align*} 
For instance $J_-$ is computed from (\ref{wm}) and (\ref{P0}) as
\begin{equation}\label{J0m}
\begin{split}
J_- &= \frac{(y)_\infty}{(\mu y)_\infty}
\sum_{n \ge l \ge 1}y^n \frac{(\mu)_n}{(q)_n}
\frac{l(q)_{l-1}}{(\mu q^{n-l})_l}
\binom{n}{l}_q
\\
&=  \frac{(y)_\infty}{(\mu y)_\infty}
\sum_{n -l \ge 0}
y^{n-l}\frac{(\mu)_{n-l}}{(q)_{n-l}} 
\sum_{l \ge 1}\frac{l y^l}{1-q^l}
=  \sum_{l \ge 1}\frac{l y^l}{1-q^l},
\end{split}
\end{equation}
where the last equality is due to (\ref{air}).
A similar calculation for $J_+$ leads to 
current-density relation, a basic characteristic of the system, given as 
\begin{align}
J_+ = \mu^{-1}h(\mu y),\qquad 
J_- = h(y)
\label{J0}
\end{align}
via $y=y(\rho)$ (\ref{rhof}).
Here the function $h(\zeta)$ is the derivative of the $q$-digamma function 
$f(\zeta)$ in (\ref{fdef}):
\begin{align}\label{hdef}
h(\zeta) = \zeta \frac{df(\zeta)}{d\zeta} = 
\sum_{i \ge 1}\frac{i \zeta^i}{1-q^i} = 
\sum_{i \ge 0}\frac{\zeta q^i}{(1-\zeta q^i)^2}.
\end{align}
It obviously satisfies the relation
\begin{align}\label{hh}
h(q^i \zeta) = h(\zeta) 
- \sum_{k=0}^{i-1}\frac{\zeta q^k}{(1-\zeta q^k)^2}\qquad (i \ge 0).
\end{align}
The result (\ref{J0m}) is restated as 
\begin{align}\label{wph}
\sum_{n \ge l \ge 1}l w_-((0,l)|(0,n))P(n) = h(y), 
\end{align}
which is $\mu$-independent 
despite that $w_-((0,l)|(0,n))$ (\ref{wm}) 
and $P(n)$ (\ref{P0}) depend on $\mu$ individually.
An expression similar to 
the total current $aJ_+ - bJ_-$ (\ref{fk}) with 
$J_\pm$ given by (\ref{J0}) has also been obtained in 
\cite[Sec.4.1]{BC}.

The function $h(\zeta)$ is monotonously increasing and
tends to $\infty$ as $\zeta$ approaches 1 from below.
Consequently 
the both currents in (\ref{J0}) grow monotonously with the density
$\rho$ via (\ref{rhof}).
Their leading asymptotic behavior is given by
\begin{align}
&J_+ \simeq J_- \simeq \frac{\rho}{1-\mu}
\;\;\text{as }\;\; \rho \rightarrow 0,
\label{jlow}\\
&J_+ \rightarrow \mu^{-1}h(\mu),\quad
J_- \simeq \rho^2 \;\;\text{as }\;\; \rho \rightarrow \infty.
\label{jhi}
\end{align}
So $J_+$ converges to a finite value 
whereas $J_-$ diverges 
as the density $\rho$ gets large.
Such large $\rho$ behavior is out of question 
in asymmetric simple exclusion processes where the well known relation
$J=\mathrm{const}\,\rho(1-\rho)$ makes sense only for $0\le \rho \le 1$.

\begin{figure}[H]
\begin{center}
\includegraphics[scale=0.55]{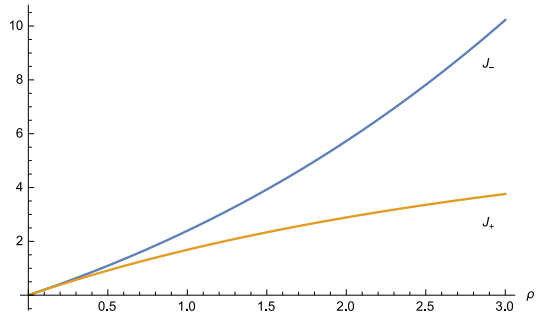}
\caption{Comparison of $J_+$ and $J_-$ in (\ref{J0}) 
as functions of the density $\rho$ 
for $(q,\mu)=(0.7, 0.5)$.}
\end{center}
\label{jbare1}
\end{figure}

\begin{figure}[H]
\begin{tabular}{cc}
\begin{minipage}[t]{0.45\hsize}
\begin{center}
\includegraphics[scale=0.6]{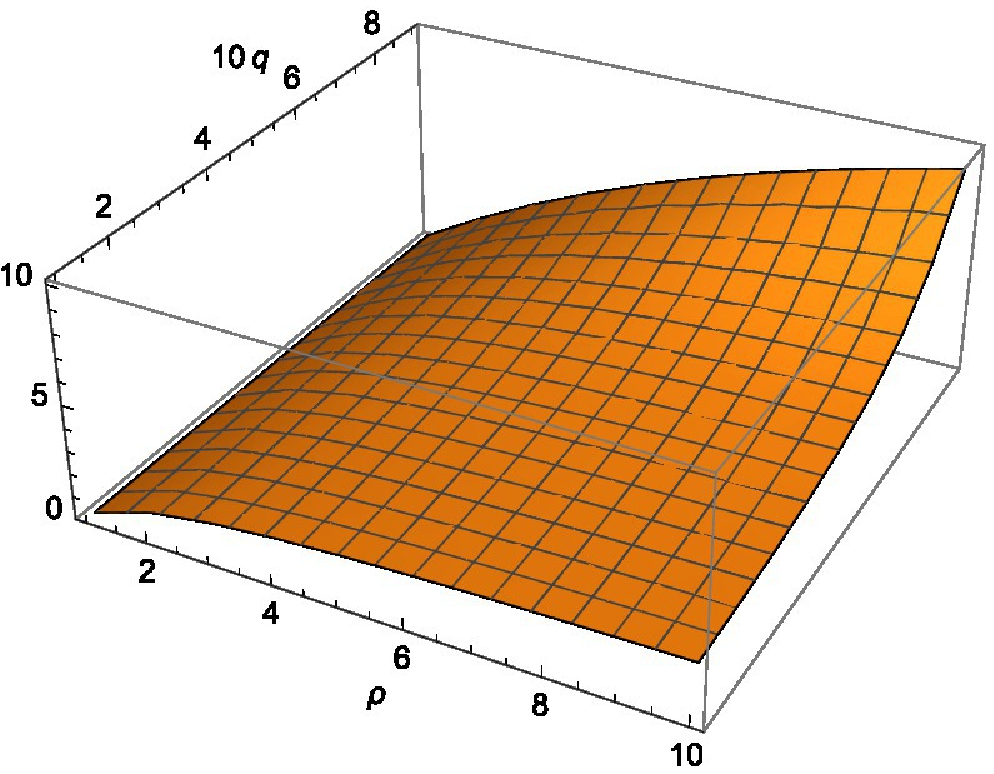}
\end{center}
\end{minipage}
&\;\;
\begin{minipage}[t]{0.45\hsize}
\begin{center}
\includegraphics[scale=0.6]{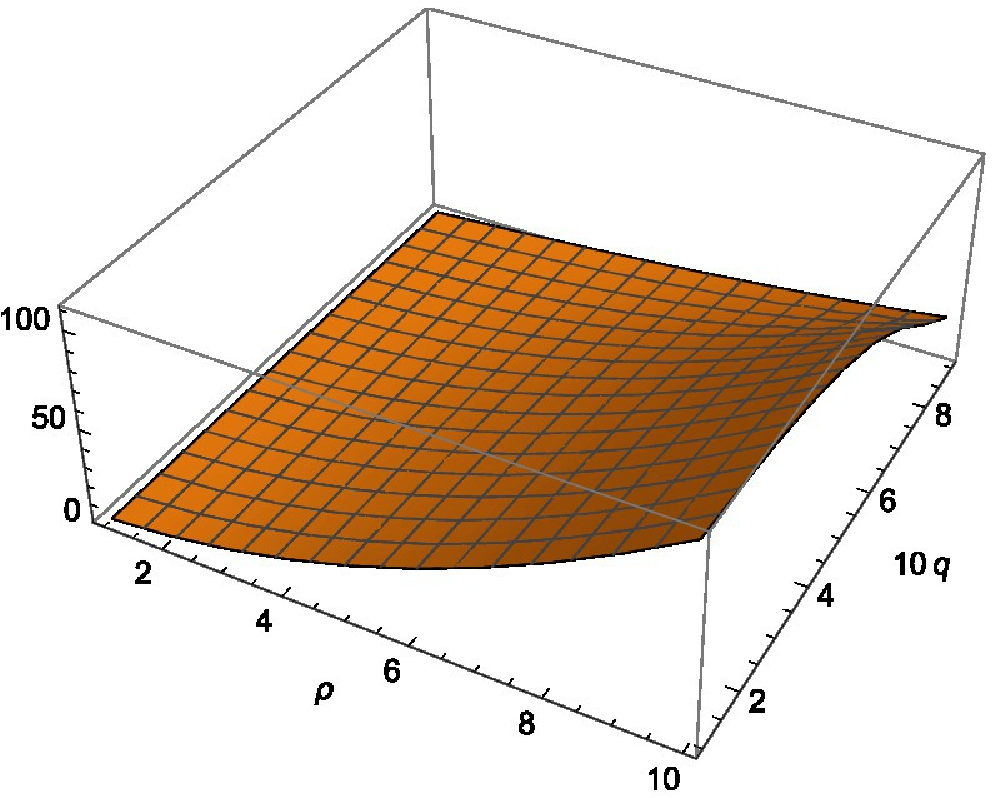}
\end{center}
\end{minipage}
\end{tabular}
\caption{The current $J_+$ (left) and $J_-$ (right) in 
(\ref{J0}) with $\mu=0.4$ in the range $0.1<\rho< 10$ and $0.1 \le q \le  0.9$.}
\label{jbare2}
\end{figure}

\begin{figure}[H]
\begin{tabular}{cc}
\begin{minipage}[t]{0.45\hsize}
\begin{center}
\includegraphics[scale=0.6]{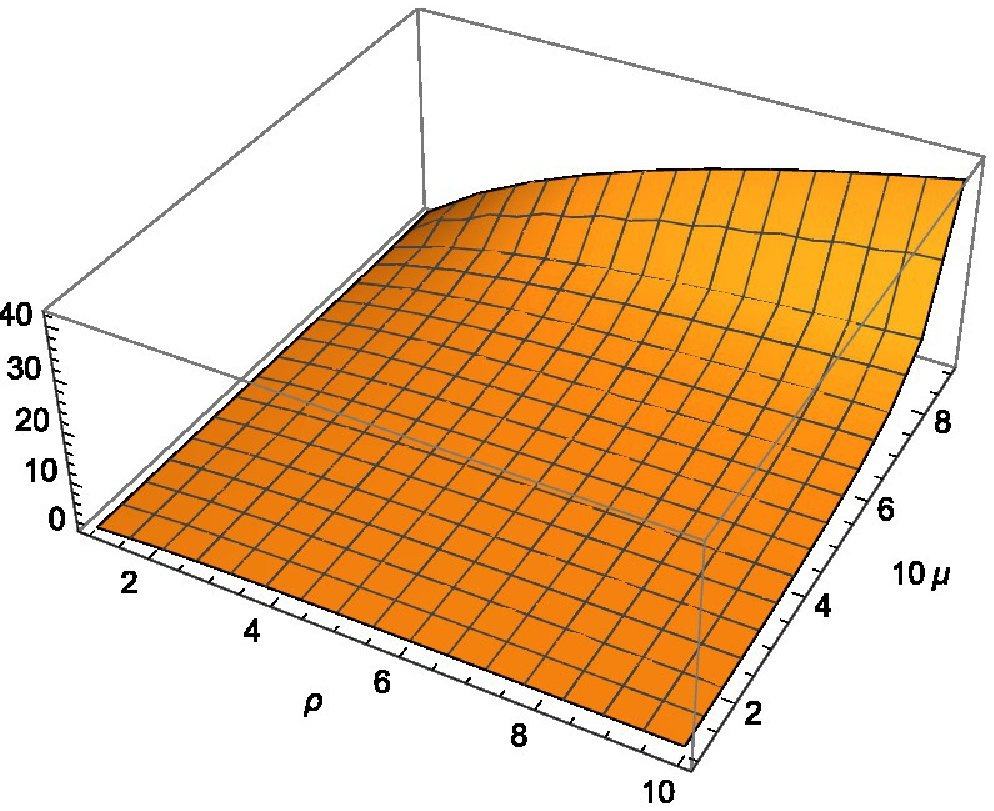}
\end{center}
\end{minipage}
&\;\;
\begin{minipage}[t]{0.45\hsize}
\begin{center}
\includegraphics[scale=0.6]{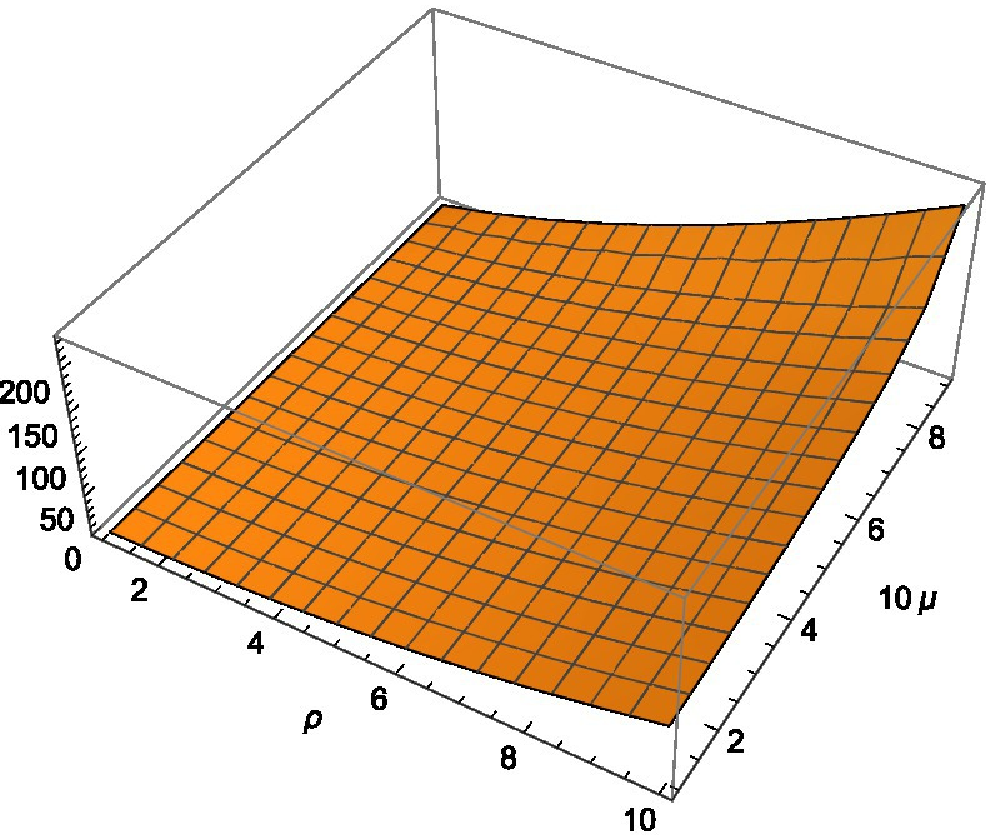}
\end{center}
\end{minipage}
\end{tabular}
\caption{The current $J_+$ (left) and $J_-$ (right) in 
(\ref{J0}) with $q=0.5$ in the range $0.1<\rho< 10$ and $0.1 \le \mu \le  0.9$.}
\label{jbare3}
\end{figure}

We close the section with the description on the limiting cases 
$q, \mu \rightarrow 0,1$.
We shall only give the leading terms.
\begin{alignat*}{3}
{\rm (i)}\; q\rightarrow 0: 
\;
\rho &= \frac{(1-\mu)y}{(1-y)(1-\mu y)},
\quad
&J_{\pm} &= \frac{y}{(1-\mu^{(1\pm 1)/2} y)^2},
\quad 
&P(n) &= y^n\frac{1-y}{1-\mu y} (1-\mu)^{\theta(n \ge 1)},
\\
{\rm (ii)}\; q \rightarrow 1: 
\;
y &\simeq -\frac{\rho\log q}{1-\mu}, 
\quad
&J_\pm &= \frac{\rho}{1-\mu},
\quad
&P(n) &= \frac{\rho^n\mathrm{e}^{-\rho}}{n!},
\\
{\rm (iii)} \; \mu \rightarrow 0:
\;
\rho &= f(y), \quad
&J_+ & = \frac{y}{1-q},
\;\,
J_- =h(y),
\quad
&P(n) &= y^n\frac{(y)_\infty}{(q)_n},
\\
{\rm (iv)} \; \mu \rightarrow 1:
\;
y&  \simeq 1-\sqrt{\frac{1-\mu}{\rho}},
\quad
&J_\pm& \simeq \frac{\rho}{1-\mu},
\quad
&P(n)& = \delta_{n,0}.
\end{alignat*}

\section{Density and currents: Results in general case}\label{sec:results}

Let us proceed to the general case in which 
$d_1,\ldots, d_s$ defect (first class) particles are present at the sites 
$1,\ldots, s$.
We present the final results on the 
conditional probability  $P(r,n)$ (\ref{prn}), 
the local density $\rho(r)$ (\ref{occ}) 
and the currents $J(r)_\pm$ (\ref{jr}) in the
regions I, II and III in  
Theorem \ref{th:naka}, \ref{th:migi} and \ref{th:hidari}.
Some of them look bit messy but the point is that 
they are always {\em finite} sums of 
appropriate building blocks, which 
allow accurate numerical evaluations.
Another point is that they are expressed in the form that elucidates 
the difference from the defect-free case $s=0$.
We continue to stay in the range $0<q,\mu, y <1$ 
and use the following quantities that have already appeared in 
the previous section:

\vspace{0.1cm}
$\rho$: average density of the second class particles in the entire system,

$y \in (0,1)$: fugacity of the second class particles 
determined from $\rho$
via (\ref{rhof}),

$P(n)$: probability in the defect-free case (\ref{P0}), 

$J_\pm$: currents in the defect-free case (\ref{J0}),

$h(\zeta)$: derivative of $q$-digamma function describing the currents (\ref{hdef}),

$\eta_j$: quantity controlling the decay of correlations (\ref{etadef}).

\vspace{0.2cm}\noindent
In addition to them the following functions, detailed in 
Section \ref{sec6}, will be the basic ingredients:

\vspace{0.1cm}
$\phi(l|m)$: constituent of $G_{m,l}(d_1,\ldots, d_s)$ (\ref{phidef}),

$G_{m,l}(d_1,\ldots, d_s)$: building block incorporating the effect of defects 
(\ref{Gdef}).

\vspace{0.2cm}\noindent
The relation (\ref{rhof}) between the fugacity $y$ and the 
average density $\rho$ of the second class particles was originally derived for the 
defect-free case.
We will justify its use in the presence of defects 
in Proposition \ref{pr:excess} and the comments following it.

\subsection{Main results}
First we consider the region inside the defect cluster.

\begin{theorem}\label{th:naka}
In the region I $(1 \le r \le s)$  the following formulas are valid:
\begin{align}
&P_{\rm I}(r,n) = y^n\frac{(q^{d_r}\mu)_n(y)_\infty}
{(q)_n(\mu y)_\infty}
\sum_{m}q^{n(m-d_r)}\frac{(\mu y)_m}{(y)_{m-d_r}}
G_{0,m}(d_1,\ldots, d_{r}),
\label{Pk2}\\
&\rho_{\rm I}(r) -\rho=
\sum_{m}G_{0,m}(d_1,\ldots, d_{r})
\Bigl(\sum_{k=0}^{m-1}\frac{\mu y q^k}{1-\mu y q^k}
-\sum_{k=0}^{m-d_r-1}\frac{y q^k}{1-y q^k}\Bigr),
\label{N2}\\
&J_{\rm I}(r)_{+} -J_+ =
-\sum_{m}
G_{0,m}(d_1,\ldots, d_{r})
\sum_{k=0}^{m-1}
\frac{y q^k}{(1-\mu y q^k)^2},
\label{Jf21}\\
&J_{\rm I}(r)_{-}  -J_- = 
-\sum_{m}
G_{0,m}(d_1,\ldots, d_{r})
\sum_{k=0}^{m-d_r-1}
\frac{y q^k}{(1-y q^k)^2},
\label{Jf22}
\end{align}
where the sums $\sum_{m}$  
extend over $m \in \Z_{\ge 0}$ satisfying 
$d_{r}\le m \le d_1+\cdots + d_{r}$.
\end{theorem}

The proof will be given in Section \ref{ss:naka}.
By means of (\ref{air}) one finds that  
the total probability is expressed as
$\sum_{n\ge 0}P_{\rm I}(r,n)
= \sum_{m}G_{0,m}(d_1,\ldots, d_{r})$.
This indeed gives $1$ thanks to (\ref{G1}).
In Section \ref{ss:excess} we will show that (\ref{N2}) 
can also be expressed as
\begin{align}
\rho_{\rm I}(r) -\rho &= -d_r + K(r)-K(r-1),
\label{rrr}\\
K(r) &= \sum_{m}G_{0,m}(d_1,\ldots, d_r)\sum_{k=0}^{m-1}
\frac{1}{1-\mu yq^k},
\label{Kdef}
\end{align}
where $K(0)=0$ and 
the sum over $m$ is taken in the same way as in Theorem \ref{th:naka}.
The difference structure $K(r)-K(r-1)$ in (\ref{rrr})
matches the sum rule in Proposition \ref{pr:excess}.

\begin{example}\label{ex:naka1}
At $r=1$ which is the left boundary 
of the defect cluster, 
Theorem \ref{th:naka} simplifies to
\begin{align*}
&P_{\rm I}(1,n)= y^n
\frac{(q^{d_1}\mu)_n(y)_\infty}{(q)_n(q^{d_1}\mu y)_\infty},
\qquad
\rho_{\rm I}(1)= \rho + \sum_{k=0}^{d_1-1}
\frac{\mu y q^k}{1-\mu y q^k},\\
&J_{\rm I}(1)_{+} = \mu^{-1}h(q^{d_1}\mu y),\quad
J_{\rm I}(1)_{-} = h(y)
\end{align*}
due to (\ref{Gex1}).
Thus we always have $\rho_{\rm I}(1)>\rho$ because of $d_1\ge 1$. 
The last result $J_{\rm I}(1)_{-} = h(y)$ may seem strange.
It follows from the properties 
\begin{align*}
P_{\rm I}(1,n) = P(n)|_{\mu \rightarrow q^{d_1}\mu},\quad
w_-((0,l)|(d_1,n)) = w_-((0,l)|(0,n))|_{\mu \rightarrow q^{d_1}\mu},
\end{align*}
which are easily seen in (\ref{P0}) and (\ref{wm})
and the remark on (\ref{wph}). 
\end{example}

\begin{example}\label{ex:naka2}
At $r=2$ which is one step inside the defect cluster from the left, 
Theorem \ref{th:naka} takes the form
\begin{align*}
P_{\rm I}(2,n)&= y^n
\frac{(q^{d_2}\mu)_n(y)_\infty}{(q)_n(\mu y)_\infty}
\sum_{m=d_2}^{d_1+d_2} q^{n(m-d_2)}
\frac{(\mu y)_m}{(y)_{m-d_2}}\phi(d_1+d_2-m|d_1),\\
\rho_{\rm I}(2)-\rho&=
\sum_{m=d_2}^{d_1+d_2}\phi(d_1+d_2-m|d_1)
\left(\sum_{k=0}^{m-1}\frac{\mu y q^k}{1-\mu y q^k}
-\sum_{k=0}^{m-d_2-1}\frac{y q^k}{1-y q^k}\right),
\\
J_{\rm I}(2)_{+} - J_+& = 
\sum_{m=d_2}^{d_1+d_2}\phi(d_1+d_2-m|d_1)
\sum_{k=0}^{m-1}\frac{y q^k}{(1-\mu y q^k)^2},
\\
J_{\rm I}(2)_{-}  - J_-& = 
\sum_{m=d_2}^{d_1+d_2}\phi(d_1+d_2-m|d_1)
\sum_{k=0}^{m-d_2-1}\frac{y q^k}{(1- y q^k)^2},
\end{align*}
where $G_{0,m}(d_1,d_2)=\phi(d_1+d_2-m|d_1)$ has been used by (\ref{Gex2}).   
\end{example}

Let us remark on the low and high density asymptotic behavior.
They correspond to $y \searrow 0$ and $y \nearrow 1$, respectively.
See (\ref{ry}).
From the properties of 
the unperturbed density and currents in 
(\ref{ry}), (\ref{jlow}) and (\ref{jhi}), 
one can derive the following behavior by utilizing 
(\ref{Gy}) and (\ref{ayk}).
\begin{align}
&\lim_{\rho\rightarrow 0}\rho_{\rm I}(r)/\rho=
\frac{1-\mu q^{d_r}}{1-\mu}q^{d_1+\cdots + d_{r-1}}, 
\label{rir}\\
&\lim_{\rho\rightarrow \infty}(\rho_{\rm I}(r)-\rho) 
= \sum_{k=0}^{d_r-1}\frac{\mu q^k}{1-\mu q^k}
-\sum_{k=0}^{d_{r-1}-1}
\frac{1}{1-\mu q^k},
\label{rir0}\\
&\lim_{\rho \rightarrow 0}
J_{\rm I}(r)_{+}/J_+ = q^{d_1+\cdots + d_r},
\qquad
\lim_{\rho \rightarrow 0}
J_{\rm I}(r)_{-}/J_- = q^{d_1+\cdots + d_{r-1}},
\label{irir}\\
&\lim_{\rho\rightarrow \infty}(J_{\rm I}(r)_{+} - J_+)=
-\sum_{k=0}^{d_r-1}\frac{q^k}{1-\mu q^k},
\\
&\lim_{\rho\rightarrow \infty}(J_{\rm I}(r)_{-} - J_-)/\rho=
-\sum_{k=0}^{d_{r-1}-1}
\frac{1}{1-\mu q^k}.
\label{ktn}
\end{align}
The result (\ref{rir}) with $r=1$ agrees with the 
$y\rightarrow 0$ case of Example \ref{ex:naka1} implied by (\ref{ry}).

\vspace{0.2cm}
Next we consider the right region II.

\begin{theorem}\label{th:migi}
In the region II $(r>s)$ the following formulas are valid:
\begin{align}
&P_{\rm II}(r,n) = P(n)
\sum_{i,j,m}G_{0,m}(d_1,\ldots, d_s)
(-1)^{i}q^{\frac{1}{2}i(i-1+2n)+j}
\frac{(q^{-m})_j}{(q)_i(q)_{j-i}}
\eta_i^{-1}\eta_j^{r-s},
\label{Pk1}\\
&\rho_{\rm II}(r)-\rho
= \sum_{j, m}\!G_{0,m}(d_1,\ldots, d_s)
\frac{q^j(q^{-m})_j\eta^{r-s}_j}{1-q^j}
\bigl(\frac{1}{(y)_j} - \frac{1}{(\mu y)_j}\bigr),
\label{N1}\\
&J_{\rm II}(r)_{\pm} - J_\pm
= \sum_{j, m}\!G_{0,m}(d_1,\ldots, d_s)
\frac{q^j(q^{-m})_j \eta_j^{r-s}}{(1-q^j)(\mu^{(1\pm1)/2}y)_j}
\sum_{k=0}^{j-1}
\frac{yq^k}{1-\mu^{(1\pm1)/2}yq^k},
\label{Jf1}
\end{align}
where the sum in (\ref{Pk1})  
extends over $i,j,m \in \Z_{\ge 0}$ such that
$0 \le i \le j \le m$ and 
$d_s \le m  \le d_1+\cdots + d_s$.
The sums in (\ref{N1}) and (\ref{Jf1}) are taken for
$1 \le j \le m$ and 
$d_s \le m  \le d_1+\cdots + d_s$.
\end{theorem}

The proof will be given in Section \ref{ss:migi}
It is based on the choice  
$(d_1,\ldots, d_r) =(d_1,\ldots, d_s, \overset{r-s}{\overbrace{0,\ldots,0}})$
in Theorem \ref{th:naka}.
In this sense Theorem \ref{th:naka} covers 
Theorem \ref{th:migi} save the 
extraction of the dependence on the 
distance $r-s$ in (\ref{Gr}).

\begin{example}\label{ex:pk1}
Consider the case  $s=1$.
From (\ref{Gex1}) we know $G_{0,m}(d_1) = \delta_{m,d_1}$.
Thus (\ref{Pk1}), (\ref{N1}) and  (\ref{Jf1})  
for $d_1=1$ and $d_1=2$ read
\begin{align*}
P_{\rm II}(r,n) & = 
y^n\frac{(\mu)_n (y)_\infty}{(q)_n (\mu y)_\infty}
\bigl(1-\eta_1^{r-1}
+ q^n \eta_1^{r-2}\bigr),\\
P_{\rm II}(r,n) & = 
y^n\frac{(\mu)_n (y)_\infty}{(q)_n (\mu y)_\infty}
\bigl(
1-q^{-1}\binom{2}{1}_q\eta_1^{r-1}
+ q^{-1}\eta_2^{r-1} +
q^{n-1}\binom{2}{1}_q(\eta_1^{r-2}-
\eta_1^{-1}\eta_2^{r-1})+
q^{2n}\eta_2^{r-2}\bigr),
\\
\rho_{\rm II}(r) -\rho& = 
-\frac{(1-\mu)y\eta_1^{r-1}}{(1-y)(1-\mu y)},
\\
\rho_{\rm II}(r) -\rho& = 
-\frac{(1-\mu)yq^{-1}((1+q)\eta_1^{r-1}-\eta_2^{r-1})}{(1-y)(1-\mu y)}
-\frac{(1-\mu)yq\eta_2^{r-1}}{(1-yq)(1-\mu yq)},
\\
J_{\rm II}(r)_{\pm} - J_\pm
&= -\frac{y\eta_1^{r-1}}{(1-\mu^{(1\pm 1)/2}y)^2},\\
J_{\rm II}(r)_{\pm} - J_\pm
&= \frac{yq^{-1}\eta_2^{r-1}-y(1+q^{-1})\eta_1^{r-1}}
{(1-\mu^{(1\pm 1)/2}y)^2}
- \frac{yq \eta_2^{r-1}}{(1-\mu^{(1\pm 1)/2}yq)^2},
\end{align*}
where $\eta_1= \frac{1-y}{1-\mu y}$,  
$\eta_2  = \frac{(1-y)(1-qy)}{(1-\mu y)(1-q \mu y)}$
by (\ref{etadef}).
\end{example}

In general the effect of defects disappears
in the long distance, i.e.,
\begin{align}\label{rinf}
\lim_{r \rightarrow \infty}P_{\rm II}(r,n) = P(n),
\qquad
\lim_{r \rightarrow \infty}\rho_{\rm II}(r) = \rho,
\qquad
\lim_{r \rightarrow \infty}J_{\rm II}(r)_{\pm}= J_\pm.
\end{align}
These are derived from $\lim_{r\rightarrow \infty}
\eta_j^{r-s}=\delta_{j,0}$ by (\ref{lin}) and (\ref{G1}).
As seen in Example \ref{ex:pk1}, 
deviation from the defect-free case comes in various {\em mode}
proportional to $\eta_j^{r-s}$ having the 
correlation (or decay)  length $-(\log \eta_j)^{-1}\, (j \ge 1)$. 
Thus from (\ref{lin}) and (\ref{etadef}) 
the influence of the defects reaches the distance 
of order $\bigl(\log \frac{1-\mu y}{1-y}\bigr)^{-1}$.
Since $y\rightarrow 0$ as the density tends to $0$,
the correlation is longer for smaller density and 
$\mu$ closer to $1$. 

By a calculation similar to the region I, one can show for $r>s$ the 
asymptotic behavior as follows.
\begin{alignat}{2}
&\lim_{\rho \rightarrow 0}\rho_{\rm II}(r)/\rho=q^{d_1+\cdots + d_s},
&
&\lim_{\rho \rightarrow \infty}(\rho_{\rm II}(r)-\rho)
= -\delta_{r,s+1}\sum_{k=0}^{d_s-1}\frac{1}{1-\mu q^k},
\label{msw}\\
&\lim_{\rho \rightarrow 0}
J_{\rm II}(r)_{\pm}/J_\pm = q^{d_1+\cdots + d_s},
&\quad
&
\label{jjq}\\
&\lim_{\rho \rightarrow \infty}
(J_{\rm II}(r)_{+} - J_+)=0,
&
&\lim_{\rho \rightarrow \infty}
(J_{\rm II}(r)_{-} - J_-)/\rho = 
-\delta_{r,s+1}\sum_{k=0}^{d_s-1}\frac{1}{1-\mu q^k}.
\label{rna}
\end{alignat}
These results have the form that
naturally extends (\ref{rir}) -- (\ref{ktn}) to $r>s$.
In particular in the large $\rho$ limit, 
the region II feels the rightmost $d_s$ defects only.

\vspace{0.2cm}
Finally we consider the left region III.
\begin{theorem}\label{th:hidari}
In the region III $(r \le 0)$  the following formulas are valid:
\begin{align}\label{N3}
P_{\rm III}(r,n) = P(n),
\qquad
\rho_{\rm III}(r) = \rho,
\qquad
J_{\rm III}(r)_{\pm} = J_\pm.
\end{align}
\end{theorem}
The proof will be given in Section \ref{ss:hidari}.

Let us introduce the quantity
\begin{align}
\Delta\rho_{\rm tot}:=\sum_{r=-\infty}^\infty(\rho(r)-\rho)
= \sum_{r=1}^s(\rho_{\rm I}(r)-\rho) 
+ \sum_{r>s}(\rho_{\rm II}(r)-\rho),
\end{align}
which represents the {\em total excess} of the number of the second class particles 
from the average value in the entire system. 
From (\ref{rir0}) and (\ref{msw}) one can easily check
$\lim_{\rho \rightarrow \infty}\Delta\rho_{\rm tot} = -(d_1+\cdots + d_s)$.
Our next result shows that this holds in general.

\begin{proposition}\label{pr:excess}
The total excess of the second class particles is given by 
\begin{align*}
\Delta\rho_{\rm tot} = -(d_1+\cdots + d_s).
\end{align*}
\end{proposition}

The proof will be given in Section \ref{ss:excess}.
Since $\Delta\rho_{\rm tot}$ is finite, 
the parameter $\rho$ indeed remains to be 
the average density of the second class particles
in the infinite volume limit.
It is curious that $\Delta\rho_{\rm tot}$ exactly cancels the total
excess $+(d_1+\cdots + d_s)$ coming from the first class particles.
To find a simple explanation of this fact is an interesting open problem.

The results in Theorem \ref{th:naka} 
(resp. Theorem \ref{th:hidari}) are 
independent of $d_{r+1},\ldots, d_s$ (resp. $d_{1},\ldots, d_s$).
They imply that the defects only influence their right in the large volume limit 
despite that the component $H_-$ in the Markov matrix (\ref{skb}) 
represents the left moving particles.
We do not have an intuitive explanation of this fact.
One should however remember that 
the both $H_\pm$ in (\ref{Hs}) originate in the discrete time 
evolution $T(\lambda|\mu,\ldots, \mu)$  
described by the right moving particles only as in (\ref{tdiag}).
Moreover the derivative in (\ref{Hs}) 
gives rise to the denominator $T(\lambda|\mu,\ldots, \mu)^{-1}$ which is  
$1$ for $H_+$ whereas it is the {\em left} cyclic shift for $H_-$
due to 
$\mathscr{S}(\mu,\mu)(|\alpha\rangle \otimes |\beta\rangle )
=|\beta\rangle \otimes |\alpha\rangle$. 
Thus we could have had the right moving dynamics only,
if not so neatly described,  
just by taking the derivative rather than the logarithmic derivative
at $\lambda=\mu$.
In other words, the curiosity is attributed to the 
commutativity $[H_+, H_-]=0$, a basic consequence of the 
integrability of the model, which implies that 
the left and the right moving dynamics should possess 
the {\em same} stationary states on the ring of finite size.

\subsection{Case of homogeneous defect}
To grasp the results in Theorem \ref{th:naka} and \ref{th:migi} qualitatively,
it is helpful to consider the homogeneous case 
$d_1 = \cdots = d_s = d$ and 
the limits $\rho \rightarrow 0$ and 
$\rho \rightarrow \infty$.
The behavior for general $\rho$ can then be
inferred as something in between.
In this setting we still have the integer parameters $d, s \in \Z_{\ge 1}$
and the continuous parameters $0<q,\mu < 1$.
Set
\begin{align*}
D = \sum_{k=0}^{d-1}\frac{1}{1-\mu q^k},\qquad
\nu = \frac{1-\mu q^d}{1-\mu}.
\end{align*}
They satisfy $D>d$  and $\nu >1$ obviously.
The results on the local density 
(\ref{rir}), (\ref{rir0}), (\ref{msw}) and (\ref{N3}) are summarized as
\begin{align}\label{asym}
\lim_{\rho\rightarrow 0}
\rho(r)/\rho = \begin{cases}
\nu q^{(r-1)d} & 1 \le r \le s,\\
q^{sd} & r>s,\\
1 & \text{otherwise},
\end{cases}
\qquad
\lim_{\rho \rightarrow \infty}
(\rho(r)-\rho) = \begin{cases}
D-d & r=1,\\
-d & 2 \le r \le s, \\
-D & r=s+1,\\
0 & \text{otherwise}.
\end{cases}
\end{align}
Schematic plots of them look as Figure \ref{sch} below.

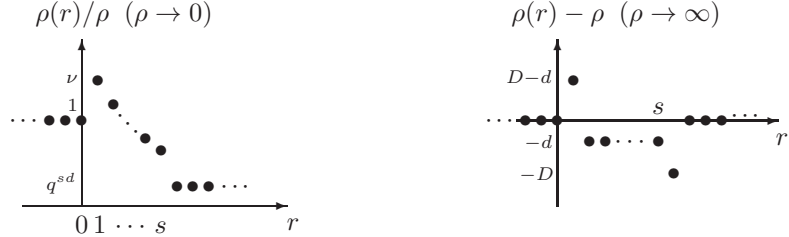
\begin{figure}[H]
\begin{tabular}{cc}
\begin{minipage}[t]{0.45\hsize}
\begin{center}

\begin{picture}(270,95)(0,-12)

\put(-15,70){$\rho(r)/\rho\; \;(\rho \rightarrow 0)$}
\put(2.4,0){\vector(0,1){63}}
\put(-20,0){\vector(1,0){100}}
\put(80,-9){$r$}

\multiput(-25,30.1)(4,0){3}{$\cdot$}

\put(-3.5,46.7){$\scriptstyle{\nu}$}
\put(-3,36.2){$\scriptstyle{1}$}
\multiput(-12.2,30)(6,0){3}{$\bullet$}

\put(6,45){$\bullet$}
\put(12,36){$\bullet$}
\put(15.6,31.8){$\cdot$}\put(18.5,28.8){$\cdot$}\put(21.6,25.8){$\cdot$}
\put(24,23.04){$\bullet$}
\put(30,18.4){$\bullet$}
\put(36,5){$\bullet$}
\put(42,5){$\bullet$}
\put(48,5){$\bullet$}
\multiput(55,5.1)(4,0){3}{$\cdot$}

\put(-11,6){$\scriptstyle{q^{sd}}$}
\put(0,-10){$0\, 1\, \cdots\, s$}

\put(180,0){

\put(-15,70){$\rho(r)- \rho\;\; (\rho \rightarrow \infty)$}
\put(2.4,0){\vector(0,1){63}}
\put(-13,32.4){\vector(1,0){99}}
\put(85,23){$r$}

\multiput(-25,30.1)(4,0){3}{$\cdot$}

\put(-17,46.7){$\scriptstyle{D-d}$}
\put(-10,22){$\scriptstyle{-d}$}
\multiput(-12.2,30)(6,0){3}{$\bullet$}

\put(6,45){$\bullet$}

\put(6,35){$\phantom{1\; \;\cdots\;\; \;}s$}
\put(12,22){$\bullet$}\put(18,22){$\bullet$}
\put(24,22){$\cdots$}
\put(38,22){$\bullet$}

\put(44,10){$\bullet$}

\multiput(50,30)(6,0){3}{$\bullet$}

\multiput(68,32)(4,0){3}{$\cdot$}

\put(-12,10){$\scriptstyle{-D}$}

}

\end{picture}

\end{center}
\end{minipage}
\end{tabular}
\caption{Density profiles under the homogeneous defect 
$d_1= \cdots = d_s = d$ in 
the two limits $\rho \rightarrow 0$ and 
$\rho \rightarrow \infty$ according to (\ref{asym}).}
\label{sch}
\end{figure}

Actually in the limit $\rho\rightarrow 0$, either 
$\rho(s)/\rho \gtrless 1$ can happen
according to $\nu q^{(s-1)d} \gtrless 1$.
In the other limit $\rho \rightarrow \infty$,
the density completely resumes the average value for $r \ge s+2$.
For general $0 < \rho < \infty$,  
the local density $\rho(r)$ gets larger than $\rho$ 
at the left boundary $r=1$ of the defect cluster forming a peak.
It then decreases until the site $r=s+1$ 
forming a valley at the left boundary of the region II.
Then in the region II, it recovers toward the average value $\rho$ 
exponentially as mentioned after (\ref{rinf}).
Such behavior and sensitivity to the defects are sharper in higher density case.

As for the local currents we are to compare 
$J(r)_+$ and $J(r+1)_-$ in view of (\ref{fk}).
The results given in 
(\ref{irir}) -- (\ref{ktn}) and (\ref{jjq}) --(\ref{N3}) are summarized as
\begin{align*}
&\lim_{\rho \rightarrow 0}J_+(r)/J_+
=\lim_{\rho \rightarrow 0}J_-(r+1)/J_-
= \begin{cases}
q^{\min(r,s)d} & r \ge 0,\\
1 & \text{otherwise},
\end{cases}
\\
&\lim_{\rho \rightarrow \infty}
(J_+(r)-J_+) = \begin{cases}
\mu^{-1}(d-D) & 1 \le r \le s,\\
0 & \text{otherwise},
\end{cases}
\qquad
\lim_{\rho \rightarrow \infty}
(J_-(r+1)-J_-)/\rho = \begin{cases}
-D & 1 \le r \le s,\\
0 & \text{otherwise}.
\end{cases}
\end{align*}
Thus we see that $J_+(r)$ and $J_-(r+1)$ have similar asymptotics 
aside from the overall normalization for large $\rho$.
As mentioned under (\ref{jhi}) about $J_\pm$,  
the former is convergent whereas the latter is divergent
as $\rho$ tends to infinity.
They do not exhibit a peak at $r=1$ 
like the density $\rho(r)$, but 
other behavior is more or less similar.
In particular when $\rho$ is sufficiently large, 
the both currents get smaller inside the defect cluster than 
the defect-free case $J_\pm$. 

\subsection{Comparison with numerical 
evaluation in canonical ensemble}\label{ss:ce}

One can evaluate
the density numerically 
by the matrix product formula (\ref{sin})
in a {\em fixed} sector with system size $L$.
It formally corresponds to the {\em canonical} ensemble average
which will be denoted by $\rho^{\rm c}_L(r)$.
In the actual calculation of $\rho^{\rm c}_L(r)$, 
we have truncated the operators (\ref{cie}) and (\ref{Adef})
to the matrices acting on the finite dimensional subspace of $F$ of the form
$\bigoplus_{m=0}^{d_1+\cdots+ d_s+t}\R |m\rangle$ and 
checked the convergence has been achieved sufficiently already for $t=1,2$, etc. 
Figure \ref{nc12} compares $\rho(r)$ and 
$\rho^{\rm c}_L(r)$ for $1 \le r \le L$.

\begin{figure}[H]
\begin{tabular}{cc}
\begin{minipage}[t]{0.45\hsize}
\begin{center}
\includegraphics[scale=0.74]{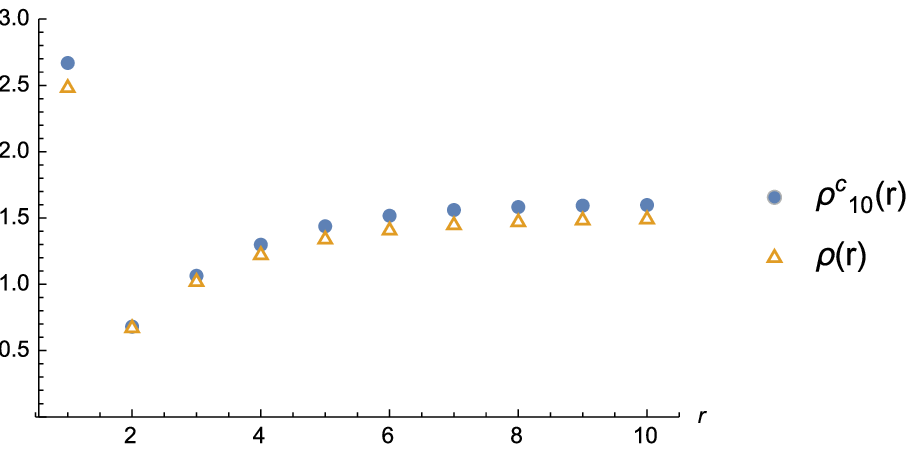}
\end{center}
\end{minipage}
&\;\;
\begin{minipage}[t]{0.45\hsize}
\begin{center}
\includegraphics[scale=0.74]{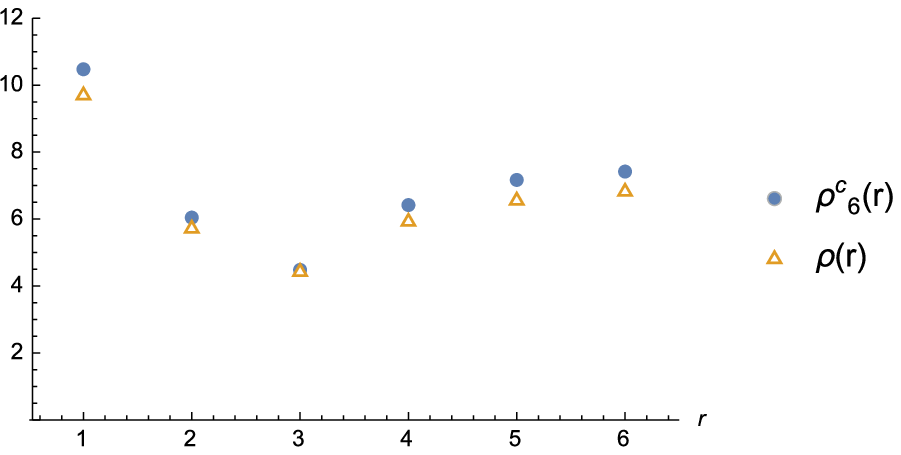}
\end{center}
\end{minipage}
\end{tabular}
\caption{Left: Single defect case $d_1=1$ with 
$(\rho,q,\mu)=(1.5, 0.2,0.7)$. 
One has 
$\rho^{\rm c}_L(1)=2.71394, 2.66876$ for $L=8,10$ approaching 
$\rho(1) = 2.4801$.
The sector for computing $\rho^{\rm c}_{10}(r)$ consists of 
$\binom{24}{15}=1307504$ states.
Right: The system with three defect particles $(d_1,d_2) = (2,1)$ with 
$(\rho, q,\mu)=(7,0.2,0.8)$. 
The sector for computing $\rho^{\rm c}_{6}(r)$ consists of 
$\binom{47}{6}=8145060$ states.}
\label{nc12}
\end{figure}

As remarked after (\ref{rinf}), 
the correlation length of the system 
is smaller for larger average density $\rho$.
Thus the agreement of $\rho(r)$ and $\rho^{\rm c}_L(r)$
is expected to be better for larger $\rho$, therefore is harder 
to observe in numerical calculations. 
Admittedly the agreement in Figure \ref{nc12}
is not quite excellent, but always exhibits the tendency to 
improve as $L$ gets large.
In general fluctuations in the density is of order 
$L^{-\frac{1}{2}}$.
So the grand canonical approach 
should coincide with the canonical one in the large volume limit
for there is no symptom of phase transition in the range $0 < q, \mu < 1$.
See the remarks following (\ref{ry}).

\subsection{Density and current profiles}\label{ss:dc}
Here we present the result of numerical evaluation of 
the formulas in Theorem \ref{th:naka}, \ref{th:migi} and \ref{th:hidari}
in a number of figures.
We first consider the profile of local density $\rho(r)$ 
(\ref{N2}), (\ref{N1}) and (\ref{N3})
in the presence of various defects
in Figure \ref{n4keta}--\ref{n1keta}.
In general the density can possibly break 
the monotonicity inside the defect cluster 
depending on the inhomogeneity of $d_i$'s. 

\begin{figure}[H]
\begin{tabular}{cc}
\begin{minipage}[t]{0.45\hsize}
\begin{center}
\includegraphics[scale=0.7]{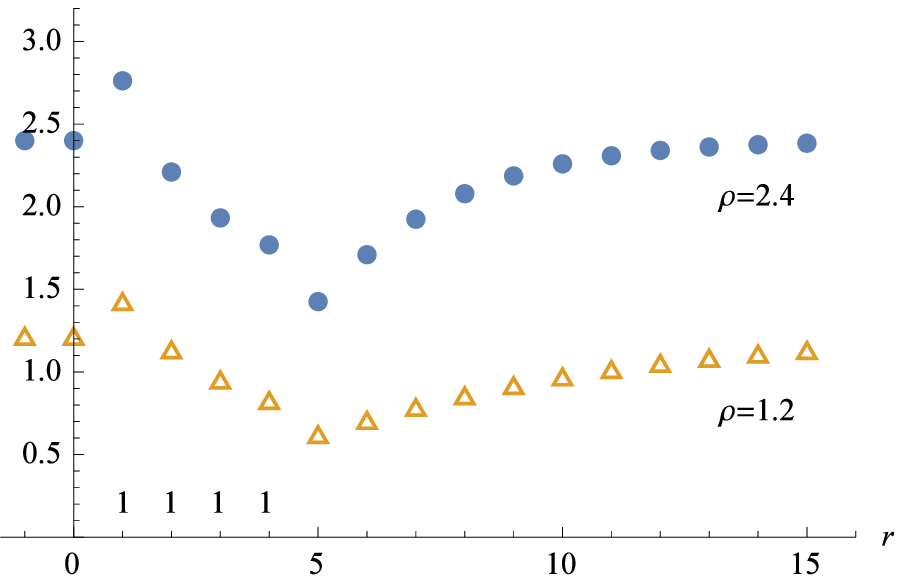}
\end{center}
\end{minipage}
&\;\;
\begin{minipage}[t]{0.45\hsize}
\begin{center}
\includegraphics[scale=0.7]{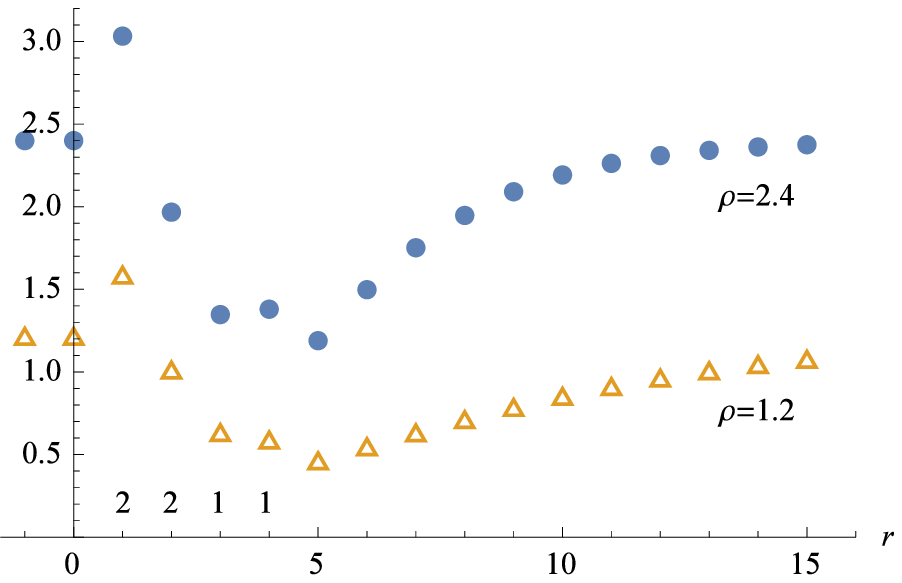}
\end{center}
\end{minipage}
\\
& \\
\begin{minipage}[t]{0.45\hsize}
\begin{center}
\includegraphics[scale=0.7]{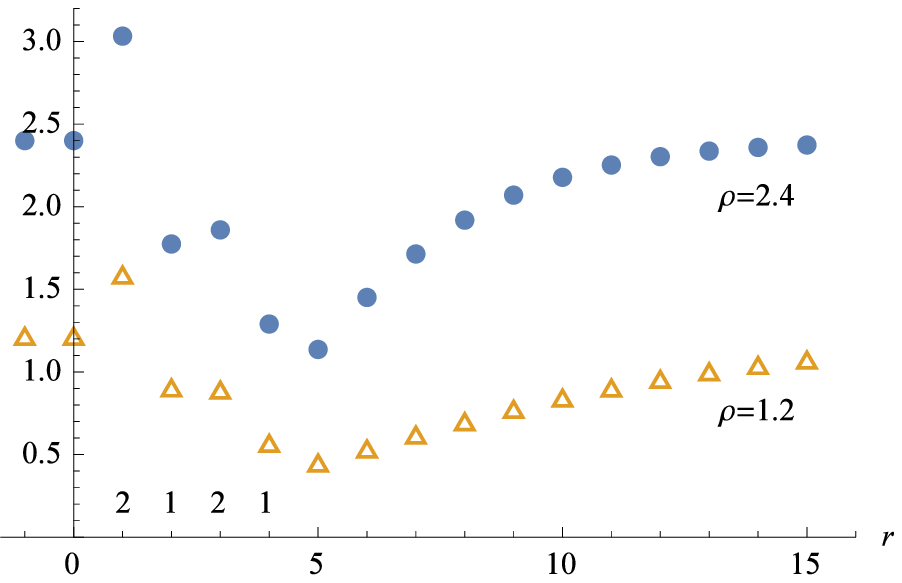}
\end{center}
\end{minipage}
&\;\;
\begin{minipage}[t]{0.45\hsize}
\begin{center}
\includegraphics[scale=0.7]{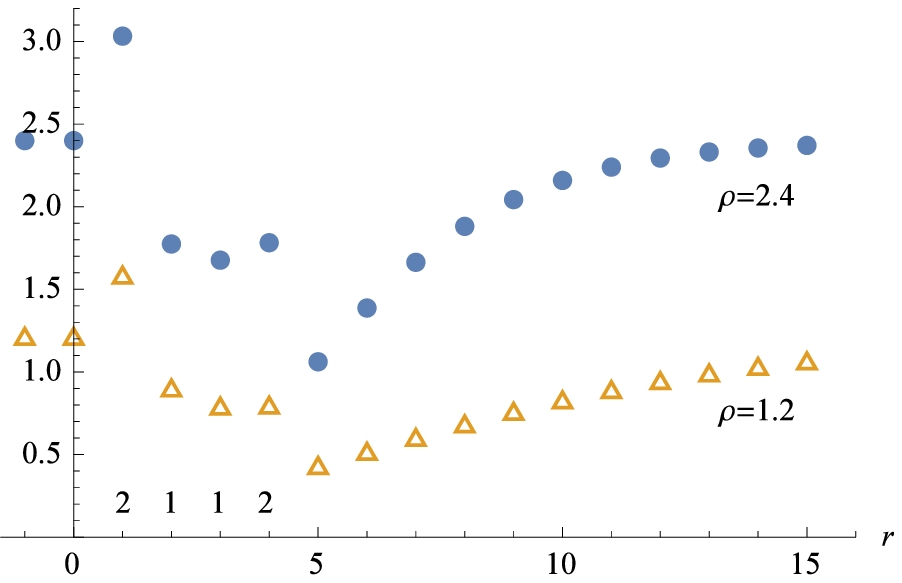}
\end{center}
\end{minipage}
\end{tabular}
\caption{Plots of $\rho(r)$ for the densities
$\rho = 1.2$ and $\rho=2.4$ with the 
presence of 
the defects $(d_1,\ldots, d_4)$ shown at sites $1,\ldots, 4$.
$(q,\mu)=(0.8,0.5)$.}
\label{n4keta}
\end{figure}

Examples of the density profile for a finite range of $\rho$
are given in Figure \ref{n2keta}.
To display the defect region I in the center, 
the lattice coordinate $r$ has been shifted.
A similar convention will be employed in the subsequent figures 
in this subsection except Figure \ref{j4keta}. 

\begin{figure}[H]
\begin{tabular}{cc}
\begin{minipage}[t]{0.45\hsize}
\begin{center}
\includegraphics[scale=0.6]{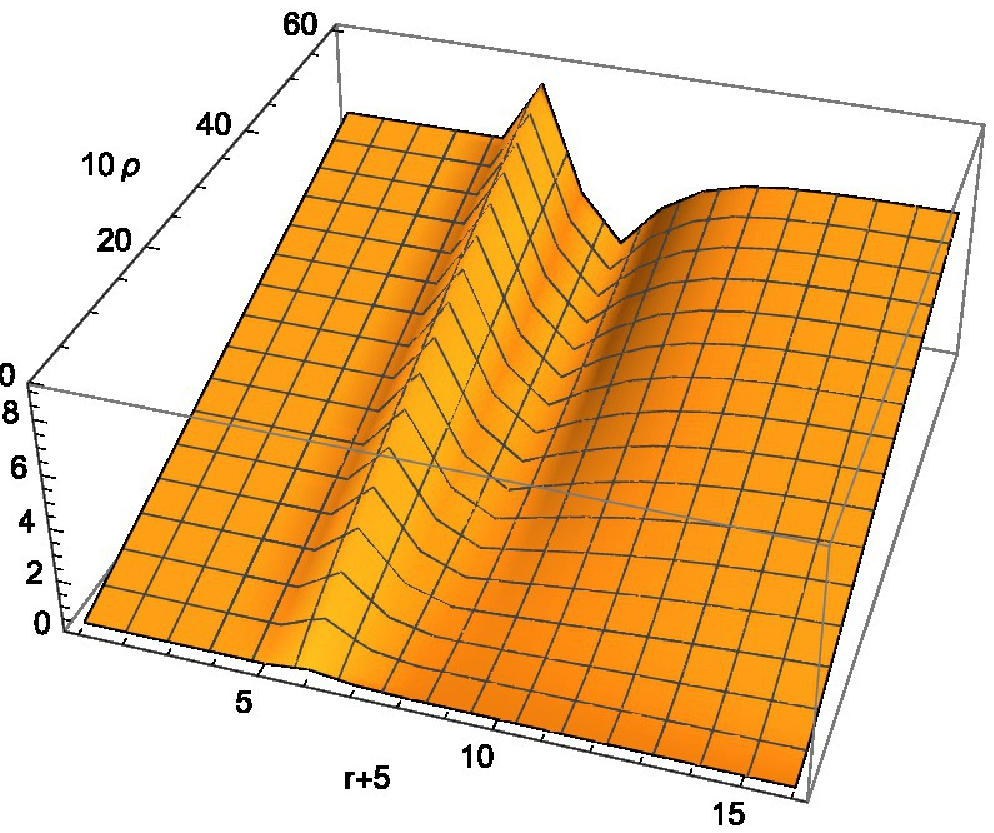}
\end{center}
\end{minipage}
&\;\;
\begin{minipage}[t]{0.45\hsize}
\begin{center}
\includegraphics[scale=0.6]{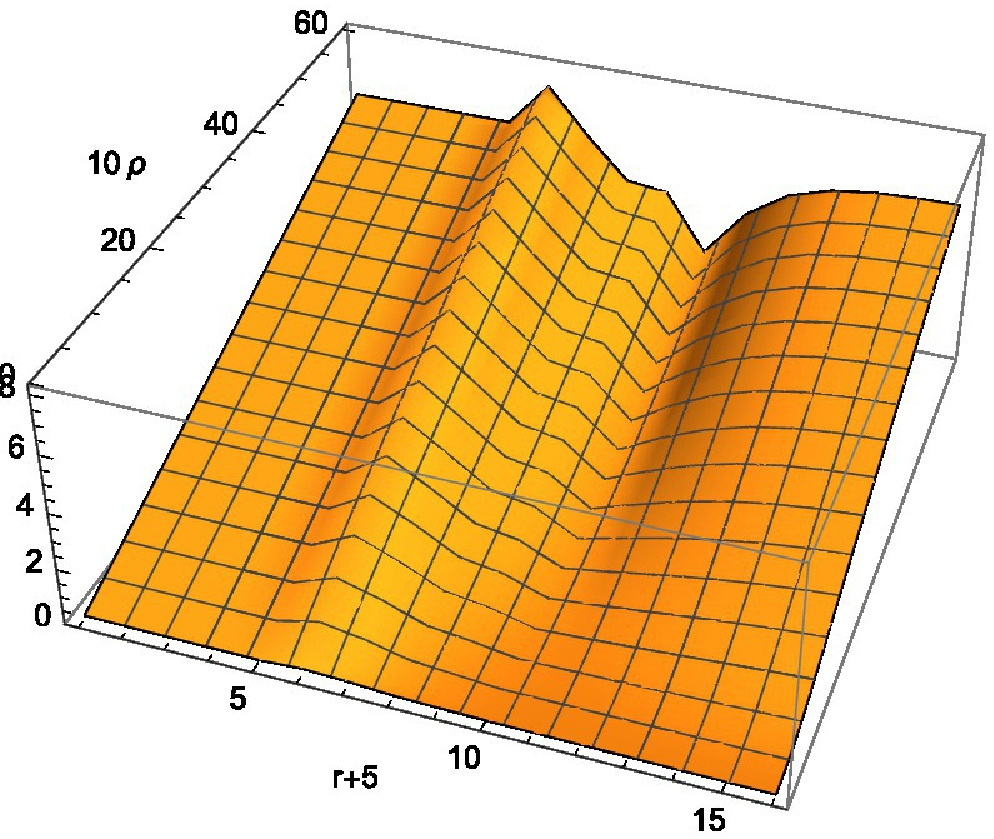}
\end{center}
\end{minipage}
\end{tabular}
\caption{Density profile $\rho(r)$ with $(q,\mu)=(0.6,0.7)$
for $0<\rho<6$.
The defects are $(d_1,d_2) = (3,3)$ (left) and 
$(d_1,\cdots,d_4) = (1,2,2,3)$ (right).}
\label{n2keta}
\end{figure}

For a fixed average density $\rho$,  
dependence of $\rho(r)$ on $q$ and $\mu$ are shown in Figure \ref{n1keta}.

\begin{figure}[H]
\begin{tabular}{cc}
\begin{minipage}[t]{0.45\hsize}
\begin{center}
\includegraphics[scale=0.6]{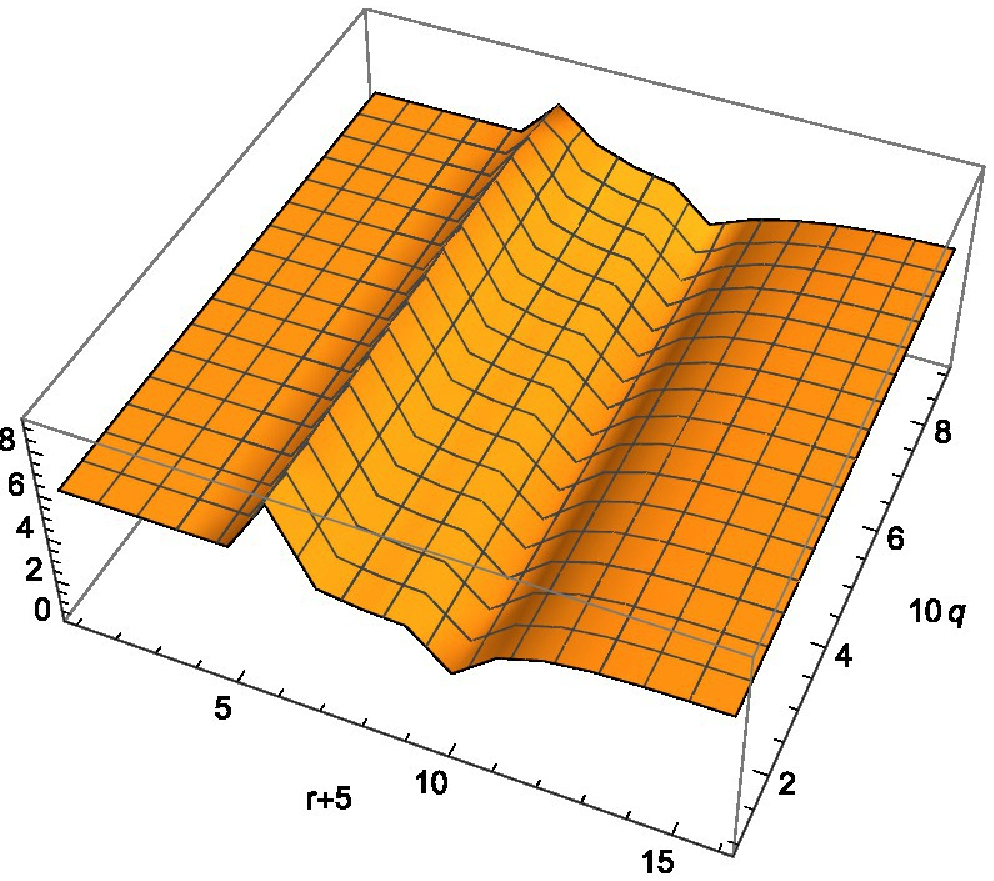}
\end{center}
\end{minipage}
&\;\;
\begin{minipage}[t]{0.45\hsize}
\begin{center}
\includegraphics[scale=0.6]{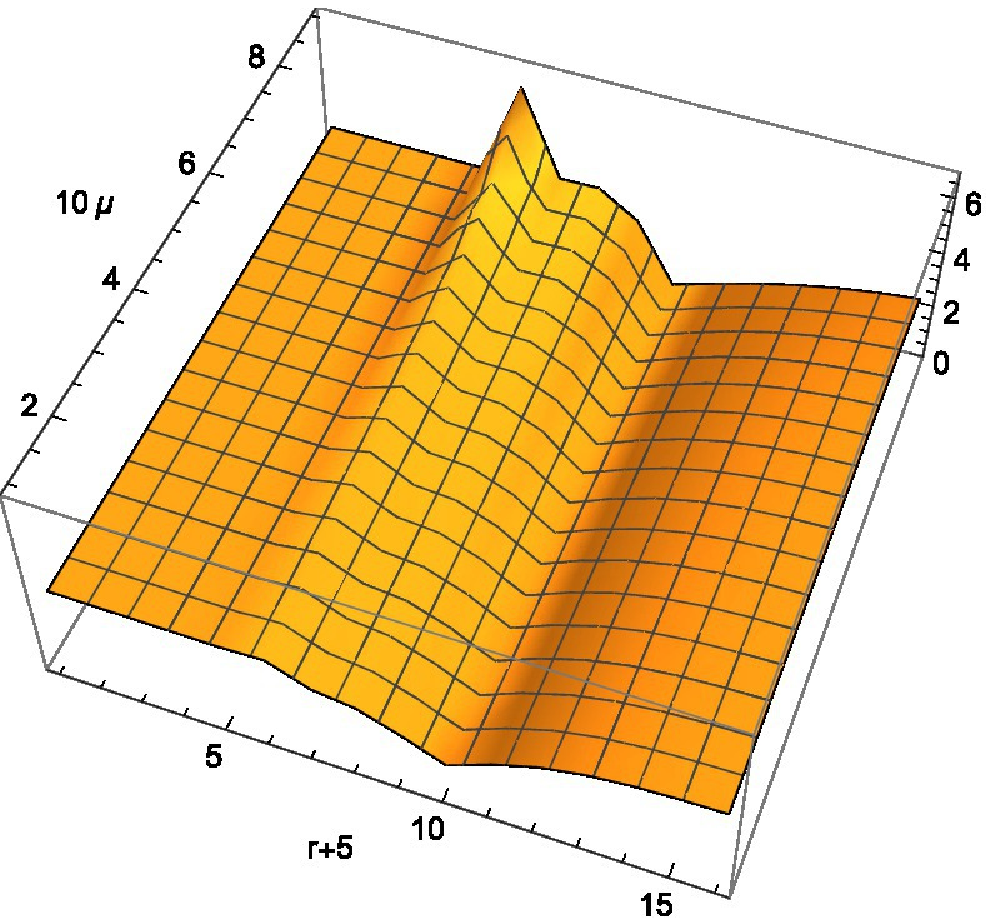}
\end{center}
\end{minipage}
\end{tabular}
\caption{Left:  $\rho(r)$ for $0.1 \le q \le 0.9$ with $(\rho,\mu)=(6,0.8)$
and defects $(d_1,\ldots, d_4) = (1,1,1,1)$.
Right:  $\rho(r)$ for $0.1 \le \mu \le 0.9$ with $(\rho,q)=(3,0.7)$
and defects $(d_1,\ldots, d_4) = (1,1,2,3)$.}
\label{n1keta}
\end{figure}

One sees that the inhomogeneous defects make the 
density profile in the region I irregular, but they  
still tend to produce a peak and a valley at their boundaries
as observed in the homogeneous case in Figure \ref{sch}.

Let us turn to the local currents $J(r)_{\pm}$ (\ref{jr}).
Recall that they have been determined as 
$J(r)_{\pm} = J_{\rm I}(r)_{\pm}$ (\ref{Jf21}), (\ref{Jf22})
for $1 \le r \le s$ and 
$J(r)_{\pm} = J_{\rm II}(r)_{\pm}$ (\ref{Jf1})
for $r >s$ and 
$J(r)_{\pm} = J_\pm$ (\ref{N3}) , (\ref{J0}) 
for $r \le 0$.
In view of (\ref{fk}) we plot $J_+(r)$ and $J_-(r+1)$.

\begin{figure}[H]
\begin{tabular}{cc}
\begin{minipage}[t]{0.45\hsize}
\begin{center}
\includegraphics[scale=0.65]{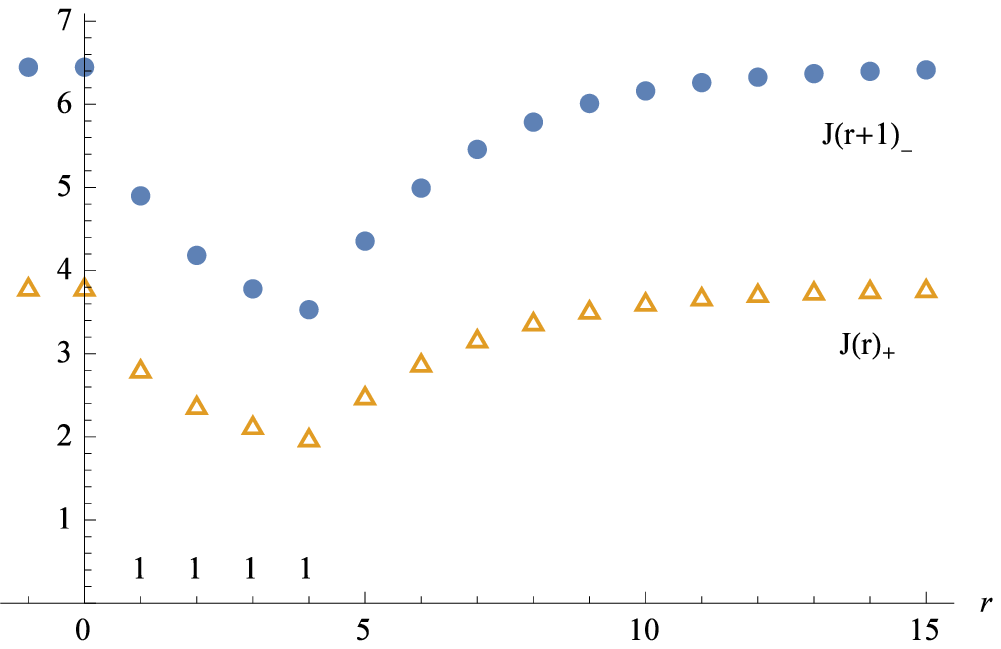}
\end{center}
\end{minipage}
&\;\;
\begin{minipage}[t]{0.45\hsize}
\begin{center}
\includegraphics[scale=0.65]{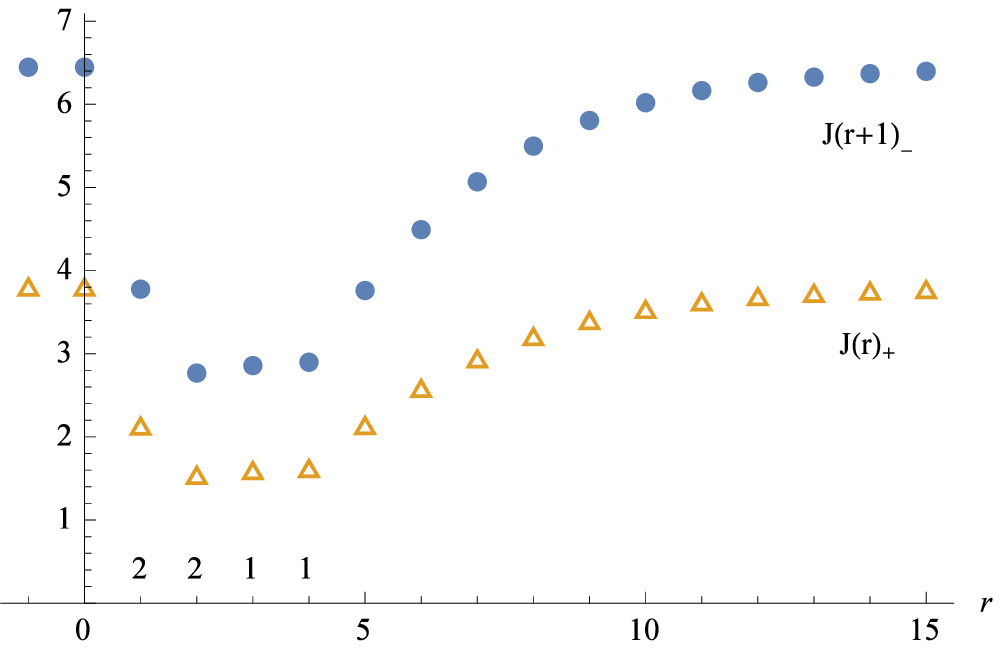}
\end{center}
\end{minipage}
\\
& \\
\begin{minipage}[t]{0.45\hsize}
\begin{center}
\includegraphics[scale=0.65]{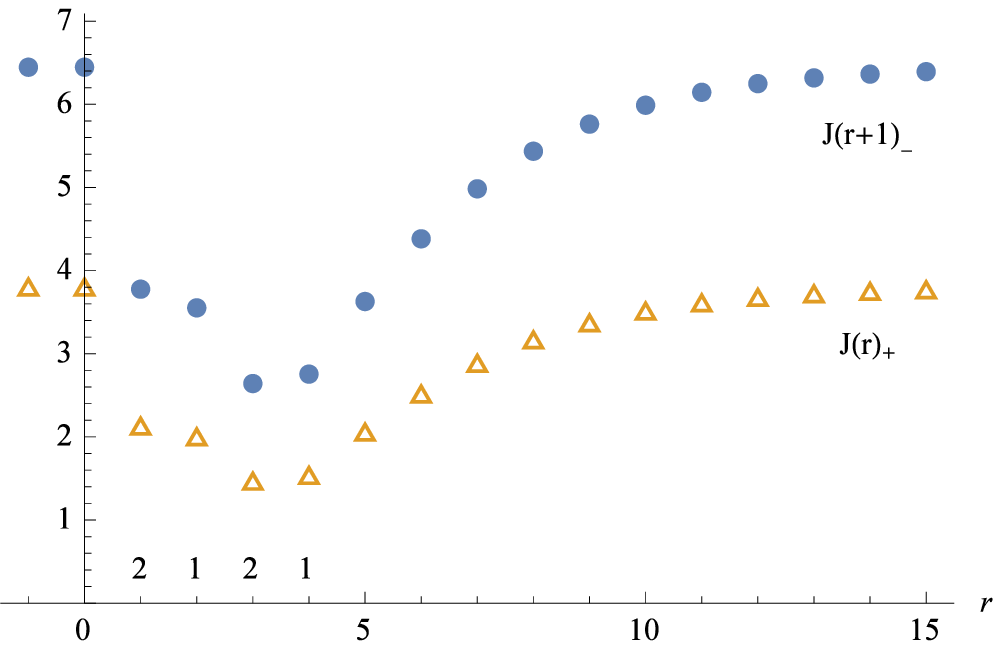}
\end{center}
\end{minipage}
&\;\;
\begin{minipage}[t]{0.45\hsize}
\begin{center}
\includegraphics[scale=0.65]{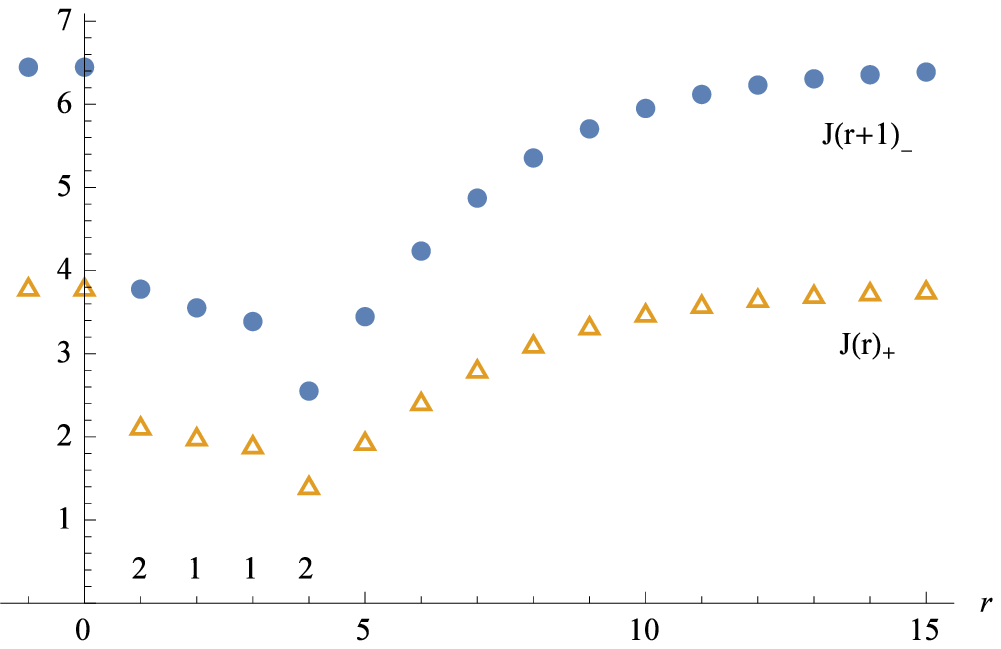}
\end{center}
\end{minipage}
\end{tabular}
\caption{Comparison of $J(r)_+$ and $J(r+1)_-$ in (\ref{fk})
for systems with the same defects as Figure \ref{n4keta}.
$(\rho, q,\mu)=(2.4, 0.8,0.5)$. 
}
\label{j4keta}
\end{figure}

\begin{figure}[H]
\begin{tabular}{cc}
\begin{minipage}[t]{0.45\hsize}
\begin{center}
\includegraphics[scale=0.6]{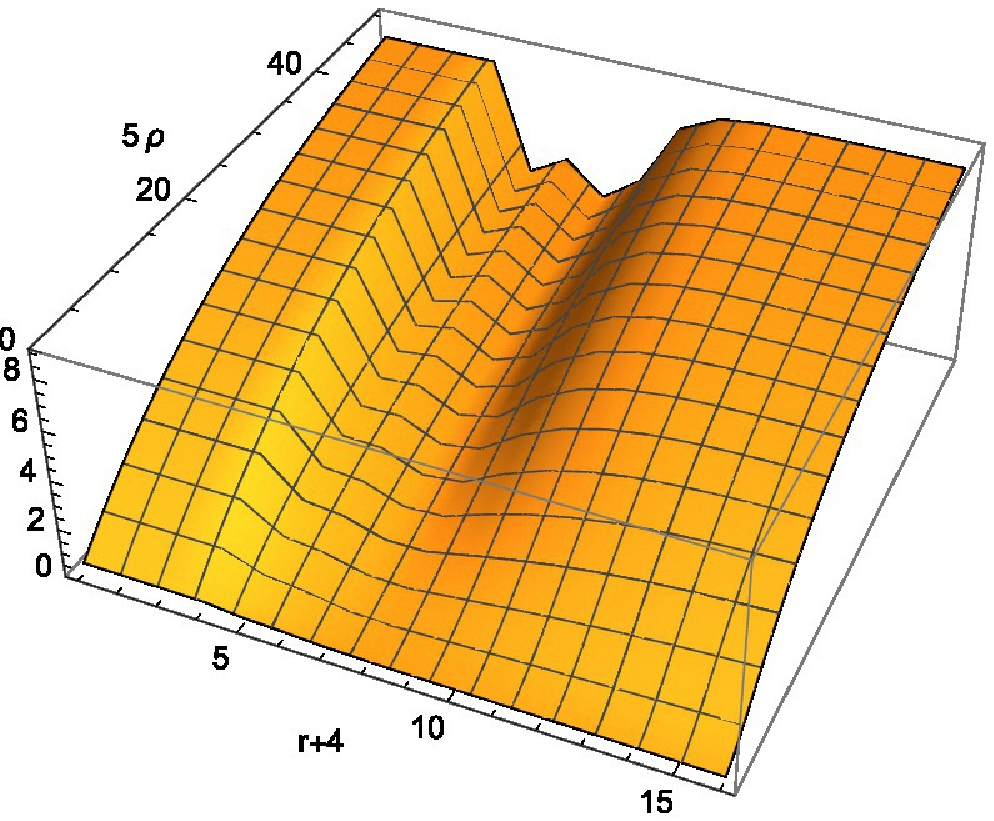}
\end{center}
\end{minipage}
&\;\;
\begin{minipage}[t]{0.45\hsize}
\begin{center}
\includegraphics[scale=0.6]{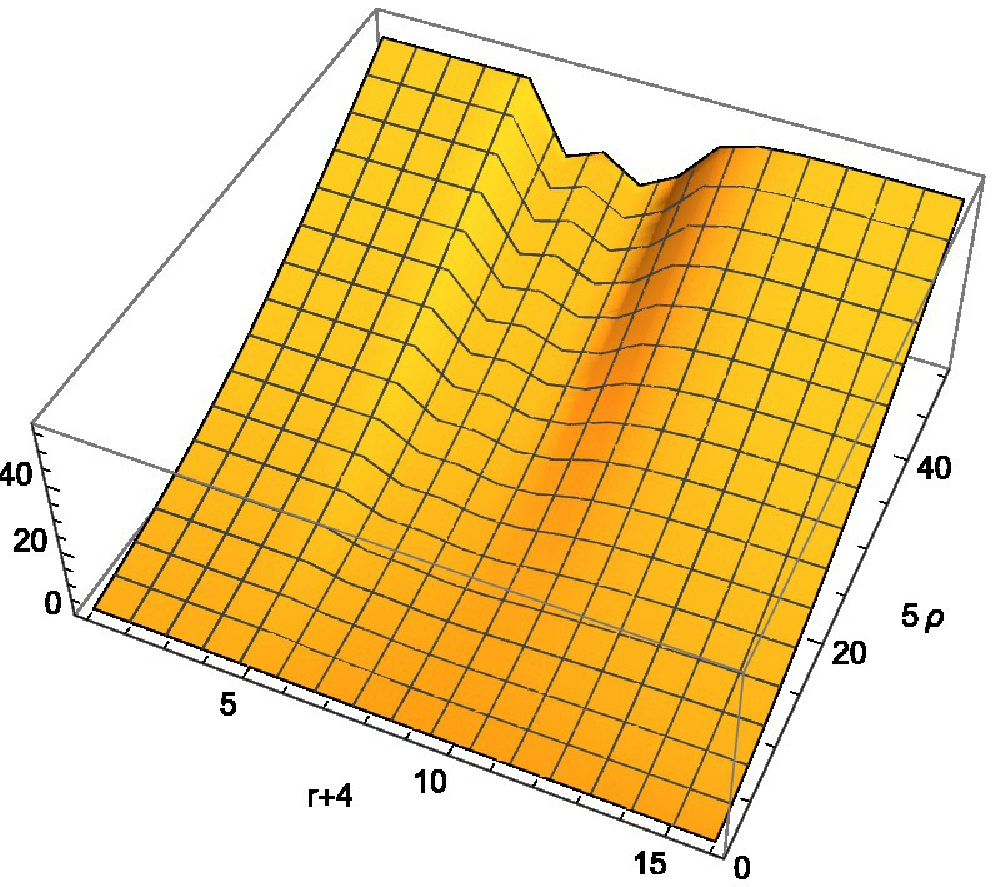}
\end{center}
\end{minipage}
\end{tabular}
\caption{Current profile $J(r)_+$ (left) and 
$J(r)_-$ (right) for $0<\rho < 10$
and $(q,\mu)=(0.8,0.5)$.
The defects are $(d_1,\ldots, d_4)=(2,1,2,1)$, which is the same as
the bottom left case of Figure \ref{n4keta} and \ref{j4keta}.
Apart from the inhomogeneity caused by the defects, 
their dependence on $\rho$ reflects the behavior 
in the defect-free case (\ref{jlow}) and (\ref{jhi}). }
\label{jpm4keta}
\end{figure}

\begin{figure}[H]
\begin{tabular}{cc}
\begin{minipage}[t]{0.45\hsize}
\begin{center}
\includegraphics[scale=0.6]{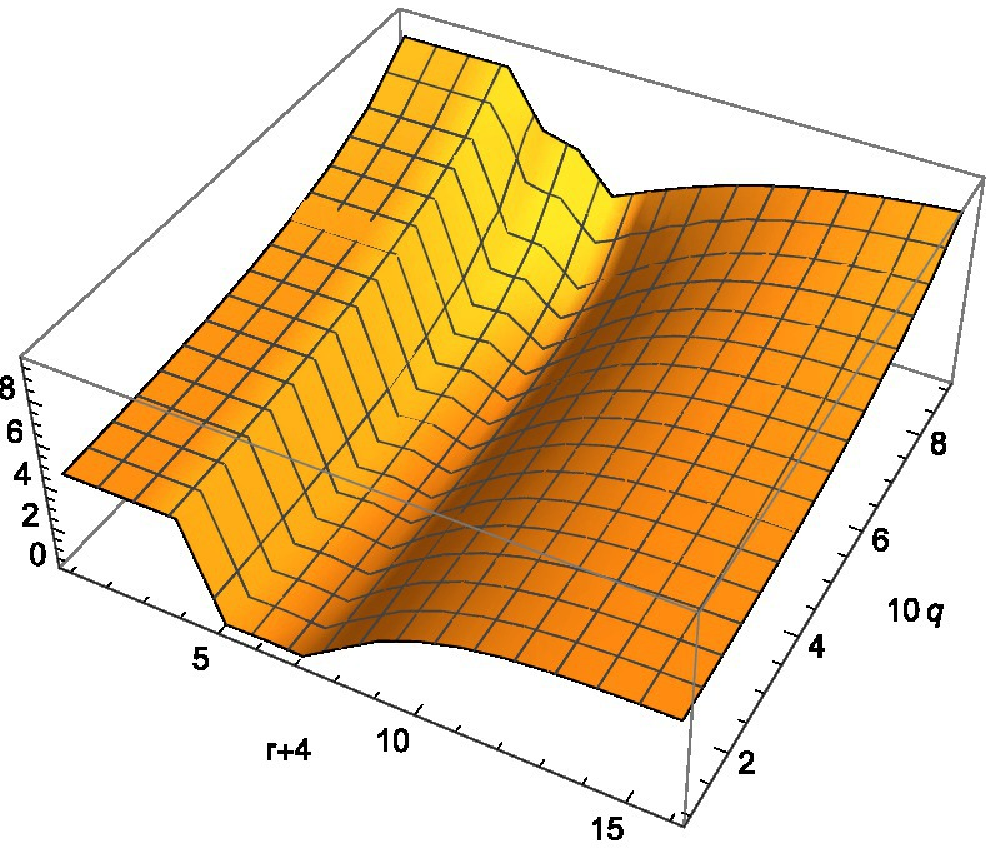}
\end{center}
\end{minipage}
&\;\;
\begin{minipage}[t]{0.45\hsize}
\begin{center}
\includegraphics[scale=0.6]{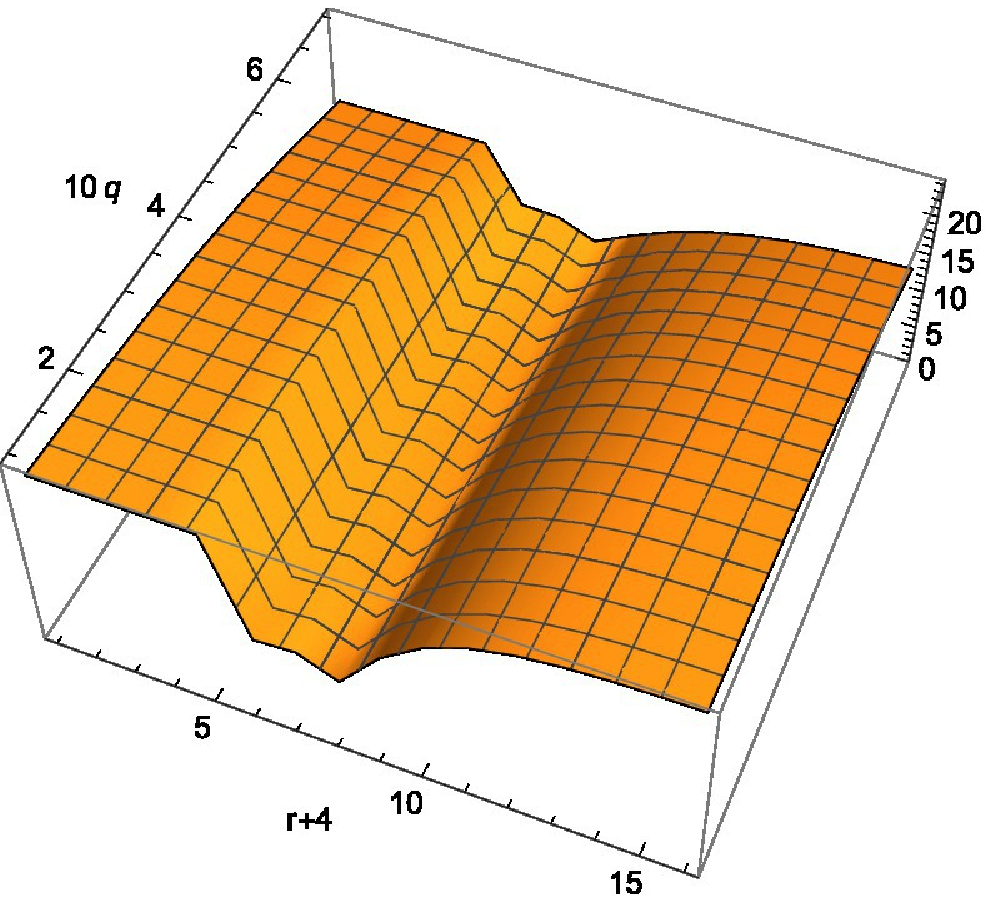}
\end{center}
\end{minipage}
\end{tabular}
\caption{Current profile $J(r)_+$ (left) and 
$J(r)_-$ (right) for $0.1 \le q  \le 0.9$,
$(\rho,\mu)=(3,0.7)$ and the defects 
$(d_1, d_2, d_3)=(2,1,3)$.}
\label{jmp213}
\end{figure}

\begin{figure}[H]
\begin{tabular}{cc}
\begin{minipage}[t]{0.45\hsize}
\begin{center}
\includegraphics[scale=0.6]{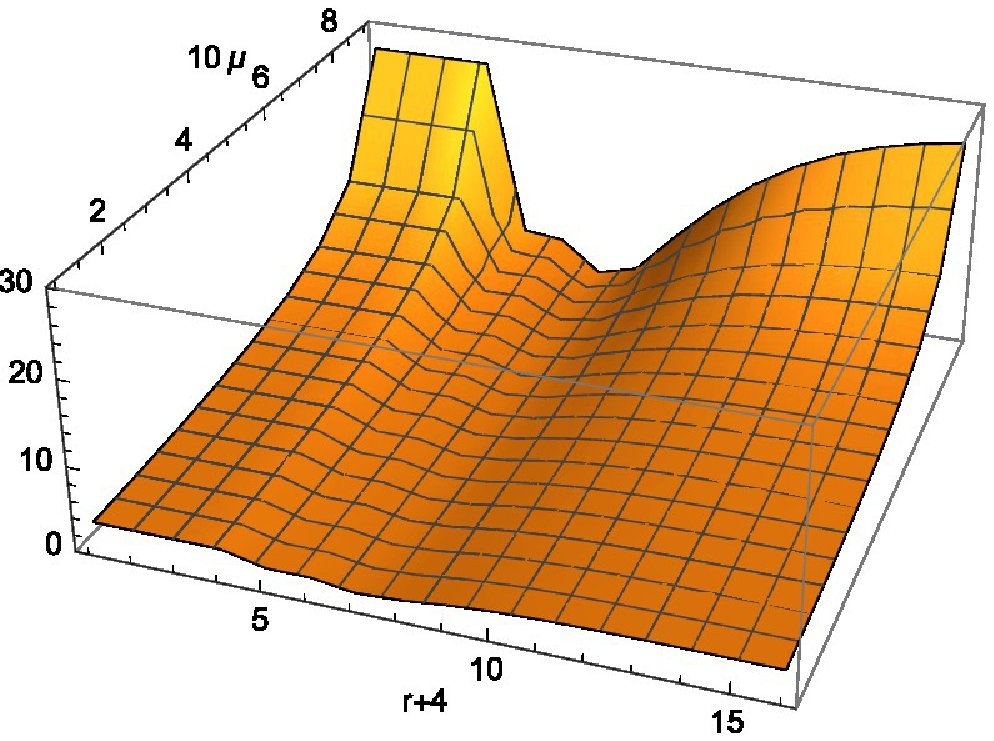}
\end{center}
\end{minipage}
&\;\;
\begin{minipage}[t]{0.45\hsize}
\begin{center}
\includegraphics[scale=0.6]{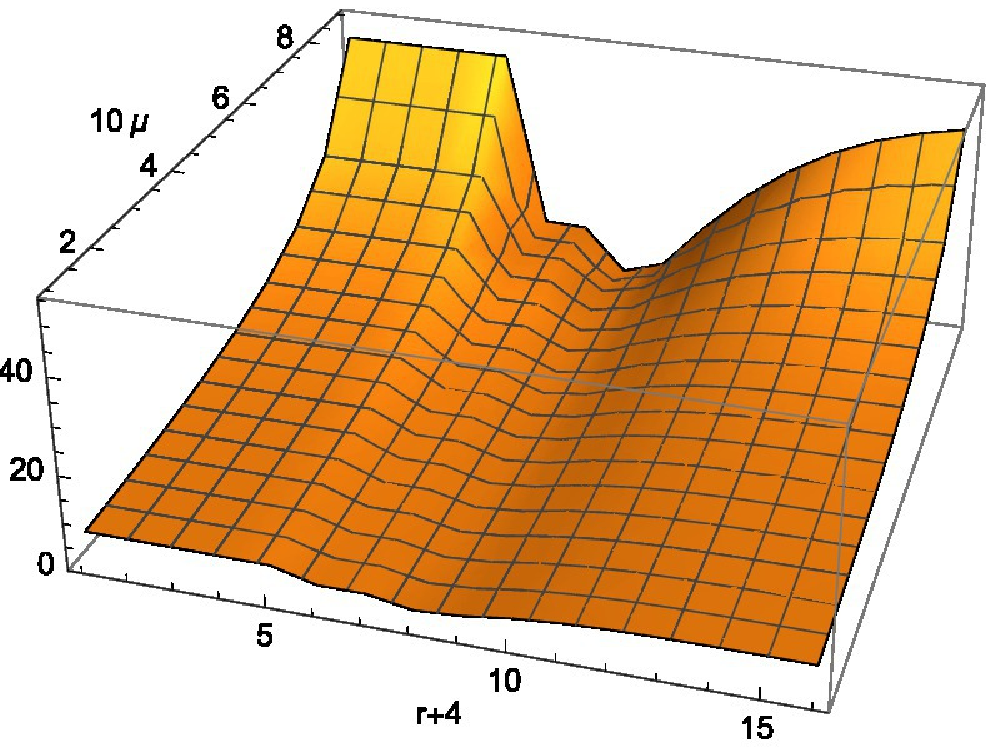}
\end{center}
\end{minipage}
\end{tabular}
\caption{Current profile $J(r)_+$ (left) and 
$J(r)_-$ (right) for $0.1 \le \mu  \le 0.9$,
$(\rho,q)=(4,0.8)$ and the defects 
$(d_1, \ldots, d_4)=(2,1,3,1)$.}
\label{jmp2131}
\end{figure}

In Figure \ref{n4keta} and \ref{j4keta},
one observes that 
$\rho(r)$ reaches its bottom at $r=s+1$ whereas 
$J_+(r)$ and $J_-(r+1)$ do not.

\section{Derivation of main results}\label{sec6}

In this section we derive the main results given in
Theorem \ref{th:naka}, \ref{th:migi}, \ref{th:hidari}
and Proposition \ref{pr:excess}.

\subsection{Preliminary}
Our first step is to reduce the trace 
$\mathrm{Tr}(\cdots)'$ over the Fock space 
to the ``vacuum expectation value" $\langle 0 | (\cdots) | 0\rangle$ 
in the infinite volume limit.  

\begin{proposition}\label{pr:red} 
The traces in the probabilities (\ref{Pmiddle})--(\ref{Pleft}) 
are reduced to the following:
\begin{align}
P_{\rm I}(r,n)&= \lim_{L \rightarrow \infty}
\frac{y^n\langle 0|A_{d_1}\cdots A_{d_{r-1}}X_{d_r,n}
A_{d_{r+1}}\cdots A_{d_s}
A_0^{L-s}|0\rangle}
{\langle 0|A_{d_1}\cdots A_{d_s} 
A_0^{L-s} |0\rangle},
\label{Pm}\\
P_{\rm II}(r,n) &= \lim_{L \rightarrow \infty}
\frac{y^n\langle 0|A_{d_1}\cdots A_{d_s} 
A_0^{r-s-1}X_{0,n}A_0^{L-r}|0\rangle}
{\langle 0|A_{d_1}\cdots A_{d_s} 
A_0^{L-s} |0\rangle},
\label{Pr}\\
P_{\rm III}(r,n)&= \lim_{L \rightarrow \infty}
\frac{y^n\langle 0|X_{0,n}A_0^{|r|} A_{d_1}\cdots A_{d_s} 
A_0^{L-|r|-s-1}|0\rangle}
{\langle 0|A_{d_1}\cdots A_{d_s} 
A_0^{L-s}|0\rangle}.
\label{Pl}
\end{align}
\end{proposition}
This is a corollary of Lemma \ref{le:zer}
in Appendix \ref{app:ls}.

Let us prepare the functions 
that will serve as building blocks to express 
(\ref{Pm}) and (\ref{Pr}).
For $l,m \in \Z_{\ge 0}$ set
\begin{align}
\phi(l|m) &= y^l \frac{(\mu)_l(y)_{m-l}}{(\mu y)_m}\binom{m}{l}_q
= \Phi_q(l|m;\mu, \mu y)|_{n=1},
\label{phidef}
\end{align}
where $\Phi_q$ was defined in (\ref{mho}).
Although the dependence on $\mu$ and $y$ is suppressed in this notation, 
it deserves attention that the fugacity $y$ 
plays the role of a spectral parameter here. 
From (\ref{syk}) we know
\begin{align}
&\sum_{l \ge 0}\phi(l|m)=1,
\label{utk}
\end{align}
where the summand is nonzero only for $0 \le l \le m$.

\begin{lemma}\label{le:ma}
For any $m \in \Z_{\ge 0}$ the following equality is valid:
\begin{align*}
\sum_{j=1}^m\phi(m-j|m)\sum_{i=0}^{j-1}\frac{1}{1-y q^i}
= \sum_{k=0}^{m-1}\frac{1}{1-\mu y q^k}.
\end{align*}
\end{lemma}

\begin{proof}
Explicitly it reads as
\begin{align}\label{mor1}
\sum_{j=0}^{m-1}
y^{j}\frac{(\mu)_{j}(y)_{m-j}}{(\mu y)_m}\binom{m}{j}_q
\sum_{i=0}^{m-j-1}\frac{1}{1-y q^i}
= \sum_{k=0}^{m-1}\frac{1}{1-\mu y q^k},
\end{align}
where we have replaced $j$ by $m-j$.
We prove (\ref{mor1}) by induction on $m$.
The case $m=0$ is obvious.
In what follows we assume (\ref{mor1}) is valid when $m$ is replaced by 
$0,1,\ldots, m-1$.

Upon multiplication by $(\mu y)_m$, 
the both sides of (\ref{mor1}) become polynomials in $\mu$ of order $m-1$.
Thus it suffices to check the equality at $m$ points, say $\mu=q^{-r}$ with 
$r=0,1,\ldots, m-1$.
So we set $\mu=q^{-r}$ in (\ref{mor1}) 
and further replace $y$ by $q^ry$ to simplify the formula slightly.
The result reads
\begin{align}\label{mor2}
\sum_{j=0}^r q^{jr} y^j \frac{(q^{-r})_j(q^ry)_{m-j}}{(y)_m}\binom{m}{j}_q
\Bigl(A + 
\sum_{i=0}^{r-j-1}
\frac{1}{1-y q^{m+i}}\Bigr)
= \sum_{k=0}^{m-1}\frac{1}{1- y q^k},
\end{align}
where $A = \sum_{i=r}^{m-1}\frac{1}{1-y q^i}$ and 
the upper bound of the $j$ sum has been reduced from $m-1$ to $r$ 
owing to the factor $(q^{-r})_j$. 
The coefficient of $A$ is $1$ due to (\ref{utk})
with $(\mu,y)$ replaced by $(q^{-r}, q^r y)$.
Let us subtract $A$ from (\ref{mor2}). 
The upper bound of the $k$ sum in the RHS becomes $r-1$. 
The second sum in the LHS restricts the upper bound of $j$ to $r-1$.
Further applying 
\begin{align*}
\frac{(q^r y)_{m-j}}{(y)_m} = \frac{(q^m y)_{r-j}}{(y)_r},\qquad
(q^{-r})_j\binom{m}{j}_q = q^{(m-r)j}(q^{-m})_j\binom{r}{j}_q
\end{align*}
in the process, we find that the result is equivalent to
\begin{align*}
\sum_{j=0}^{r-1} q^{mj}y^j
\frac{(q^{-m})_j(q^my)_{r-j}}{(y)_r}\binom{r}{j}_q
\sum_{i=0}^{r-j-1}
\frac{1}{1-y q^{m+i}} = \sum_{k=0}^{r-1}\frac{1}{1-y q^k}.
\end{align*}
This coincides with (\ref{mor1}) with $(m,\mu,y)$ replaced by 
$(r,q^{-m}, q^m y)$.
Since $r \le m-1$, its validity is assured by the induction hypothesis. 
\end{proof}

We note that the limit $y\rightarrow 1$ in Lemma \ref{le:ma} leads to 
the identity  
\begin{align}\label{kzn}
\frac{(q)_{d}}{(\mu)_{d}}
\sum_{k=1}^d \frac{1}{1-q^k}\frac{(\mu)_{d-k}}{(q)_{d-k}}
= \sum_{k=0}^{d-1}\frac{1}{1-\mu q^k}\quad (d \ge 0).
\end{align}

The following function plays the basic role in our working.
\begin{align}
G_{m,l}(d_1,\ldots, d_s) &=
\sum_{l_1+ \cdots + l_s + l\atop
= d_1+\cdots +d_s+m}\prod_{i=1}^s
\phi(l_i| m+d_1+\cdots + d_{i-1}-l_1-\cdots - l_{i-1})
\label{Gdef},
\end{align}
where $l,m \in \Z_{\ge 0}$.
The sum in (\ref{Gdef}) extends over 
$l_1, \ldots, l_s \in \Z_{\ge 0}$ obeying the specified condition. 
Since $\phi(l|m)=0$ unless $0 \le l \le m$, 
nonzero summands in (\ref{Gdef}) are only those satisfying
$l_1+\cdots + l_i \le m+d_1+\cdots d_{i-1}$ for 
$i=1,\ldots, s$.
The definition (\ref{Gdef}) is depicted as
\begin{equation}\label{Gm}
\begin{picture}(180,85)(-130,-7)

\put(-140,35){$G_{m,l}(d_1,\ldots, d_s)
= {\displaystyle \sum_{l_1,\ldots, l_s}}$}

\put(0,2.5){
\put(0,44){\vector(0,1){20}}
\put(-1.8,68){$l$}
\put(-21.5,49.4){$d_s$}\put(14,49){$l_s$}}

\multiput(-1.2,37.7)(0,3){3}{\put(0,0){.}}

\put(0,0){\line(0,1){35}}
\put(-10,55){\vector(1,0){20}}
\put(-10,25){\vector(1,0){20}}
\put(-10,10){\vector(1,0){20}}

\put(-21.5,21.4){$d_2$}\put(14,21){$l_2$}
\put(-21.5,6.4){$d_1$}\put(14,6){$l_1$}
\put(-3.8,-8.5){$m$}
\end{picture}
\end{equation}
if each vertex is interpreted 
as an element of the $n=1$ stochastic $R$ matrix
according to (\ref{phidef}), (\ref{ask2}) and (\ref{vertex}).
According to this diagram, the
$G_{m,l}(d_1,\ldots, d_s)$ may be regarded a sum of elements of 
a {\em column monodromy matrix} 
of the $U_q(A^{(1)}_1)$ ZRP containing the fugacity $y$ 
as a spectral parameter via (\ref{phidef}).
Some special cases which will be of frequent use are
\begin{align}
&G_{m,l}(\emptyset) = \delta_{m,l},\quad
G_{m,l}(d_1) = \phi(m+d_1-l|m),
\quad G_{0,l}(d_1)= \delta_{l, d_1},
\label{Gex1}\\
&G_{0,l}(d_1,\ldots, d_s) =
\sum_{l_1+ \cdots + l_{s-1} + l\atop
= d_1+\cdots +d_s}\prod_{i=1}^{s-1}
\phi(l_i|d_1+\cdots + d_i-l_1-\cdots - l_{i-1})\quad (s\ge 2),
\label{Gex2}
\end{align}
where the sum in (\ref{Gex2}) extends over 
$l_1, \ldots, l_{s-1} \in \Z_{\ge 0}$ under  the specified condition. 
It is easy to see 
\begin{align}
&G_{m,l}(d_1,\ldots, d_s) = 0  \;\;\text{unless}\;\;
d_s \le l \le d_1+\cdots + d_s + m,
\label{Gz}\\
&\sum_{l \ge 0}G_{m,l}(d_1,\ldots, d_s)=1,
\label{G1}\\
&G_{m,l}(d_1,\ldots, d_s) = 
\sum_{k \ge 0}
G_{m,k}(d_1,\ldots, d_t)G_{k,l}(d_{t+1},\ldots, d_s)\quad
(1 \le t <s),
\label{Grec}\\
&G_{m,l}(\,\overset{u}{\overbrace{0,\ldots,0}}\,)
= \eta_l^u q^{l(l-m)}\binom{m}{l}_q 
\Bigl(G_{m-l,0}(\,\overset{u}{\overbrace{0,\ldots,0}}\,)
\left|_{y \rightarrow q^l y}\right.\Bigr),
\label{G0}\\
&\lim_{y \rightarrow 0}G_{m,l}(d_1,\ldots, d_s)
= \delta_{l,m+d_1+\cdots + d_s}, 
\qquad
\lim_{y \rightarrow 1}
G_{m,l}(d_1,\ldots, d_s) = \delta_{l, d_s},
\label{Gy}
\\
&\lim_{y \rightarrow 1}\frac{G_{m,d_s+j}(d_1,\ldots, d_s)}{1-y}
= \frac{1}{1-q^j}\frac{(q)_{d_{s-1}}(\mu)_{d_{s-1}-j}}
{(\mu)_{d_{s-1}}(q)_{d_{s-1}-j}}
\quad
\text{if }j \ge 1. 
\label{Gyy}
\end{align}
In our working there appear many sums involving 
$G_{0,m}(d_1,\ldots, d_s)$.
They always range over those $m$'s that satisfy the non-vanishing condition 
implied by (\ref{Gz}).
The sum rule (\ref{G1}) is derived by successive use of (\ref{utk}).
The relation (\ref{Grec}) follows directly from the 
diagrammatic representation (\ref{Gm}).
The sum over $k$ in it is finite due to (\ref{Gz}).
In particular when $t=1$ it implies 
$G_{0,l}(d_1,\ldots, d_s) = G_{d_1,l}(d_2,\ldots, d_s)$ 
due to (\ref{Gex1}).
The relation (\ref{G0}) is derived by applying 
$\phi(\gamma|\beta)|_{y\rightarrow q^i y}
= \phi(\gamma|\beta+i)
\frac{(q)_{\beta-\gamma+i}(q)_\beta}
{(q)_{\beta+i}(q)_{\beta-\gamma}}q^{i\gamma}\eta_i^{-1}$
to the representation (\ref{Gm}).
Combining (\ref{Gyy}) with (\ref{kzn}), we get 
\begin{align}\label{ayk}
\lim_{y \rightarrow 1} \;\sum_{j\ge 1} 
\frac{G_{0,d_s+j}(d_1,\ldots, d_s)}{1-y}
= \sum_{k=0}^{d_{s-1}-1}
\frac{1}{1- \mu q^k}.
\end{align}
This relation will be utilized to extract the large $\rho$ limits 
(\ref{rir0}),  (\ref{ktn}), (\ref{msw}) and (\ref{rna}).
We note that neither $G_{m,l}(0,d_1,\ldots,d_{s-1})$ nor
$G_{m,l}(d_1,\ldots,d_{s-1},0)$ are equal to 
$G_{m,l}(d_1,\ldots,d_{s-1})$.
The function $G_{m,l}(d_1,\ldots, d_s)$ originates 
in the following quantity representing the effect of defects.
\begin{lemma}\label{le:GF}
\begin{align*}
&\langle m | A_{d_1}\cdots A_{d_s}|l\rangle 
=y^{l-m-d_1-\cdots -d_s}g_{d_1} \cdots g_{d_s}
\Lambda(y)^s\frac{(q)_l(\mu y)_m}{(\mu y)_l}
G_{m,l}(d_1,\ldots, d_s).
\end{align*}
\end{lemma}
\begin{proof}
Substitute the expansion 
$A_{d_i} = g_{d_i}
\sum_{l_i \ge 0} g_{l_i}\bb^{l_i}
\frac{(q^{d_i}\mu y \bk)_\infty}{(y \bk)_\infty}\bc^{d_i}$ 
of (\ref{Adef}) into the LHS.
By sending all the factors 
$\frac{(q^{d_i}\mu y \bk)_\infty}{(y \bk)_\infty}$
in the bracket to the left by (\ref{akn}), we find that 
$\frac{\langle m | A_{d_1}\cdots A_{d_s}|l\rangle}
{g_{d_1}\cdots g_{d_s}}$
is expressed as
\begin{align*}
\sum_{l_1+ \cdots + l_s + l\atop
= d_1+\cdots +d_s+m} g_{l_1}\cdots g_{l_s}\prod_{i=1}^s
\frac{(q^{m+d_1+\cdots +d_i-l_1-\cdots - l_i}\mu y)_\infty}
{(q^{m+d_1+\cdots +d_{i-1}-l_1-\cdots - l_i}y)_\infty}
\langle m | \bb^{l_1}\bc^{d_1}\cdots \bb^{l_s}\bc^{d_s}|l\rangle,
\end{align*}
where the sum is over $l_1, \ldots, \l_s \in \Z_{\ge 0}$ 
and the constraint is imposed to pick the non-vanishing bracket.
From 
$\frac{(q^{M+d_i} \mu y)_\infty}{(q^M y)_\infty}
= \Lambda(y) \frac{(y)_M}{(\mu y)_{M+d_i}}$
and
\begin{align*}
\langle m | \bb^{l_1}\bc^{d_1}\cdots \bb^{l_s}\bc^{d_s}|l\rangle
=\delta^{l_1+ \cdots + l_s + l}_{d_1+\cdots +d_s+m}
\prod_{i=1}^s 
\binom{m+d_1+\cdots + d_{i-1}-l_1-\cdots - l_{i-1}}{l_i}_q
(q)_{l_i},
\end{align*}
the claimed formula follows.
\end{proof}

\begin{lemma}
\begin{align}
\lim_{L \rightarrow \infty}
\Lambda(y)^{-L} 
\langle m |A_{d_1}\cdots A_{d_s}A_0^{L-s} | 0\rangle 
= y^{-m-d_1-\cdots - d_s}g_{d_1}\cdots g_{d_s}(\mu y)_m.
\label{Ad0}
\end{align} 
\end{lemma}
\begin{proof}
First we show the $s=0$ case, i.e.,
\begin{align}
&\lim_{L \rightarrow \infty}
\Lambda(y)^{-L}\langle m | A_0^L | 0\rangle = 
y^{-m}(\mu y)_m
\label{ma0}
\end{align}
by induction on $m$.
At $m=0$, it is valid since 
$\langle 0 | A_0^L | 0\rangle = \Lambda(y)^L$ for any finite $L$.
Assume that it is valid for $0,1,\ldots, m-1$ and let $U_m$ denote the 
LHS of (\ref{ma0}) to be determined.
Then we have
\begin{align*}
U_m &= \Lambda(y)^{-1}
\sum_{l=0}^m \frac{ \langle m | A_0 | l\rangle}{(q)_l}
\lim_{L \rightarrow \infty}
\Lambda(y)^{-L+1}\langle l | A_0^{L-1} | 0\rangle
\nonumber\\
&= \frac{(y)_\infty}{(\mu y)_\infty}\sum_{l=0}^m
\frac{(\mu)_{m-l}(q)_m}{(q)_{m-l}(q)_l}
\frac{(q^l\mu y)_\infty}{(q^l y)_\infty}
\bigl(y^{-l}(\mu y)_l+ \delta_{l,m}(U_m- y^{-m}(\mu y)_m)\bigr)
\nonumber\\
&= y^{-m}(\mu y)_m\sum_{l=0}^m\phi(l|m)-y^{-m}(y)_m + 
\frac{(y)_m}{(\mu y)_m}U_m.
\end{align*}
Solving this taking (\ref{utk}) into account we find 
$U_m = y^{-m}(\mu y)_m$ establishing (\ref{ma0}).
Now we prove (\ref{Ad0}) for general $s$.
By inserting $1= \sum_{l\ge 0}\frac{|l\rangle\langle l|}{(q)_l}$
into the LHS, it becomes
\begin{align*}
\Lambda(y)^{-s}\sum_{l \ge 0}
\frac{1}{(q)_l}
\langle m | A_{d_1}\cdots A_{d_s} | l\rangle 
\lim_{L \rightarrow \infty}\Lambda(y)^{-L+s}
\langle l| A_0^{L-s} | 0\rangle.
\end{align*}
Due to Lemma \ref{le:GF} and (\ref{ma0}),
this is equal to
$y^{-m-d_1-\cdots - d_s}g_{d_1}\cdots g_{d_s}(\mu y)_m
\sum_{l \ge 0}G_{m,l}(d_1,\ldots, d_s)$.
Thus the proof is completed by (\ref{G1}).
\end{proof}

\begin{remark}\label{re:hrm}
By Lemma \ref{le:zer}, 
the result (\ref{Ad0}) with $m=0$ is equivalent to 
\begin{align}
\lim_{L \rightarrow \infty}
\Lambda(y)^{-L} 
\mathrm{Tr}\bigl(A_{d_1}\cdots A_{d_s} A_0^{L-s}\bigr)'
= y^{-d_1-\cdots - d_s}g_{d_1}\cdots g_{d_s}.
\label{Ad1}
\end{align} 
The factor $g_{d_1}\cdots g_{d_s}$ here
is proportional to the stationary probability
of the configuration $(d_1,\ldots, d_s)$ of 
the single species model \cite{P} of the first class particles on 
the length $s$ ring.
In this sense 
the effect of the second class particles in the grand canonical picture
is ``renormalized" into the extra factor $y^{-1}$ 
for the first class particles' fugacity,
reducing the 
$U_q(A^{(1)}_2)$ ZRP effectively to the 
$U_q(A^{(1)}_1)$ ZRP. 
We expect a similar recursive feature in 
the $U_q(A^{(1)}_n)$ ZRP with general $n$ \cite{KMMO}, which may be viewed
as a reminiscent of the {\em nested Bethe ansatz}.
The result (\ref{Ad1}) is valid including $d_i=0$.
It implies that separated clusters of the first class particles 
do not attract nor repel depending on their distance 
under the grand canonical background of 
the second class particles.
They feel each other only when they merge together at the same site or 
split from a common site.
\end{remark} 

The following result reduces the multiple sum 
in $G_{m,l}(d_1,\ldots, d_s)$ (\ref{Gdef})
into a {\em single} sum when $d_1=\cdots = d_s=0$ and 
elucidates the $s$-dependence explicitly.
It will be utilized in the region II to extract the large distance behavior of
the density and currents away from the defects.

\begin{lemma}\label{le:GFred}
For $0 \le i \le m$, the following formula is valid:
\begin{align*}
 G_{m,i}(\,\overset{r}{\overbrace{0,\ldots, 0}}\,)
& = \sum_{i\le j \le m}(-1)^{i+j}
q^{\frac{1}{2}i(i-1)+\frac{1}{2}j(j+1-2m)}
\binom{m}{j}_{\!\!q}\binom{j}{i}_{\!\!q}\eta_j^r.
\end{align*}
\end{lemma}

\begin{proof}
Introduce the function 
\begin{align}
F_{m,r}(y) 
= \sum_{l_1+ \cdots + l_r = m}
g_{l_1}\cdots g_{l_r}\prod_{j=1}^{r-1}
\eta_{l_{j+1}+l_{j+2}+\cdots+ l_r}\qquad
(m\ge 0, r \ge 1),
\label{Fdef}
\end{align}
where the sum is over 
$l_1, \ldots, l_r \in \Z_{\ge 0}$ under the 
specified condition.
It is related to the LHS by 
\begin{align}\label{gfq}
G_{m,i}(\,\overset{r}{\overbrace{0,\ldots, 0}}\,) =
 y^{m-i}\frac{(q)_m(\mu y)_i}{(q)_i(\mu y)_m}\eta_i^{r}F_{m-i,r}(q^i y),
\end{align}
where $F_{m-i,r}(q^i y)$ is given by (\ref{Fdef})
with $\eta_j=\eta_j(y)$ (\ref{etadef})
replaced by $\eta_j(q^i y)$.
The relation (\ref{gfq}) is easily checked by means of (\ref{G0}) and 
$\eta_{i+j}=\eta_{i+j}(y) = \eta_i\eta_j(q^i y)$.

Replacing $y,m$ and the summation variable $j$ with $q^{-i}y,m+i,j+i$ and using
\begin{equation} \label{manga1}
(y)_{m-l}=(-q^my)^{-l}q^{l(l+1)/2}\frac{(y)_m}{(q^{1-m}y^{-1})_l}
\end{equation}
with $y=q$, the formula to show becomes 
\begin{equation}\label{manga2}
F_{m,r}(y)=y^{-m}\frac{(\mu y)_m}{(q)_m}\sum_{0\le j\le m}
\frac{(q^{-m})_j}{(q)_j}\Bigl(\frac{(y)_j}{(\mu y)_j}\Bigr)^r q^j.
\end{equation}
We prove this by induction on $r$. When $r=1$, the RHS is expressed 
in terms of the $q$-hypergeometric \cite[eq.(1.2.14)]{GR} as 
$y^{-m}\frac{(\mu y)_m}{(q)_m}{}_2\phi_1
\left({q^{-m},y\atop\mu y};q, q\right)$.
Due to the formula \cite[p354(II.6)]{GR}
\begin{equation}\label{manga3}
{}_2\phi_1\left({q^{-n},a\atop c}; q, q\right)=a^n\frac{(c/a)_n}{(c)_n},
\end{equation}
we obtain $\frac{(\mu)_m}{(q)_m}$, 
which agrees with (\ref{Fdef}) with $r=1$.
We now prove the $r+1$ case assuming the $r$ case holds. 
From (\ref{Fdef})
we have
\[
F_{m,r+1}(y)=\sum_{l=0}^m
\frac{(\mu)_l(y)_{m-l}}{(q)_l(\mu y)_{m-l}}F_{m-l,r}(y).
\]
Using the induction hypothesis and \eqref{manga1}, exchanging the orders
of the summations with respect to $l$ and $j$, it can be written as
\[
F_{m,r+1}(y)=y^{-m}\frac{(y)_m}{(q)_m}\sum_{j=0}^m
\frac{(q^{-m})_j}{(q)_j}\left(\frac{(y)_j}{(\mu y)_j}\right)^rq^j
{}_2\phi_1\left({q^{-m+j},\mu\atop q^{1-m}y^{-1}};q, q\right).
\]
Using \eqref{manga3} and 
$(q^{1-m}y^{-1})_m=(-y)^{-m}q^{-m(m-1)/2}(y)_m$,
we arrive at the expression \eqref{manga2} with $r$ replaced by $r+1$, 
which finishes the proof.
\end{proof}

Lemma \ref{le:GFred} and (\ref{gfq}) with $i=0$ yield the formula
\begin{align}\label{bx}
F_{m,r}(y) = y^{-m}\frac{(\mu y)_m}{(y)_m}
\sum_{j=0}^m(-1)^{j}
q^{\frac{1}{2}j(j+1-2m)}\binom{m}{j}_{\!\!q}
\eta_j^r.
\end{align}

We remark that by the definition (\ref{Gdef}) and (\ref{phidef}) 
the LHS in Lemma \ref{le:GFred} is regular at $q=0$ and 
manifestly positive in the range $0<q,\mu,y<1$ under consideration,
whereas such features are highly nontrivial in the RHS.
In particular, individual terms therein can 
become $O(q^{-\frac{1}{2}m^2})$
which needs care in numerical evaluation for small $q$.

\subsection{Proof of Theorem \ref{th:naka}}\label{ss:naka}

Theorem \ref{th:naka} 
is concerned with the region I (\ref{reg}).
Let us derive (\ref{Pk2}) from (\ref{Pm}).
Substituting (\ref{cie}) into it and applying (\ref{Ad0}) to the 
denominator we get 
\begin{align*}
P_{\rm I}(r,n) = y^n \frac{(\mu)_{d_r+n}}{(q)_{d_r}(q)_n}
\Bigl(\prod_{i=1}^s g_{d_i}^{-1}y^{d_i}\Bigr)
\lim_{L \rightarrow \infty}
\Lambda(y)^{-L}
\langle 0 | A_{d_1} \cdots A_{d_{r-1}} 
\frac{(\mu \bb)_\infty}{(\bb)_\infty}\bk^n\bc^{d_r}
A_{d_{r+1}} \cdots A_{d_{s}} A_0^{L-s}|0\rangle.
\end{align*}
Insert $1= \sum_{i\ge 0}\frac{|i\rangle \langle i|}{(q)_i}$
at the two places so as to use 
\[
\langle l |\frac{(\mu \bb)_\infty}{(\bb)_\infty}\bk^n\bc^{d_r}|m \rangle
= g_{l-m+d_r}q^{n(m-d_r)}\frac{(q)_m(q)_l}{(q)_{m-d_r}}.
\]
By Lemma \ref{le:GF} and (\ref{Ad0}) the result reads
\begin{align*}
P_{\rm I}(r,n)  = y^n\frac{(q^{d_r}\mu)_n}{(q)_n}\Lambda(y)^{-1}
\sum_{l,m \ge 0}q^{n(m-d_r)}\frac{(\mu y)_m}{(y)_{m-d_r}}
G_{0,l}(d_1,\ldots, d_{r-1})\phi(l-m+d_r|l),
\end{align*}
where $\phi(l-m+d_r|l)$ is defined in (\ref{phidef}).
Thus (\ref{Pk2}) follows by taking the $l$ sum 
using (\ref{Grec}) with $(s,t)=(r,r-1)$
and (\ref{Gex1}).

Next we show (\ref{N2}) for the local density.
After substituting (\ref{Pk2}) into (\ref{occ}), the sum is taken by  
\begin{align*}
\sum_{n\ge 0} n y^n q^{n(m-d_r)}\frac{(q^{d_r}\mu)_n}{(q)_n} 
\frac{(y)_\infty}{(\mu y)_\infty}
&= \frac{(y)_\infty}{(\mu y)_\infty}
y\frac{\partial }{\partial y}
\frac{(q^m \mu y)_\infty}{(q^{m-d_r} y)_\infty}\\
& = 
\frac{(y)_{m-d_r}}{(\mu y)_m}
\Bigl(\rho  + \sum_{k=0}^{m-1}\frac{\mu y q^k}{1-\mu y q^k}
- \sum_{k=0}^{m-d_r-1}\frac{y q^k}{1-y q^k}\Bigr),
\end{align*}
where $\rho$ is the average density (\ref{rhof}).
The resulting expression for $\rho_{\rm I}(r)$ agrees with (\ref{N2})
except that $\rho$ comes with the coefficient 
$\sum_m G_{0,m}(d_1,\ldots, d_r)$. 
But this equals 1 by (\ref{G1}).

Let us proceed to the currents (\ref{Jf21}) and (\ref{Jf22}).
From (\ref{jr}) they are expressed as
\begin{align*}
J_{\rm I}(r)_{\pm} = \sum_{n \ge l \ge 1}
lw_\pm((0,l)|(d_r,n))P_{\rm I}(r,n).
\end{align*}
Substitution of (\ref{wp}), (\ref{wm}) and 
(\ref{Pk2}) leads to the calculation similar to (\ref{J0m}).
Its essential part is
\begin{align*}
\frac{(q^{m-d_r}y)_\infty}{(q^m\mu y)_{\infty}}
\sum_{n \ge l \ge 1}lw_\pm((0,l)|(d_r,n))
y^n q^{n(m-d_r)}\frac{(q^{d_r}\mu)_n}{(q)_n}
= 
\begin{cases}
\mu^{-1}h(q^m \mu y),\\
h(q^{m-d_r}y),
\end{cases}
\end{align*}
where $h(y)$ is the derivative of the $q$-digamma function 
defined in (\ref{hdef}).
Now we have
\begin{equation}\label{jgh}
J_{\rm I}(r)_\pm = 
\sum_m G_{0,m}(d_1,\ldots, d_r)\times
\begin{cases}
\mu^{-1}h(q^m \mu y), \\
h(q^{m-d_r}y).
\end{cases}
\end{equation}
Expand $h(q^m \mu y)$ and $h(q^{m-d_r}y)$
into the sum of $h(y)$ and the remainder terms by (\ref{hh}).
From (\ref{J0}) the contribution from the $h(y)$ is expressed as 
$\sum_m G_{0,m}(d_1,\ldots, d_r)J_\pm  = J_\pm $ by (\ref{G1}).
The remainder terms give the RHS of
(\ref{Jf21}) and (\ref{Jf22}).

\subsection{Proof of Theorem \ref{th:migi}}\label{ss:migi}
Theorem \ref{th:migi} 
is concerned with the region II (\ref{reg}).
The conditional probability 
$P_{\rm II}(r,n)$ for $r>s$ with the defects $(d_1,\ldots, d_s)$ 
is deduced from
the region I result $P_{\rm I}(r,n)$ (\ref{Pk2}) by setting  
$(d_1,\ldots, d_r) =(d_1,\ldots, d_s, \overset{r-s}{\overbrace{0,\ldots,0}})$.
It yields the expression
\begin{align}\label{oys}
P_{\rm II}(r,n) &=
P(n)
\sum_{i} q^{ni}\eta^{-1}_i
G_{0,i}(d_1,\ldots, d_s, \overset{r-s}{\overbrace{0,\ldots,0}}),
\end{align}
where $P(n)$ is the probability in the defect-free case (\ref{P0}) and 
the sum extends over $0 \le i \le d_1+\cdots + d_s$.
On the other hand 
the decomposition (\ref{Grec}) and Lemma \ref{le:GFred} lead to
\begin{align}
&G_{0,i}(d_1,\ldots, d_s, \overset{r-s}{\overbrace{0,\ldots,0}})
= \sum_m G_{0,m}(d_1,\ldots, d_s)
G_{m,i}(\overset{r-s}{\overbrace{0,\ldots,0}})
\nonumber\\
&= \sum_{i\le j \le m} G_{0,m}(d_1,\ldots, d_s)(-1)^{i+j}
q^{\frac{1}{2}i(i-1)+\frac{1}{2}j(j+1-2m)}
\binom{m}{j}_{\!\!q}\binom{j}{i}_{\!\!q}\eta_j^{r-s}.
\label{Gr}
\end{align}
Substituting 
$\binom{m}{j}_q = (-1)^j q^{mj -j(j-1)/2}\frac{(q^{-m})_j}{(q)_j}$
into this, we get (\ref{Pk1}).
The normalization $\sum_{n \ge 0}P_{\rm II}(r,n)=1$ 
is attributed to that of $P_{\rm I}(r,n)$.
In fact it can also be directly verified from (\ref{oys}) 
by noting $\sum_{n \ge 0}P(n)q^{ni} = \eta_i$ and (\ref{G1}).

To derive the density (\ref{N1}), we again set 
$(d_1,\ldots, d_r) =(d_1,\ldots, d_s, \overset{r-s}{\overbrace{0,\ldots,0}})$
in the region I result (\ref{N2}) and use (\ref{Gr}) to get
\begin{equation}\label{snn}
\begin{split}
\rho_{\rm II}(r)-\rho
&= \sum_{i,j,m}\!G_{0,m}(d_1,\ldots, d_s)
(-1)^{i+j}q^{\frac{1}{2}i(i-1)+\frac{1}{2}j(j+1-2m)}
\binom{m}{j}_{\!\!q}\binom{j}{i}_{\!\!q}\eta_j^{r-s}\\
& \quad \times
\sum_{k=0}^{i-1}\Bigl(
\frac{1}{1-\mu yq^k}-
\frac{1}{1-yq^k}\Bigr),
\end{split}
\end{equation}
where the sum extends over $i,j,m \in \Z_{\ge 0}$ such that
$0 \le i \le j \le m$ and 
$d_s \le m  \le d_1+\cdots + d_s$.
We may restrict $j$ to $1 \le j \le m$. 
Then the sums over $i$ and $k$ are taken by applying the identity
\begin{equation}\label{hgi}
\begin{split}
\sum_{i=0}^j \sum_{k=0}^{i-1}(-1)^i q^{\frac{1}{2}i(i-1)}
\binom{j}{i}_q\frac{1}{1-\zeta q^k}
&= \sum_{l \ge 0}
\sum_{i=0}^j(-1)^i q^{\frac{1}{2}i(i-1)}\binom{j}{i}_q 
\frac{1-q^{li}}{1-q^l}\zeta^l\\
&=\sum_{l \ge 0}\frac{(1)_j-(q^l)_j}{1-q^l}\zeta^l
= - \sum_{l \ge 0}\frac{(q)_{j+l-1}}{(q)_l}\zeta^l 
= - \frac{(q)_{j-1}}{(\zeta)_j},
\end{split}
\end{equation}
where the $q$-binomial theorem and (\ref{air}) are used.
After some simplification 
the result becomes (\ref{N1}).

To derive the current (\ref{Jf1}), we set 
$(d_1,\ldots, d_r) =(d_1,\ldots, d_s, \overset{r-s}{\overbrace{0,\ldots,0}})$
in the region I result (\ref{Jf21}), (\ref{Jf22}) and use (\ref{Gr}) to get
\begin{align*}
J_{\rm II}(r)_{\pm} - J_\pm
&= \sum_{i,j,m}\!G_{0,m}(d_1,\ldots, d_s)
(-1)^{i+j+1}q^{\frac{1}{2}i(i-1)+\frac{1}{2}j(j+1-2m)}
\binom{m}{j}_{\!\!q}\binom{j}{i}_{\!\!q}
\sum_{k=0}^{i-1}
\frac{y\eta_j^{r-s}q^k}{(1-\mu^{(1\pm1)/2}yq^k)^2},
\end{align*}
where the sum $\sum_{i,j,m}$ is taken in the same way as (\ref{snn}).
This time we apply the identity
\begin{align*}
\sum_{i=0}^j \sum_{k=0}^{i-1}(-1)^i q^{\frac{1}{2}i(i-1)}
\binom{j}{i}_q\frac{\zeta q^k}{(1-\zeta q^k)^2}
= - \frac{(q)_{j-1}}{(\zeta)_j}
\sum_{k=0}^{j-1}\frac{\zeta q^k}{1-\zeta q^k},
\end{align*}
which follows from (\ref{hgi}) by differentiation.
The result yields (\ref{Jf1}).

\subsection{Proof of Theorem \ref{th:hidari}}\label{ss:hidari}
Theorem \ref{th:hidari} 
is concerned with the region III (\ref{reg}).
From the definition of $X_{0,n}$ in (\ref{cie}) and 
$\langle 0 | \frac{(\mu \bb)_\infty}{(\bb)_\infty} = 
\langle 0 | \bk^n = \langle 0|$ by (\ref{qb}),
the conditional probability (\ref{Pl}) in the region III becomes
\begin{align*}
P_{\rm III}(r,n) = \lim_{L \rightarrow \infty}
\frac{y^ng_n\langle 0|A_0^{|r|} A_{d_1}\cdots A_{d_s} 
A_0^{L-|r|-s-1}|0\rangle}
{\langle 0|A_{d_1}\cdots A_{d_s} 
A_0^{L-s}|0\rangle} = y^n g_n \Lambda(y)^{-1}
\end{align*}
due to (\ref{Ad0}).
This is independent of $r$ and 
coincides with the probability $P(n)$ (\ref{P0}) for the defect-free case.  
Consequently the local density and current also reduce to the average density 
$\rho$ and $J_\pm$ in (\ref{J0}).

\subsection{Proof of Proposition \ref{pr:excess}}\label{ss:excess}

From (\ref{N1}), (\ref{etadef}) and (\ref{lin}), it follows that
\begin{align*}
\sum_{r>s} (\rho_{\rm II}(r)-\rho) = 
\sum_{j, m} G_{0,m}(d_1,\ldots, d_s) \frac{q^j(q^{-m})_j}{(1-q^j)(\mu y)_j},
\end{align*}
where the sum $\sum_{j, m}$ ranges over $m$ such that $G_{0,m}\neq 0$ 
implied by (\ref{Gz}) and $1 \le j \le m$.
On the other hand we have 
\begin{align*} 
\sum_{j=1}^m\frac{q^j(q^{-m})_j}{(1-q^j)(\mu y)_j} 
= -\sum_{k=0}^{m-1}\frac{1}{1-\mu yq^k}
\end{align*}
by replacing $(q,d,\mu)$ in (\ref{kzn}) 
with $(q^{-1}, m, q^{m-1}\mu y)$.
Thus the above expression is simplified to 
\begin{align}\label{hate}
\sum_{r>s} (\rho_{\rm II}(r)-\rho) = 
-\sum_m G_{0,m}(d_1,\ldots, d_s) 
\sum_{k=0}^{m-1}\frac{1}{1-\mu yq^k}.
\end{align}
Now we introduce
\begin{align*}
&K(r) = \sum_m G_{0,m}(d_1,\ldots, d_r) 
\sum_{k=0}^{m-1}\frac{1}{1-\mu y q^k},
\quad
{\tilde K}(r) = \sum_m G_{0,m}(d_1,\ldots, d_r) 
\sum_{k=0}^{m-d_r-1}\frac{1}{1-y q^k}.
\end{align*}
From (\ref{N2}) it is easy to see 
\begin{align*}
\rho_{\rm I}(r) - \rho = -d_r + K(r)-{\tilde K}(r).
\end{align*}
Thus Proposition \ref{pr:excess} follows from
\begin{align*}
&{\tilde K}(1)=0,\qquad 
K(r) = {\tilde K}(r+1) \quad (1 \le r <s),\\
&K(s) + \sum_{r>s}(\rho_{\rm II}(r)-\rho) = 0.
\end{align*}
The first equality is obvious from (\ref{Gex1}) and 
the last one has already been derived in (\ref{hate}).
So we are left to show $K(r) = {\tilde K}(r+1)$ only.
By using (\ref{Grec}) with $(t,s)=(r,r+1)$ and (\ref{Gex1}) we find
\begin{align*}
{\tilde K}(r+1) = \sum_m G_{0,m}(d_1,\ldots, d_r)\sum_{j=0}^m
\phi(m-j|m)\sum_{i=0}^{j-1}\frac{1}{1-y q^i},
\end{align*}
which tells that ${\tilde K}(r+1)$ 
is actually independent of $d_{r+1}$ as with $K(r)$.
Now the proof is finished by Lemma \ref{le:ma}.

\section{Discussion}\label{sec7}

In this paper we have investigated the 
stationary properties of the $U_q(A^{(1)}_2)$ ZRP \cite{KMMO} on a ring
whose stationary probability is not factorized but expressed in a nontrivial
matrix product form.
The second class particles are treated in the grand canonical ensemble
subject to the influence of 
the $d_1,\ldots, d_s$ first class particles 
fixed as defects at the sites $1,\ldots, s$.
The local density and current of the second class particles 
in the infinite volume limit are obtained in 
Theorem \ref{th:naka}, \ref{th:migi} and \ref{th:hidari}.
The density profile is shown to exhibit a peak and a 
valley at the boundaries of the defect cluster.
The currents are locally suppressed by the defects.
The detail is dependent on the inhomogeneity of the data 
$(d_1,\ldots, d_s)$.
Outside the defect cluster, 
its influence reaches longer distance when the average 
density of the second class particles is lower.
It was shown that 
the second class particles are decreased 
by the number of the defect particles in Proposition \ref{pr:excess}.

The effect of the defects is expressed through the function 
$G_{l,m}(d_1,\ldots, d_s)$, which is essentially a column monodromy matrix of 
the $U_q(A^{(1)}_1)$ ZRP containing the fugacity $y$ as the spectral parameter.
Compare the diagrams (\ref{Gm}) and (\ref{tdiag}).
We regard it as another reminiscent of the nested Bethe ansatz structure 
in addition to Remark \ref{re:hrm}.

It is natural to generalize the analysis in this paper to   
the grand canonical ensemble for the both classes of particles simultaneously.
Denoting the fugacity of the first class particles by $x$, 
one is led to the analogue of the grand canonical partition function 
in two variables as
\begin{align*}
Z_L(x,y) &= \mathrm{Tr}(V(x,y)^L)' 
= \mathrm{Tr}(\tilde{V}(x,y)^L)',\\
V(x,y) &= \sum _{m,n \ge 0}X_{m,n}x^m y^n
= \frac{(\mu \bb)_\infty}{(\bb)_\infty}
\frac{\Gamma(\mu x, \mu y)}{\Gamma(x,y)},
\quad
\tilde{V}(x,y) = \Xi(\mu x,\mu y,\mu)\Xi(x,y,1)^{-1},\\
\Gamma(x,y)&= (x \bc)_\infty(y \bk)_\infty,\qquad
\Xi(x,y,\mu) = (x\bc)_\infty(y\bk)_\infty(\mu\bb)_\infty,
\end{align*}
where we have used the 
cyclicity of the trace and the commutativity 
$[\Gamma(x,y), \Gamma(\mu x, \mu y)]=0$  
which follows from (\ref{cie}) by applying 
$(\mu)_{m+n}=(\mu)_m(q^m\mu)_n = (\mu)_n(q^n\mu)_m$.
See also \cite[Rem.3]{KO2}.
We leave the large $L$ asymptotics of $Z_L(x,y)$  
and related issues for a future study.

\appendix
\section{Proof of (\ref{ex1}) and Proposition \ref{pr:red}}\label{app:ls}
We invoke the interchangeability of limits and infinite sums 
in elementary calculus adapted in the following form.
\begin{lemma}\label{le:s1}
For the sequence 
$\{S_{L,l} \in \R_{\ge 0} \mid L,l \in \Z_{\ge 0}\}$,
the interchangeability 
$\lim_{L \rightarrow \infty}\sum_{l=0}^\infty S_{L,l}
= \sum_{l=0}^\infty \lim_{L \rightarrow \infty} S_{L,l}$
holds if the following (i) or (ii) is satisfied for 
$L$ sufficiently large:
\begin{enumerate}
\item
$\lim_{L \rightarrow \infty} S_{L,l} < \infty$ and 
$S_{L,l} \le T_l$ for some $T_l$ 
such that $\sum_{l=0}^\infty T_l < \infty$.

\item 
$S_{L+1,l} \le S_{L,l}$ and 
$\sum_{l=0}^\infty S_{L,l}<B$ for some $B$ independent of $L$.
\end{enumerate}
\end{lemma}
\begin{proof}
(i) Set $s_l=\lim_{L \rightarrow \infty}S_{L,l}$
By the definition $s_l \le T_l$.
For any $\varepsilon>0$ there is $M_\varepsilon$ such that 
$\sum_{l > M_\varepsilon}T_l < \varepsilon$, and there is also
$L_\varepsilon$ such that
$|S_{L,l}-s_l|<\frac{\varepsilon}{1+M_\varepsilon}\, (0 \le l \le M_\varepsilon)$
for $L > L_\varepsilon$.
Then one has
\begin{equation*}
\begin{split}
\left| \sum_{l=0}^\infty S_{L,l}-\sum_{l=0}^\infty s_l\right| 
&\le \left| \sum_{l=0}^{M_\varepsilon}(S_{L,l}-s_l) \right|
+\left| \sum_{l > M_\varepsilon}(S_{L,l}-s_l)\right|\\
&< \varepsilon + \sum_{l > M_\varepsilon}(S_{L,l}+s_l)
\le \varepsilon + 2 \sum_{l > M_\varepsilon}T_l< 3\varepsilon.
\end{split}
\end{equation*}
(ii) Without loss of generality we assume the condition holds for 
$L \ge 0$.
From the assumption $\lim_{L \rightarrow \infty} S_{L,l}$ exists.
Set $a_{i,l} = S_{i,l}-S_{i-1,l}\,(i\ge 1)$ and $a_{0,l} = S_{0,l}$. 
Then $|a_{i,l}| = (-1)^{\theta(i \ge 1)}a_{i,l}$.
Since $\sum_{i=0}^L\sum_{l=0}^\infty |a_{i,l}| 
= 2\sum_{l=0}^\infty S_{0,l}-\sum_{l=0}^\infty S_{L,l}<2B$,
The sum of $a_{i,l}$ over $(i,l) \in \Z_{\ge 0}^2$ is absolutely convergent
and can be taken in any order.
Thus we get
$\lim_{L \rightarrow \infty} \sum_{i=0}^L\sum_{l=0}^\infty a_{i,l}
= \lim_{m \rightarrow \infty} \sum_{l=0}^m \sum_{i=0}^\infty a_{i,l}$.
This means $\lim_{L \rightarrow \infty}\sum_{l=0}^\infty S_{L,l}
= \sum_{l=0}^\infty \lim_{L \rightarrow \infty} S_{L,l}$.
\end{proof}

{\em Proof of (\ref{ex1})}.
For $n \ge 1$, set
$S_{L,l} = q^{ln}\eta^L_l$ and $T_l = q^{ln}$.
Then the condition (i) in Lemma \ref{le:s1} is satisfied.
For $n=0$, set $S_{L,l} = \eta^L_l-\eta^L_\infty$ which 
satisfies $S_{L+1,l} \le S_{L,l}$ for sufficiently large $L$.
Set further $B_0 = \sum_{l=0}^\infty S_{1,l} 
= \eta_\infty \sum_{j\ge 1}\frac{g_j y^j}{1-q^j}$ which is finite.
From (\ref{lin}), we have
$\sum_{l=0}^\infty S_{L,l} = 1-\eta_\infty^L + 
\sum_{l=1}^\infty (\eta_l-\eta_\infty)(\eta^{L-1}_l + \eta^{L-2}_l\eta_\infty
+ \cdots + \eta^{L-1}_\infty) < 1 + B_0L\eta_1^{L-1}$.
Thus the condition (ii) in Lemma \ref{le:s1} is satisfied by taking  
$B = 1+B_0\max \{ L\eta_1^{L-1} \mid L \in \Z_{\ge 1}\}< \infty$.
\qed

Let $Q$ be the operators of the form
\begin{align}\label{Qf}
Q = A_{m_1} \cdots A_{m_\kappa}\;\;\text{or}\;\;
Q = A_{m_1} \cdots A_{m_{t-1}} X_{e,h} A_{m_{t+1}}\cdots A_{m_\kappa}
\end{align}
for some $m_1,\ldots, m_\kappa, e, h \ge 0$ and $1 \le t \le \kappa$.
They are contained in 
the matrix product formula of the probabilities
(\ref{Pmiddle})--(\ref{Pleft}) in the form
$\mathrm{Tr}(QA_0^L)'$ with $L$ large, where $(\cdots)'$ 
is explained after (\ref{Bd}).

\begin{lemma}\label{le:zer}
For the operator $Q$ of the form (\ref{Qf}), 
the following limits exist and are equal:
\begin{align}\label{rz}
\lim_{L \rightarrow \infty}
\Lambda(y)^{-L}
\mathrm{Tr}\bigl(Q A_0^L \bigr)'
=
\lim_{L \rightarrow \infty}
\Lambda(y)^{-L}
{\langle 0|Q A_0^L|0\rangle}.
\end{align}

\end{lemma}
\begin{proof} 
As for the RHS we have
\begin{align*}
&
\Lambda(y)^{-L}\sum_{m\ge 0}
\frac{\langle 0 |Q |m \rangle }{(q)_m}
\langle m|\Bigl(\frac{(\mu \bb)_\infty}{(\bb)_\infty}
\frac{(\mu y \bk)_\infty}{(y \bk)_\infty}\Bigr)^L|0\rangle
\\
&=\Lambda(y)^{-L}\sum_{m\ge 0}
\frac{\langle 0 |Q |m \rangle }{(q)_m}\!\!\!
\sum_{l_1+\cdots + l_L = m}\!\!
g_{l_1}\cdots g_{l_L}
\langle m|
\bb^{l_1}\frac{(\mu y \bk)_\infty}{(y \bk)_\infty}
\cdots 
\bb^{l_L}\frac{(\mu y \bk)_\infty}{(y \bk)_\infty}
|0\rangle
\\
&=
\sum_{m\ge 0}
\langle 0 |Q|m \rangle
\sum_{l_1+\cdots + l_L = m}
g_{l_1}\cdots g_{l_L}
\eta_{l_2+\cdots+ l_L}\cdots 
\eta_{l_{L-1}+l_L}\eta_{l_L}
= \sum_{m\ge 0}
\langle 0 |Q|m \rangle F_{m,L}(y),
\end{align*}
where $l_1,\ldots, l_L$ are summed over $\Z_{\ge 0}$ with the 
specified condition.
We have used the expansion (\ref{air}), 
the quantity $\eta_m$ (\ref{etadef}) and the definition of the 
function $F_{m,L}(y)$ (\ref{Fdef}).
By the assumption on $Q$,   
$\langle 0 |Q|m \rangle\neq 0$ holds only for 
{\em finitely many} $m \in \Z_{\ge 0}$.
Thus we get
\begin{align*}
\text{RHS of (\ref{rz})} &= 
\sum_{m\ge 0}
\langle 0 |Q|m\rangle \lim_{L \rightarrow \infty} F_{m,L}(y)\\
&= \sum_{m\ge 0}
\frac{y^{-m}(\mu y)_m\langle 0 |Q|m\rangle}{(q)_m}
\lim_{L \rightarrow \infty}
\sum_{0 \le j \le m}(-1)^jq^{\frac{1}{2}j(j+1-2m)}
\binom{m}{j}_q\eta^L_j\\
&=  \sum_{m\ge 0}
\frac{y^{-m}(\mu y)_m\langle 0 |Q|m\rangle}{(q)_m},
\end{align*}
where the second equality is
due to (\ref{gfq}) and Lemma \ref{le:GFred} with $(i,r)=(0,L)$.
This result shows that the RHS is finite. 

As for the LHS 
we expand $Q$ by 
$\frac{(\mu \bb)_\infty}{(\bb)_\infty}
= \sum_{j\ge 0} g_j \bb^j$ (\ref{air}) and apply the 
commutation relation (\ref{akn}).
The resulting terms contain
$\bk$ as an overall factor and are grouped into two types as
\begin{align}\label{2t}
\Lambda(y)^{-L}
\mathrm{Tr}\bigl(Q A_0^L \bigr)'&=
\Lambda(y)^{-L}\sum_{m\ge l \ge 0}\left(
 \frac{\langle l |Q_0 |m \rangle }{(q)_l(q)_m}\langle m|A_0^L|l\rangle
+ \frac{\langle l |Q_1 |m \rangle }{(q)_l(q)_m}
\langle m|\bigl(A_0^L-(A_0|_{y=0})^L\bigr)|l\rangle\right).
\end{align}
Here $Q_0$ and $Q_1$ are {\em finite} linear combinations 
of the operators of the form
\begin{align*}
Q_0 \,:&\;\;
\prod_{i=1}^\kappa\frac{(q^{\alpha_i}\mu y \bk)_\infty}
{(q^{\beta_i} y \bk)_\infty}\,\bk^n \bc^{\lambda},
\qquad 
Q_1 \,:\;\;  
\prod_{i=1}^\kappa\frac{(q^{\alpha_i}\mu y \bk)_\infty}
{(q^{\beta_i} y \bk)_\infty}\,\bc^{\lambda}.
\end{align*}
with various 
$\alpha_i, \beta_i \in \Z,\, \lambda \in \Z_{\ge 0},\,
n \in \Z_{\ge 1}$, and $\kappa$ is the one in (\ref{Qf}).
Let us pick one such term for each type:
\begin{align*}
W_0 &= \Lambda(y)^{-L}\sum_{m\ge l \ge 0}
\frac{1}{(q)_l(q)_m}
\langle l |\prod_{i=1}^\kappa\frac{(q^{\alpha_i}\mu y \bk)_\infty}
{(q^{\beta_i} y \bk)_\infty}\,\bk^n \bc^{\lambda}|m \rangle
\langle m| A_0^L|l\rangle,\\
W_1&= \Lambda(y)^{-L}\sum_{m\ge l \ge 0}
\frac{1}{(q)_l(q)_m}
\langle l |\prod_{i=1}^\kappa\frac{(q^{\alpha_i}\mu y \bk)_\infty}
{(q^{\beta_i} y \bk)_\infty}\,\bc^{\lambda}|m \rangle
\langle m|\bigl(A_0^L-(A_0|_{y=0})^L\bigr)|l\rangle.
\end{align*}
These are actually single sums over $l$ 
since $m$ is frozen to $m=l+\lambda$.
Denote them by
$W_0 =\sum_{l \ge 0}U_{L,l}$ and 
$W_1 =\sum_{l \ge 0}V_{L,l}$.
A direct calculation gives
\begin{equation}\label{uvdef}
\begin{split}
&U_{L,l} = q^{ln}\Gamma_l u_{L,l},
\quad V_{L,l} = \Gamma_l v_{L,l},
\qquad
\Gamma_l = \frac{(q)_{l+\lambda}}{(q)_l}
\prod_{i=1}^\kappa\frac{(q^{l+\alpha_i}\mu y)_\infty}
{(q^{l+\beta_i} y)_\infty},
\\
&u_{L,l}= \sum_{l_1+\cdots + l_L = \lambda}\!\!
g_{l_1}\cdots g_{l_L}
\eta_{l+l_2+\cdots +l_L}\cdots \eta_{l+l_L}\eta_l,
\\
&v_{L,l}= \sum_{l_1+\cdots + l_L = \lambda}\!\!
g_{l_1}\cdots g_{l_L}\bigl(
\eta_{l+l_2+\cdots +l_L}\cdots 
\eta_{l+l_L}\eta_l -\eta_\infty^L\bigr).
\end{split}
\end{equation}
We also introduce 
\begin{align*}
\overline{V}_L &=
\Lambda(y)^{-L}\sum_{m\ge 0}
\frac{1}{(q)_m}
\langle 0 |\prod_{i=1}^\kappa\frac{(q^{\alpha_i}\mu y \bk)_\infty}
{(q^{\beta_i} y \bk)_\infty}\,\bc^{\lambda}|m \rangle
\langle m|A_0^L|0\rangle
= \Gamma_0 \overline{v}_L,\\
\overline{v}_L &= 
\sum_{l_1+\cdots + l_L = \lambda}\!\!
g_{l_1}\cdots g_{l_L}
\eta_{l_2+\cdots +l_L}\cdots 
\eta_{l_{L-1}+l_L}\eta_{l_L}\; (= F_{\lambda,L}(y) \,\;\text{in }\, (\ref{Fdef})).
\end{align*}
Since the $l$'s in $U_{L,l}$ and $V_{L,l}$ have the same meaning as 
those in (\ref{2t}), 
the proof is reduced to showing that the only $l=0$ term survives 
in the large $L$ limit.
\begin{align}\label{uu}
\lim_{L \rightarrow \infty}\sum_{l \ge 0}U_{L,l}
= \lim_{L \rightarrow \infty}U_{L,0},\qquad
\lim_{L \rightarrow \infty}\sum_{l \ge 0}V_{L,l}
= \lim_{L \rightarrow \infty}\overline{V}_L.
\end{align} 
From (\ref{lin}) the equalities
$\lim_{L \rightarrow \infty}u_{L,l} 
= \delta_{l,0}\lim_{L \rightarrow \infty}u_{L,0}$ and 
$\lim_{L \rightarrow \infty}v_{L,l} 
= \delta_{l,0}\lim_{L \rightarrow \infty}v_{L,0}
= \delta_{l,0}\lim_{L \rightarrow \infty}\overline{v}_{L}$
are valid obviously.
Therefore (\ref{uu}) follows from the next Lemma \ref{le:krt}.
\end{proof}

\begin{lemma}\label{le:krt}
The following interchangeability holds for the series involving 
$U_{L,l}, V_{L,l}$ in (\ref{uvdef}):   
\begin{align*}
&\mathrm{(a)}\;\lim_{L \rightarrow \infty} \sum_{l \ge 0}U_{L,l}
= \sum_{l \ge 0}\lim_{L \rightarrow \infty}U_{L,l},\qquad 
\mathrm{(b)}\;\lim_{L \rightarrow \infty} \sum_{l \ge 0}V_{L,l}
= \sum_{l \ge 0}\lim_{L \rightarrow \infty}V_{L,l}.
\end{align*}
\end{lemma}

\begin{proof}
(a).
From (\ref{bx}) and  (\ref{lin}) we have
\begin{align*}
u_{L,l} &\le F_{\lambda, L}(y) = 
y^{-\lambda}\frac{(\mu y)_\lambda}{(q)_\lambda}
\sum_{0 \le j \le \lambda}
(-1)^j q^{\frac{1}{2}j(j+1-2\lambda)}\binom{\lambda}{j}_q\eta_j^L\\
&< y^{-\lambda}\frac{(\mu y)_\lambda}{(q)_\lambda}\sum_{0 \le j \le \lambda}
q^{\frac{1}{2}j(j+1-2\lambda)}\binom{\lambda}{j}_q = 
y^{-\lambda}\frac{(\mu y)_\lambda(-q^{1-\lambda})_\lambda}{(q)_\lambda}.
\end{align*}
Thus by setting  $S_{L,l} = U_{L,l}$ and 
$T_l = y^{-\lambda}q^{nl}\Gamma_l
\frac{(\mu y)_\lambda(-q^{1-\lambda})_\lambda}{(q)_\lambda}$,
the condition (i) in Lemma \ref{le:s1} is satisfied, hence (a) is valid.
In particular, $\sum_{l=0}^\infty T_l < \infty$ 
because the series 
$\sum_{l=0}^\infty w^l \Gamma_l$ is convergent for $|w| < 1$
due to $\lim_{l\rightarrow \infty}\Gamma_l = 1$.

(b) We prove that Lemma \ref{le:s1} (ii) applies to
$S_{L,l} = V_{L,l}$.
Thus we are to verify 
(b)$_1: \, v_{L+1,l}\le v_{L,l}$ for sufficiently large $L$,
and (b)$_2: \,\sum_{l \ge 0}V_{L,l} < B$ for some $B$.
To show (b)$_1$, note the large $L$ asymptotic behavior
\begin{align*}
&\sum_{l_1+\cdots + l_L = \lambda}\!\!
g_{l_1}\cdots g_{l_L}
\eta_{l+l_2+\cdots +l_L}\cdots \eta_l
= \eta_l^LF_{\lambda,L}(q^ly)\\
&=
\frac{(q^l\mu y)_\lambda(q^l y)^{-\lambda}}{(q)_\lambda}
\sum_{0 \le j \le \lambda}(-1)^j q^{\frac{1}{2}j(j+1-2\lambda)}
\binom{\lambda}{j}_q \eta_{j+l}^L 
\simeq 
\frac{(q^l\mu y)_\lambda(q^l y)^{-\lambda}}{(q)_\lambda}
\eta_l^L(1+ O(\eta_{l+1}^L/\eta_l^L)),
\end{align*}
where the first and the second equalities 
follow from (\ref{Fdef}) and (\ref{bx})
with $\eta_{l+m}= \eta_l\eta_m(q^ly)$, 
and the last estimation is due to (\ref{lin}).
As for the second term in $v_{L,l}$ containing $-\eta^L_\infty$, 
one can estimate the coefficient of it as
$\binom{L-1+\lambda}{\lambda}c_1 
\le \sum_{l_1+\cdots + l_L = \lambda}
g_{l_1}\cdots g_{l_L} 
\le \binom{L}{\lambda}c_2$,
where $c_1, c_2$ are $L$-independent quantities
$c_1= \min\{g_{l_1}\cdots g_{l_\lambda}\mid l_i \in \Z_{\ge 0},
l_1+\cdots + l_\lambda = \lambda\}$
and $c_2 = \sum_{l_1+\cdots + l_\lambda = \lambda}
g_{l_1}\cdots g_{l_\lambda}$.
From these results and the Stirling formula
one finds that $v_{L,l}$ has the large $L$ asymptotic behavior
$v_{L,l} \simeq c_3 \eta_l^L - c_4 L^\lambda \eta_\infty^L$
with $L$-independent $c_3, c_4 >0$.
Thus the inequality (b)$_1$ for sufficiently large $L$ follows from (\ref{lin}).
To show (b)$_2$, we use
\begin{align*}
v_{L,\lambda} &\le 
\eta^L_\infty\sum_{l_1+\cdots + l_L = \lambda}
g_{l_1}\cdots g_{l_L}\Bigl(
\frac{(q^l \mu y)_\infty^L}{(q^l y)_\infty^L}-1\Bigr)
\\
&\le 
\eta^L_\infty\sum_{l_1+\cdots + l_L = \lambda}
g_{l_1}\cdots g_{l_L}\Bigl(
\frac{(q^l \mu y)_\infty}{(q^l y)_\infty}-1\Bigr)L 
\frac{(\mu y)_\infty^{L-1}}{(y)_\infty^{L-1}}
\\
&= L\eta_\infty\sum_{l_1+\cdots + l_L = \lambda}
g_{l_1}\cdots g_{l_L}
\sum_{j \ge 1}g_j (q^l y)^j.
\end{align*}
From (\ref{uvdef}), there exists $l_0 \ge 0$ such that 
$\Gamma_l >0$ holds for all $l \ge l_0$\footnote{
Such $l_0$ is not unique but the non-uniqueness does not spoil the 
subsequent argument.}.
Then the above estimation leads to  
$\sum_{l \ge 0}V_{L,l}
= C+ \sum_{l \ge l_0}\Gamma_l v_{L,l}
\le
C+ L\eta_\infty  \sum_{l_1+\cdots + l_L = \lambda}g_{l_1}\cdots g_{l_L}
\sum_{j \ge 1}f_jg_j y^j$,
where 
$C= \sum_{0 \le l <l_0}V_{L,l}$ is a finite constant and we have set 
$f_j = \sum_{l \ge l_0}\Gamma_l q^{jl}$.
This series $f_j$ is convergent because of 
$\lim_{l\rightarrow \infty}\Gamma_l = 1$ as noted in (a).
The series 
$\sum_{j \ge 1}f_jg_j y^j$ in $y$ is also convergent
due to $0 < y < 1$ and 
$\lim_{j \rightarrow \infty}\frac{f_jg_j}{f_{j+1}g_{j+1}}
= q^{-l_0}\lim_{j \rightarrow \infty}\frac{1-q^{j+1}}{1-q^j \mu} =
q^{-l_0}>1$.
Thus $\sum_{l \ge 0}V_{L,l}$ is convergent for any fixed $L$.
On the other hand (b)$_1$ tells that $\sum_{l \ge 0}V_{L,l}$
is monotonously decreasing with respect to $L$ for sufficiently large $L$.
This proves (b)$_2$.
\end{proof}

\section*{Acknowledgments}
The authors thank Masato Okado for discussion and 
Tomohiro Sasamoto for communication.
This work is supported by 
Grants-in-Aid for Scientific Research 
No.~15K13429 from JSPS.

\end{document}